\documentclass[12pt]{article}

\usepackage{natbib}

\def\diag{\hbox{diag}}

\def\diag{\hbox{diag}}

\def\tr{\mathrm{tr}}

\def\vect{\hbox{vec}}

\usepackage{amsmath}
\usepackage{verbatim,color,amssymb}
\usepackage{amsmath}	
\usepackage{pdflscape}
\usepackage{rotating}

\usepackage{amsthm}					
\usepackage{natbib}
\usepackage{multirow}
\usepackage{setspace}
\usepackage[mathscr]{euscript}
\usepackage{fancyhdr}
\usepackage{enumitem}
\usepackage{graphicx}
\usepackage{subfigure}
\usepackage{lineno}
\usepackage{array,booktabs}
\usepackage{bm, bbold, bbm}
\usepackage{url}
\newcommand{\Frob}[1]{\Vert#1\Vert_{\mathrm{F}}}
\newcommand{\op}[1]{\Vert#1\Vert_{\mathrm{op}}}

\usepackage{lscape}

\usepackage{subfig}
\usepackage{setspace}
\usepackage{tikz}
\usetikzlibrary{arrows}
\usepackage{multirow}

\newtheorem*{Proof*}{Proof}

\newtheorem{proposition}{Proposition}

\newtheorem{theorem}{Theorem}

\newtheorem{lemma}[theorem]{Lemma}

\def\G{{\cal G}}

\def\diag{\hbox{diag}}

\def\diag{\hbox{diag}}

\def\vect{\hbox{vec}}

\def\IG{\hbox{Inv-Ga}}

\def\Normal{\hbox{Normal}}

\def\bse{\begin{eqnarray*}}
\def\ese{\end{eqnarray*}}
\def\be{\begin{eqnarray}}
\def\ee{\end{eqnarray}}
\def\bqe{\begin{equation}}
\def\eq{\end{equation}}

\usepackage[algoruled,boxed,lined]{algorithm2e}

\def\trans{^{\rm T}}

\def\b1e{{\mathbf e}}
\newcommand{\beqa}{\begin{eqnarray*}}
\newcommand{\eeqa}{\end{eqnarray*}}
\newcommand{\beqn}{\begin{eqnarray}}
\newcommand{\eeqn}{\end{eqnarray}}
\newcommand{\baa}{\begin{array}}
\newcommand{\eaa}{\end{array}}
\newcommand{\bcc}{\begin{center}}
\newcommand{\ecc}{\end{center}}
\newcommand{\btab}{\begin{tabular}}
\newcommand{\etab}{\end{tabular}}

\newcommand{\lb}{\label}

\newcommand{\iy}{\infty}

\newcommand{\mcC}{\mathbb{C}}
\newcommand{\mcD}{\mathcal{D}}

\newcommand{\mcS}{\mathcal{S}}

\newcommand{\Var}{{\rm Var}}

\newcommand{\ep}{\epsilon}

\newcommand{\bSigma}{\boldsymbol{\Sigma}}

\newcommand{\bLambda}{\boldsymbol{\Lambda}}
\newcommand{\blambda}{\boldsymbol{\lambda}}

\newcommand{\bDelta}{\boldsymbol{\Delta}}

\newcommand{\bxi}{\boldsymbol{\xi}}
\newcommand{\bGamma}{\boldsymbol{\Gamma}}

\newcommand{\bkappa}{\boldsymbol{\kappa}}

\newcommand{\bUpsilon}{\boldsymbol{\Upsilon}}
\newcommand{\bmu}{\boldsymbol{\mu}}
\newcommand{\bzeta}{\boldsymbol{\zeta}}

\newcommand{\bOmega}{\boldsymbol{\Omega}}

\newcommand{\bK}{\boldsymbol{K}}
\newcommand{\bA}{\boldsymbol{A}}
\newcommand{\bB}{\boldsymbol{B}}

\newcommand{\bC}{\boldsymbol{C}}

\newcommand{\bD}{\boldsymbol{D}}

\newcommand{\bV}{\boldsymbol{V}}

\newcommand{\bG}{\boldsymbol{G}}
\newcommand{\bH}{\boldsymbol{H}}

\newcommand{\bM}{\boldsymbol{M}}

\newcommand{\bL}{\boldsymbol{L}}
\newcommand{\bI}{\boldsymbol{I}}

\newcommand{\bS}{\boldsymbol{S}}

\newcommand{\bE}{\boldsymbol{E}}
\newcommand{\bF}{\boldsymbol{F}}

\newcommand{\bP}{\boldsymbol{P}}

\newcommand{\bQ}{\boldsymbol{Q}}
\newcommand{\bR}{\boldsymbol{R}}

\newcommand{\bW}{\boldsymbol{W}}

\newcommand{\bU}{\boldsymbol{U}}

\newcommand{\bX}{\boldsymbol{X}}

\newcommand{\bY}{\boldsymbol{Y}}
\newcommand{\bZ}{\boldsymbol{Z}}

\newcommand{\bzero}{\boldsymbol{0}}

\newcommand{\CC}{{\mathbb C}}

\newcommand{\EE}{{\mathbb E}}

\newcommand{\RR}{{\mathbb R}}

\numberwithin{equation}{section}
\theoremstyle{plain}

\newcommand{\Appendix}
{
\def\thesection{Appendix~\Alph{section}}
\def\thesubsection{A.\arabic{subsection}}
}

\title{Relational Graph in Vector Autoregression: A Case Study on the Effect of the Great Recession on Connectivity of Economic Indicators}

\author{Arkaprava Roy\footnote{Department of Biostatistics, University of Florida, {arkaprava.roy@ufl.edu}}, Anindya Roy\footnote{Department of Mathematics and Statistics, University of Maryland Baltimore County, {anindya@umbc.edu}} and Subhashis Ghosal\footnote{Department of Statistics, North Carolina State University, {sghosal@ncsu.edu}}}

\begin{document}
\maketitle

\begin{abstract}

Under a high-dimensional vector autoregressive (VAR) model, we propose a way of efficiently estimating both the stationary graph structure between the nodal time series and their temporal dynamics. The framework is then used to make inferences on the change in interdependencies between several economic indicators due to the impact of the Great Recession, the financial crisis that lasted from 2007 through 2009. There are several key advantages of the proposed framework; (1) it develops a reparametrized VAR likelihood that can be used in general high-dimensional VAR problems, (2) it strictly maintains causality of the estimated process, making inference on stationary features more meaningful and (3) it is computationally efficient due to the reduced rank structure of the parameterization. We apply the methodology to the seasonally adjusted quarterly economic indicators available in the FRED-QD database of the Federal Reserve. The analysis essentially confirms much of the prevailing knowledge about the impact of the Great Recession on different economic indicators. At the same time, it provides deeper insight into the nature and extent of the impact on the interplay of the different indicators. We also contribute to the theory of Bayesian VAR by showing the consistency of the posterior under sparse priors for the parameters of the reduced rank formulation of the VAR process. 
\end{abstract}

\section{Introduction}
Large dynamic systems like the US economy go through 
periods of upturns and downturns. However, extreme events such as the Great Recession and the COVID-19 epidemic cause upheavals beyond the reach of normal cycles. The Great Recession disrupted traditional and well-understood relationships among U.S. economic indicators, leading to unconventional dynamics in and among economic, labor, and financial markets during and after the event. Despite several government interventions, the sheer magnitude of the recession made recovery extremely hard. It was posited that many of the market indicators went through permanent structural changes as the economy emerged from the recession. 
The recession not only had a disparate impact on different market sectors but also restructured their inter-dependencies, altering how different sectors of the broader US economy interacted. In the aftermath, with accruing data, researchers were naturally interested in ascertaining the nature of the changes, not only in the individual indicators but also in their interconnections. We investigate the hypothesis of pre- and post-interconnections being the same by studying the partial correlation graph among the economic indicators for different market sectors.

In periods between successive extreme events, the economic system usually follows a more predictable set of dynamics that can be modeled as part of a stationary process (after suitable transformations). To analyze an event's impact, one can use independent models to represent the stationary dynamics before and after the event and compare the parameters of the two models within a Bayesian framework. A change point model is often used to determine structural changes in the time series dynamics; \cite{bai2023multiple}. However, for a significant and catastrophic event like the Great Recession or the COVID-19 pandemic, one typically requires multiple change points to adequately capture the effect of all the shocks to the system. The epochs for the multiple change points are generally unknown. During the Great Recession, the regulatory agencies took several actions whose dates and times are known. However, several additional shocks to the system contributed to the downturn. Accurate testing of hypotheses on structural changes could be challenging without knowing the number and location of change points. Moreover, the period during the event will have very different dynamics from the stable dynamics to which the driving indicators are expected to return after the event. 
The advantage of a Bayesian paradigm is that we can effectively draw inferences on the impact by studying the posterior of the differences in model parameters from the before and after periods, thereby identifying and quantifying the changes induced by the disruptive event. This approach allows one to quantify the shift in the economic system and glean insight into how the event has changed the underlying dynamics of the system.

 We use the FRED-QD data \citep{mccracken2020fred} for our investigation. The data comprise 248 different economic time series observed at a quarterly frequency. The FRED-QD time series has often been modeled using vector autoregression or its variants \cite{mccracken2020fred, chan2023large,  samadi2024reduced}. We use the vector autoregression to model the data and investigate the change in the interconnections among the different component series within the VAR framework. Specifically, we study the change in the stationary precision matrix due to the great recession period (December 2007–-June 2009) in a graphical model set up. The graphical model is ideal for studying the interconnection among the economic indicators and any potential changes that may have occurred post-recession. While a change point analysis in vector autoregression can test hypotheses of structural changes, for events spanning a long time window including multiple time points, such as the Great Recession, the specification of a point null hypothesis may be tenuous. Moreover, assumptions such as stability and stationarity are important properties to analyze under a VAR framework \citep{wang2022high, zheng2024interpretable}.
 Thus, we focus on investigating the difference in model parameters between the pre- and the post-recession periods, where the assumption of stationarity (after transformation) may be more meaningful. Given the duration of the crisis, the periods before and after the recession can be reasonably assumed to be independent for VAR modeling purposes. We assume that the underlying dynamics of the economy have recovered to a stable state in the period following the recession, possibly to a different stable state compared to before the financial crisis. We compare independent causal VAR models from the two periods in a Bayesian framework and make inference on the changes in the stationary graphs given by the corresponding stationary precision matrices. The joint analysis of the before and after states is somewhat reminiscent of usual difference-in-difference analysis \cite{angrist2009mostly} done to evaluate the impact of an extreme shock to a dynamical system.

Graphical models are popular models that encode scientific linkages between variables of interest through a conditional independence structure and provide a parsimonious representation of the joint probability distribution of the variables, mainly when the number of variables is large. Thus, learning the graph structure underlying the joint probability distribution of an extensive collection of variables is an important problem, and it has a long history. When the joint distribution is Gaussian, learning graphical structure can be done by estimating the precision matrix leading to the popular Gaussian Graphical Model (GGM); even in the non-Gaussian case one can think of encoding the graph using the partial correlation structure and learn the Partial Correlation Graphical Model (PCGM) by estimating the precision matrix. However, estimating the precision matrix efficiently from the sample can be challenging when the sample has temporal dependence. If estimating the graphical structure is the main objective, the temporal dependence can be treated as a nuisance feature and ignored in graph estimation. When the temporal features are also important to learn, estimating the graph structure and the temporal dependence simultaneously can be considerably more complex.

There are several formulations of graphical models for time series. A time series graph could be one with the nodes representing the entire coordinate processes. In particular, if the coordinate processes are jointly stationary, this leads to a  `stationary graphical' structure.  Conditional independence in such a graph can be expressed equivalently in terms of the absence of partial spectral coherence between two nodal series; see \cite{Dahlhaus2000}. Estimation of such stationary graphs has been investigated by \cite{jung2015graphical},  \cite{basu2015network}, \cite{ma2016joint}, \cite{fiecas2019spectral} and \cite{Basu2023}.  

A weaker form of conditional coding for stationary multivariate time series is a `contemporaneous stationary graphical' structure, where the graphical structure is encoded in the marginal precision matrix. In a contemporaneous stationary graph, the nodes are the coordinate variables of the vector time series. 
If the time series is Gaussian, this leads to contemporaneous conditional independence among the coordinates of the vector time series. This is the GGM structure, which has the appealing property that the conditional independence of a pair of variables given the others is equivalent to the zero value of the corresponding off-diagonal entry of the precision matrix. Hence, graph learning can be achieved via the estimation of precision matrices. A rich literature on estimating the graphical structure can be found in classical books like \cite{Lauritzen1996, koller2009probabilistic}; see also \cite{Maathuis2019} and references therein for more recent results.  When the distribution is specified only up to the second moment, common in the study of second-order stationary time series, one could learn the contemporaneous partial correlation graphical structure by estimating the stationary precision matrix. In this paper, we use the contemporaneous precision matrix of a stationary time series to estimate the graph along the line of \cite{qiu2016joint,p:zha-17}. We use a vector autoregression (VAR) model to model the temporal dynamics.

To study the graphical relations and changes within the macroeconomic series, this paper combines two popular models: the PCGM for graph estimation and the VAR model for estimating the temporal dynamics in vector time series. 
In the context of our specific investigation, we require that the processes over the study periods are stable, as encoded by the causality of the VAR process. Also, given the high dimensionality of the problem, dimension reduction in the form of reduced rank VAR with sparse precision matrix is desired. Modeling the partial correlation graph and the temporal dependence with the VAR structure simultaneously is challenging, particularly under constraints such as reduced rank and causality on the VAR model and sparsity on the graphical model. This is achieved in this paper via a novel parameterization of the VAR process. The main contribution of the paper is a methodology that allows for meeting the following two challenging objectives simultaneously:
\begin{enumerate}
    \item [(i)] Estimating the contemporaneous stationary graph structure under a sparsity constraint.
    \item  [(ii)] Estimating VAR processes with a reduced rank structure under causality constraints.
\end{enumerate}
A novelty of the proposed approach is that while developing the methodology for performing the above two tasks, we can
\begin{enumerate}
    \item [(a)] develop a recursive computation scheme for computing the reduced-rank VAR likelihood through low-rank updates; 
    \item [(b)] establish posterior concentration under priors based on the new parameterization.
\end{enumerate}

\section{FRED-QD data and Exploratory Analysis}
\label{sec: EDA}
We study the change in the economic output due to shocks from the Great Recession period (December 2007-June 2009) using the FRED-QD data \citep{mccracken2020fred}. This dataset contains a total of 248 quarterly frequency series that are further divided into 14 disjoint groups, namely 1) NIPA (23 variables); 2) Industrial Production (16 variables);
3) Employment and Unemployment (50 variables); 4) Housing (13 variables); 5) Inventories, Orders, and Sales (9 variables); 6) Prices (48 variables); 7) Earnings
and Productivity (14 variables); 8) Interest Rates (19 variables); 9) Money and Credit (14 variables); 10) Household Balance Sheets (9 variables); 11) Exchange
Rates (4 variables); 12) Other (2 variables); 13) Stock Markets (3 variables); and 14) Non-Household Balance Sheets (13 variables). The numbers in parentheses indicate the variable counts after excluding those with missingness. The group `other' only contains two variables and is thus excluded from the analysis. The Unemployment and the Prices groups have close to 50 series. Since the number of time points is limited and the groups can be naturally split into further economic subgroups, we decided to split the Unemployment and the Price groups into two parts each. The parts of the Unemployment group comprise all series related to Employment volume (22 variables) (Employment) and those related to unemployment and workforce participation (Unemployment and Participation). In contrast, the parts of the Price group comprise series related to Personal Consumption and Expenditure, with 20 variables (Prices(PCE)), and the prices series related to Consumer Price Indices, with the remaining 28 series (Prices(CPI)). Thus, in our analysis, we have 15 different groups. 
While estimating the post-recession parameter, we extract the data for 43 time points from July 2009 to January 2020, avoiding the start of the COVID-19 pandemic. We extract the data for 43 time points preceding December 2007 for the pre-recession parameter estimation to have comparable lengths. As a pre-processing step, we use transformations for each 248 time series as outlined in \cite{mccracken2020fred}.

As an exploratory analysis, we look at the sample estimates of the stationary precision matrices for the time series within each group and each period. Specifically, we perform an informal yet illustrative comparison of each economic group's precision matrices from before and after the recession.  Figure~\ref{fig:precision1} shows the empirical precision estimates for the groups NIPA, Employment, Housing, Prices (CPI), and Money and Credits from the pre- and post-recession periods. The images for the remaining groups are given in Section 3 of the supplement.  From the estimates, it is clear that the great recession's impact was substantial across most of the economic sectors. It is expected to see pervasive changes in the interdependence among the economic indicators since the Great Recession is known to have affected every aspect of the economy \cite{Fred}. This is apparent from the estimated precision matrices. Some sectors are impacted more than others, and the interdependencies among the indicators show more resilience in some sectors, such as exchange rates. The exchange rates showed some degree of resistance; \cite{fenn2009dynamic} demonstrates that certain exchange rate interdependencies remained stable, while others evolved significantly during the credit crisis, highlighting the resilience and adaptability of specific currency relationships. These facts are borne out somewhat in a visual inspection of the empirical estimates of the precision matrices from periods before and after the recession. 
However, a more rigorous analysis is needed to substantiate or refute such hypotheses. 

\begin{figure}[htbp]
\centering
\subfigure{\includegraphics[width = 1\textwidth, trim=1cm 1cm 0cm 0cm, clip=true]{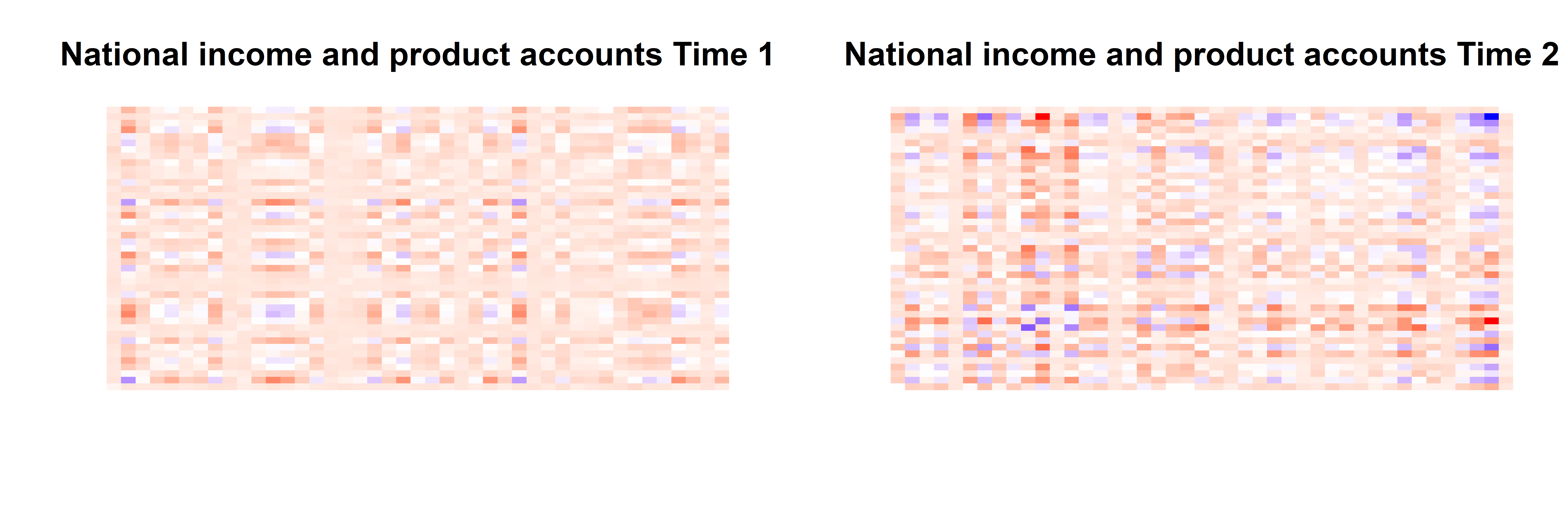}}
\\
\vskip-25pt
\subfigure{\includegraphics[width = 1\textwidth, trim=1cm 1cm 0cm 0cm, clip=true]{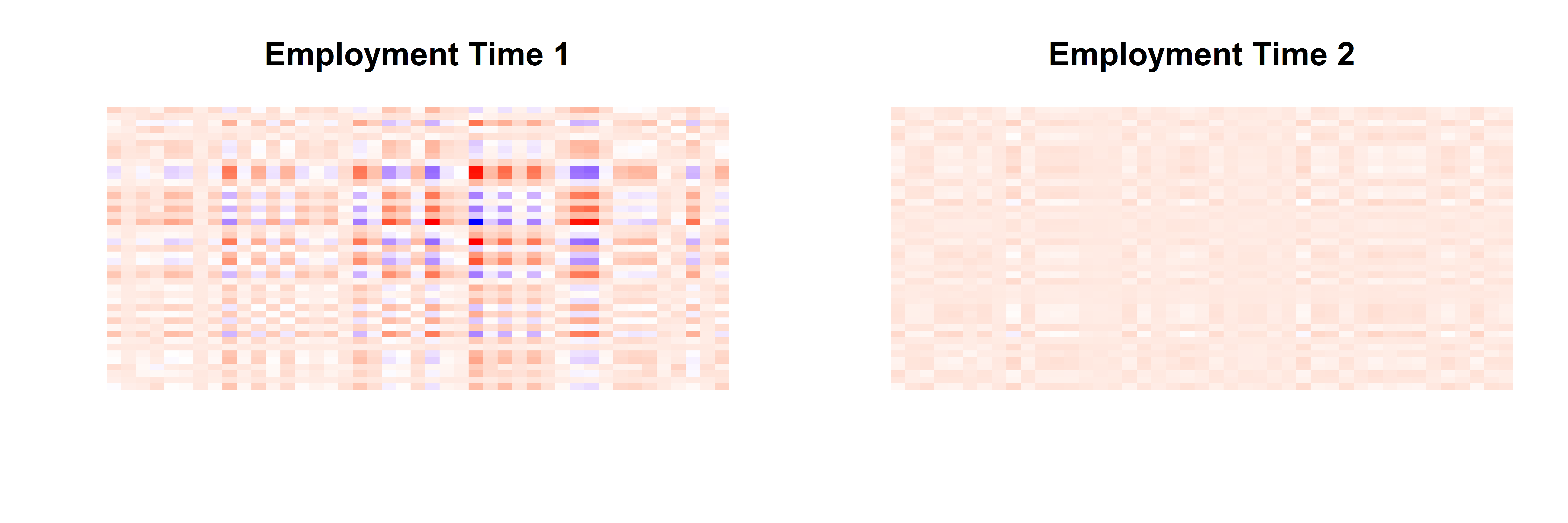}}
\\
\vskip-25pt
\subfigure{\includegraphics[width = 1\textwidth, trim=1cm 1cm 0cm 0cm, clip=true]{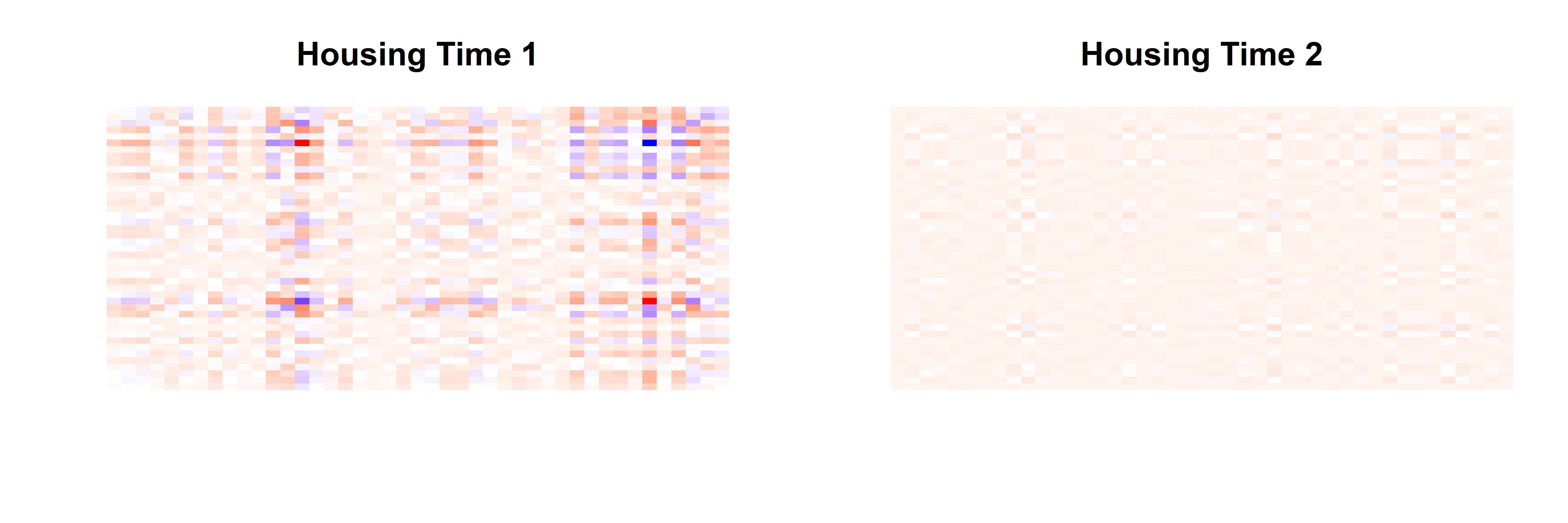}}
\\
\vskip-25pt
\subfigure{\includegraphics[width = 1\textwidth, trim=1cm 1cm 0cm 0cm, clip=true]{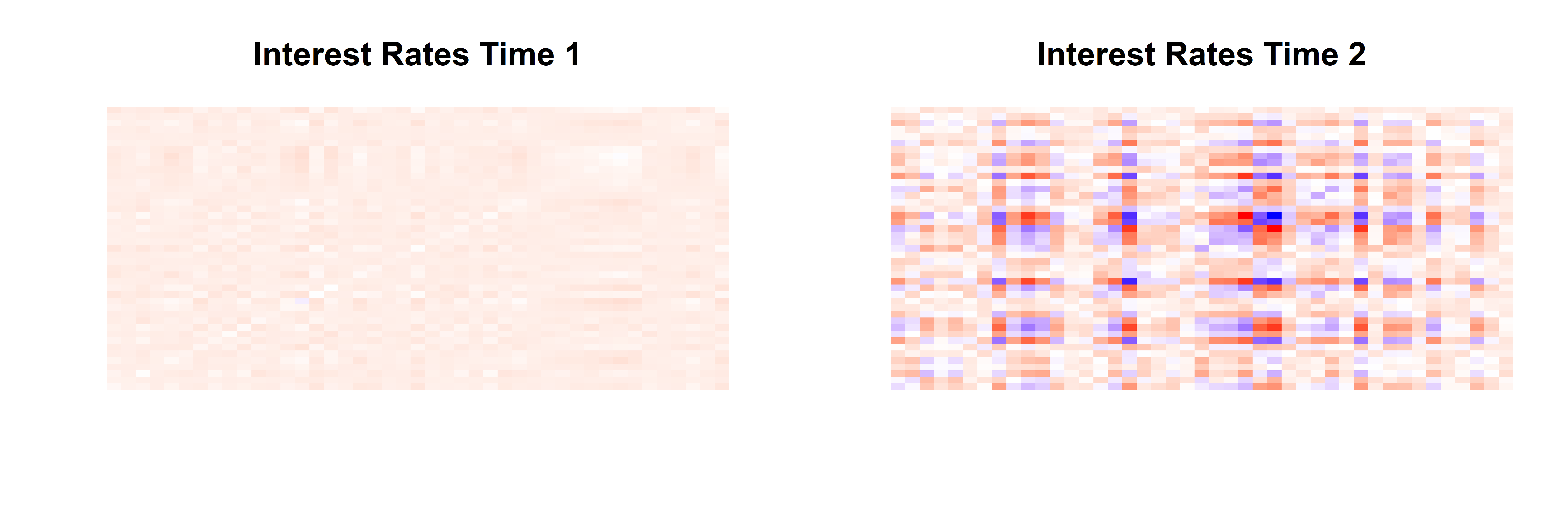}}
\\
\vskip-25pt
\subfigure{\includegraphics[width = 1\textwidth, trim=1cm 1cm 0cm 0cm, clip=true]{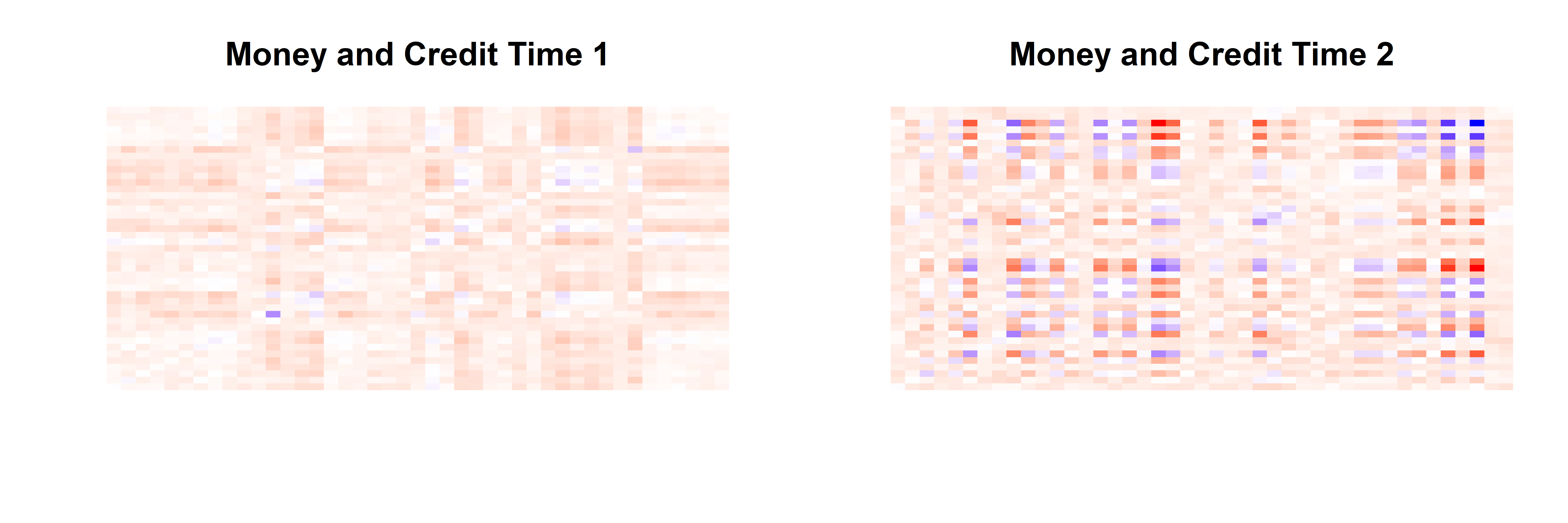}}
\caption{Estimated sample precisions for some example groups}
\label{fig:precision1}
\end{figure}

We use existing modeling tools to perform a model-based analysis to estimate changes in these group-specific time series before and after the recession period. There is evidence that vector autoregressive models, after suitable transformations, adequately describe stationary dynamics in most macroeconomic series \cite{banbura2010large}.  Given that some of the groups have more than twenty series, and that the period of the series is limited, VAR modeling needs to be augmented with some dimension-reduction technique before the entire group can be modeled jointly. To this end, we apply the Sparse-VAR(1) model to the before and after data fragments using the R package {\tt sparsevar} \citep{sparsevarR} with the default elastic net penalty.  To illustrate how the model fits from the pre- and post-recession periods differ, we look at the change in the estimated marginal precision, computed based on the estimated VAR parameters. To this end, for each FRED-QD group $g$, we compute $\hat{\bOmega}_{g,b}$, the precision matrix from before the recession, and $\hat{\bOmega}_{g,a} $, the precision matrix from after the recession period. We then compute the precision difference matrix ${\hat{\bOmega}}^d_g = {\hat{\bOmega}}_{g,a} - {\hat{\bOmega}}_{g,b}$. We then define a normalized edge difference value by dividing the off-diagonal entries of the precision difference matrix by the square root of the associated diagonal entries of the sum of the precision matrices, ${\hat{\bOmega}}^s_g = {\hat{\bOmega}}_{g,b} + {\hat{\bOmega}}_{g,a}$. 
Specifically, we define the scale-free edge-difference parameter as 
\be
\theta_{g}(i,j) = {\hat{\Omega}}^d_{g,i,j}/[{\hat{\Omega}}^s_{g,i,i} {\hat{\Omega}}^s_{g,j,j}]^{1/2}
\lb{eq:edge_diff}
\ee
where for a matrix $\bA$, $A_{i,j}$ denotes the $(i,j)$th entry of $\bA$. 
Subsequently, we report the percentage of edge-difference parameters $\theta_{g,i,j}$ that were bigger than a threshold $\tau$ in absolute value. For the present analysis, we choose $\tau  = 0.1$. The aggregate percentage of significant edge-difference values for each group is reported in Table~\ref{tab:EDA}. Although the reported change scores depend on the choices of $\tau$, the relative rankings of the groups from high to low changes remained stable even for other choices.

\begin{table}[ht]
\centering
\caption{Change from pre to post-recession for different economic groups based on VAR(1) estimates}
\begin{tabular}{lr}
  \hline
 & Change-score \\ 
  \hline
National income and product accounts & 0.05 \\ 
  Industrial Production & 0.22 \\ 
  Employment & 0.06 \\ 
   Unemployment and Participation & 0.10 \\ 
  Housing & 0.21 \\ 
  Inventories, Orders, and Sales & 0.33 \\ 
  Prices (PCE)  & 0.06 \\ 
   Prices (CPI) & 0.07 \\ 
  Earnings
and Productivity & 0.04 \\ 
  Interest Rates & 0.13 \\ 
  Money and Credit & 0.10 \\ 
  Household Balance Sheets & 0.06 \\ 
  Exchange
Rates & 0.17 \\ 
  Stock Markets & 0.33 \\ 
  Non-Household Balance Sheets & 0.23 \\ 
   \hline
\end{tabular}
\label{tab:EDA}
\end{table}

Some groups show appreciable differences from pre- to post-recession, and these findings align with the visual inspection of the empirical precision estimates. However, contrary to the visual findings, several groups do not show significant differences between the pre- and post-recession SparseVar output. These are false negatives, as there is a body of economic research pointing to possibly structural changes in the interdependencies from pre- to post-recession in these component series. The SparseVar findings contradict what is known from economic theory and show the limitations of approaches for estimating changes in interdependencies among the variables when sample dependence is present.  The relationship among indicators in Housing is known to have been severely disrupted by the recession (mainly since the initial dynamics of the recessive economy are attributed to the collapse in the housing market), taking several years in the post-recession period to recover, but possibly at a different level. While the sparse VAR does detect some changes in the interdependencies in the housing market, the degree of change is severely underestimated. This highlights the need for a more flexible method for graph estimation in the presence of VAR-type dependence. Furthermore, the sparse VAR estimation procedure {\tt sparsevar} does not necessarily produce stationary estimates even when the working hypothesis is that the series, after the transformations recommended in \cite{mccracken2020fred}, are stable in the respective periods. This is another deficiency that is overcome by the proposed model. 

\section{Partial Correlation Graph Under Autoregression}

The partial correlation graph for a set of variables $\bX=(X_1,\ldots, X_d)$ can be identified by the precision matrix of $\bX$, that is, the inverse of the dispersion matrix of $\bX$. Two components $X_{j}$ and $X_{k}$ are conditionally uncorrelated given the remaining components if and only if the $(j,k)$th entry of the precision matrix of $\bX$ is zero. Then the relations can be expressed as a graph on $\{1,\ldots,p\}$ where $j$ and $k$ are connected if and only if $X_{j}$ and $X_{k}$ are conditionally correlated given the remaining components. Equivalently, an edge connects $j$ and $k$ if and only if the partial correlation between $X_{j}$ and $X_{k}$ is non-zero. 

In many contexts, the set of variables of interest evolves and is temporally dependent. If the process $\{\bX_t: t=0,1,2,\ldots\}$ is stationary, that is, the joint distributions remain invariant under a time-shift, the relational graph of $\bX_t$ remains time-invariant. 
A vector autoregressive (VAR) process provides a simple, interpretable mechanism for temporal dependence by representing the process as a fixed linear combination of itself at a few immediate time points plus an independent random error. It is widely used in time series modeling. In this paper, we propose a Bayesian method for learning the relational graph of a stationary VAR process.

Let $\bX_1, \ldots, \bX_T$ be a sample of size $T$ from a Vector Autoregressive process of order $p$, VAR$(p)$ in short, given by 
\be
\bX_t = \bmu + \bA_1 \bX_{t-1} + \cdots + \bA_p \bX_{t-p} + \bZ_t
\lb{eq:var}
\ee 
where $\bmu \in \RR^d$, $\bA_1, \ldots, \bA_p$ are $d\times d$ real matrices and $\bZ_t \sim \mathrm{WN}(\bzero, \bSigma)$ is a $d$ dimensional white noise process with covariance matrix $\bSigma$, that is, $\bZ_t$ are independent $\mathrm{N}_p(\bzero, \bSigma)$. We consider $p$ to be given, but in practice, it may not be known and may have to be assessed by some selection methods. 

Causality is a property that plays an important role in multivariate time series models, particularly in terms of forecasting. For a causal time series, the prediction formula includes current and past innovations, and hence, a causal time series allows a stable forecast of the future in terms of present and past data. However, the condition of causality imposes complex constraints on the parameters of the process, often making it extremely difficult to impose causality during the estimation process. A practical approach is to parameterize the constrained parameter space of causal processes in terms of unconstrained parameters and write the likelihood and the prior distribution in terms of the unconstrained parameters.

For a VAR$(p)$ time series defined by \eqref{eq:var}, causality is determined by the roots of the determinantal equation 
\begin{align} 
\det(\bA(z)) = 0 \mbox{ where } \bA(z) = \bI - \bA_1 z - \cdots - \bA_p z^p, \; z \in \CC, 
\end{align}
the matrix polynomial associated with the VAR equation. For the VAR model to be causal, all roots of the determinantal equation must lie outside the unit disc $\mcD = \{z \in \CC: |z| \leq 1\}$. When the roots of the determinantal equation lie outside the unit disc, the associated monic matrix polynomial ${\tilde{\bA}}(z) = z^p - \bA_1 z^{p-1} - \cdots - \bA_p$ is called {\it Schur stable.}
\cite{roy2019constrained} provided a parameterization of the entire class of Schur stable polynomials and thereby parameterized the space of causal VAR models. 

For convenience, we briefly describe the \cite{roy2019constrained} framework. \cite{roy2019constrained} noted that the VAR$(p)$ model \eqref{eq:var} is causal if and only if the block Toeplitz covariance matrix $\bUpsilon_{p}$ is positive definite, where 
\be 
\bUpsilon_{j} = \begin{bmatrix} \bGamma(0) & \bGamma(1) & \cdots & \bGamma(j) \\
\bGamma(1)^{\mathrm{T}} & \bGamma(0) & \cdots & \bGamma(j-1) \\
\vdots & \vdots & \ddots & \vdots \\
\bGamma(j)^{\mathrm{T}} & \bGamma(j-1)^{\mathrm{T}} & \cdots & \bGamma(0) \end{bmatrix}
\lb{eq:blockToeplitz}
\ee  
is the covariance matrix of $(j+1)$ consecutive observations $(\bX_t, \bX_{t-1}, \ldots, \bX_{t-j})^{\mathrm{T}}$ and  
\be 
\bGamma(h) = \EE[\bX_t - \EE(\bX_t))(\bX_{t-h}  - \EE(\bX_{t-h}))^{\mathrm{T}}]
\lb{eq:gamma}
\ee  
is the lag-$h$ autocorrelation matrix of the process. As shown in \cite{roy2019constrained}, this condition is equivalent to 
\be  
\bGamma(0) = \bC_0  \geq \bC_1 \geq \cdots \geq \bC_p = \bSigma ,
\lb{eq:CausalCond}
\ee 
where  $\bC_j = \mathrm{Var}(\bX_{j+1} |\bX_{j-1}, \ldots, \bX_1)$ are the the conditional dispersion matrices; here and below, we use the Lowner ordering: for two matrices $\bA, \bB$, $\bA \geq \bB$ means that $\bA - \bB$ is nonnegative definite. The condition  \eqref{eq:CausalCond} plays a central role in the formulation of our parameterization. Since the VAR parameters $\bA_1, \ldots, \bA_p, \bSigma$ can be expressed as a one-to-one map of the sequence $\bC_0 - \bC_1, \ldots, \bC_{p-1} - \bC_p, \bC_p$, the parameterization of the causal VAR process is achieved by parameterizing the nonnegative matrices of successive differences $\bC_{j-1} - \bC_j, j = 1, \ldots, p$, and the positive definite matrix $\bC_p$ in terms of unconstrained parameters. Several options for parameterizing nonnegative matrices in terms of unconstrained parameters are available in the literature. 

Our main objectives in this paper are to estimate the stationary precision matrix of the process, $[\Var(X_t)]^{-1} = \Omega = \bGamma^{-1}(0)$, under sparsity restrictions.  In the \cite{roy2019constrained} parameterization, the stationary variance $\bGamma(0)$ and hence $\bOmega$ are functions of the basic free parameters and hence not suitable for estimation of the graphical structure with desired sparsity properties. We suggest a novel modification of the previous parameterization that achieves the goal of parameterizing the graphical structure and the VAR correlation structure directly and separately, thereby facilitating the estimation of both components under the desired restrictions, even in higher dimensions.

\section{Reduced Rank Parameterization and Priors}

The main idea in the proposed parameterization is the separation of the parameters pertaining to the graphical structure, $\bOmega,$ and the parameters used to describe the temporal correlation present in the sample. The following result is essential in developing the parameterization. It follows easily from \eqref{eq:CausalCond}, but we state it formally due to its central nature in the new parameterization. 

\begin{proposition}
Let $\bX_t$ satisfying \eqref{eq:var} be a stationary vector autoregressive time series with stationary variance $\Var(\bX_t)=\bOmega^{-1}$, a positive definite matrix and error covariance matrix $\bSigma$. Then $\bX_t$ is causal if and only if  
\begin{align} 
\Omega \leq \bC_1^{-1} \leq \cdots \leq \bC_p^{-1} = \bSigma^{-1}, 
\end{align}
\lb{eq:NewParameterization}
where for any $j \geq 1,$ $\bC_j = \Var(\bX_{j+1} | \bX_{j}, \ldots, \bX_1)$.  
\lb{prop:NewParameterization} 
\end{proposition}

\noindent {\bf Constraint-free parameterization:}
Based on Proposition~\ref{prop:NewParameterization}, for a causal VAR process the successive differences $\bC_j^{-1} - \bC_{j-1}^{-1}$ are nonnegative definite for $j = 1, \ldots, p$,  where $\bOmega = \bC_0^{-1}$. Hence, these successive differences and the precision matrix $\bOmega$ can be parameterized using an unrestricted parameterization that maps non-negative definite matrices to free parameters.

\subsection{Efficient Computation of the Likelihood}
\label{sec:likelihood_comp}
The computation of the likelihood is difficult for a multivariate time series. For a Gaussian VAR process, the likelihood can be computed relatively fast using the Durbin-Levison (DL) or innovations algorithm \citep{brockwell2009time}. However, under different parameterizations, the computational burden can increase significantly depending on the complexity of the parameterization. Moreover, the DL-type algorithm can still be challenging if the process dimension is high.

In models with many parameters,  a common approach is to seek low-dimensional sub-models that could adequately describe the data and answer basic inferential questions of interest.

Using the proposed parameterization of $\bOmega$ and the differences $\bC_{j}^{-1} - \bC_{j-1}^{-1}$, we provide a low-rank formulation of a causal VAR process that leads to an efficient algorithm for computing the likelihood based on a Gaussian VAR($p$) sample. The algorithm achieves computational efficiency by avoiding the inversion of large-dimensional matrices. 
The sparse precision matrix $\bOmega = \bC_0^{-1}$ is modeled as a separate parameter, allowing direct inference about the graphical structure  under sparsity constraints . The likelihood computation requires the inversion of only $r_j \times r_j $ matrices instead of $d\times d$ matrices. The last fact substantially decreases the computational complexity. In particular, when $r_j = r = 1$,  i.e. when the conditional precision updates, $\bC_{j}^{-1} - \bC_{j-1}^{-1}$,  are all rank-one, all components of likelihood computation can be done without matrix inversion or factorization. 

Before we proceed, we define some notations that are used throughout the article. 
The stationary variance matrix $\bUpsilon_j$ of $(\bX_t, \bX_{t-1} , \ldots, \bX_{t-j})^{\mathrm{T}},$ as defined in \eqref{eq:blockToeplitz}, will be written in the following nested structures
\be   
\bUpsilon_j = \begin{bmatrix} \bGamma(0) & \bxi_{j}^{\mathrm{T}} \\
\bxi_{j} & \bUpsilon_{j-1} \end{bmatrix} = \begin{bmatrix} \bUpsilon_{j-1} & \bkappa_{j} \\
\bkappa_{j}^{\mathrm{T}} &  \bGamma(0) \end{bmatrix}
\ee   
where 
\be 
\bxi_j^{\mathrm{T}} = (\bGamma(1), \,\bGamma(2),\ldots,\bGamma(j)),\quad
\bkappa_j^{\mathrm{T}} = (\bGamma(j)^{\mathrm{T}}, \ldots,\bGamma(2)^{\mathrm{T}},\,\bGamma(1)^{\mathrm{T}}).
\lb{eq:XiKappa}
\ee 
Denote the Schur-complements of $\bUpsilon_{j-1}$ in the two representations as 
\be 
\bC_j = \bGamma(0) - \bxi_j^{\mathrm{T}} \bUpsilon_{j-1}^{-1} \bxi_j,\quad
\bD_j  = \bGamma(0) - \bkappa_j^{\mathrm{T}} \bUpsilon_{j-1}^{-1} \bkappa_j.
\lb{eq:SchurComplement}
\ee
Also, let $\phi_d(\cdot|\bmu,\bSigma)$ and $\Phi_d(\cdot|\bmu,\bSigma)$ denote the probability density and cumulative distribution function of the $d$-dimensional normal with mean $\bmu$ and covariance $\bSigma$. 
The basic computation will be to successively compute the likelihood contribution of the conditional densities $f(\bX_j|\bX_{j-1}, \ldots, \bX_1),\, j = 2, \ldots, T$. Let 
 $\bY_j = (\bX_j^{\mathrm{T}}, \ldots, \bX_{\max(1, j - p)}^{\mathrm{T}})^{\mathrm{T}}$. 
Under the assumption that the errors are Gaussian, i.e., $\bZ_j \sim \mathrm{N}_d(\bzero, \bSigma)$, and that they are independent, and writing $f(\bX_1| \bY_0) = f(\bX_1)$, $\mu_1 = \bzero$, $\bSigma_1 = \bC_0$  we have, 
\be  
f(\bX_j|\bY_{j-1}) = \phi_d(\bX_j|\bmu_j, \bSigma_j),
\ee  
where the conditional mean and variance are given by 
$\bmu_j = \bxi_{j-1}^{\mathrm{T}}\bUpsilon_{j-2}^{-1}\bY_{j-1}$, $\bSigma_j = \bC_{j-1}$ for any $ 2 \leq j \leq p,$ and $\bmu_j = \bxi_p^{\mathrm{T}}\bUpsilon_{p-1}^{-1}\bY_j$, $\bSigma_j = \bC_p$ for $p+1 \leq j \leq T$. 
Thus, the full Gaussian likelihood 
\begin{align}
 L = f(\bX_1)\prod_{j=2}^{\mathrm{T}} f(\bX_j | \bY_{j-1}) 
 \end{align}
can be obtained by recursively deriving the conditional means and variances from the constraint-free parameters describing the sparse precision matrix $\bOmega$ and the reduced-rank conditional variance differences, $\bC_j^{-1} - \bC_{j-1}^{-1}, j = 1, \ldots, p$. 
For parameterization of the conditional precision updates, we use a low-rank parameterization. Specifically, let  
\begin{align}
\bC_j^{-1} - \bC_{j-1}^{-1} = \bL_j\bL_j^{\mathrm{T}}, \mbox{ where }\bL_j \mbox{ are }d\times r_j \mbox{ matrices with }r_j \ll d. 
\end{align}
The rank factors $\bL_j$ are not directly solvable from the precision updates $\bC_j^{-1} - \bC_{j-1}^{-1}$. To complete the parameterization, we need to define a bijection. The bijection can be defined by  augmenting $\bL_j$ with $d\times r_j$ semi-orthogonal matrices $\bQ_j$ and define $\bL_j\bQ_j^T$ as a unique square-root of $\bC_j^{-1} - \bC_{j-1}^{-1}$, such as the unique symmetric square-root. This would lead to an identifiable parameterization. However, given that the $d\times r_j$ entries in $\bQ_j$ are constrained by the orthogonality requirement, 
we will use a slightly over-identified system of parameters to describe the reduced rank formulation of $\bC_j^{-1} - \bC_{j-1}^{-1}$, $j = 1, \ldots, p$. The over-identification facilitates computation enormously without creating any challenges in inference for the parameters of interest. 

Specifically, for each $j = 1, \ldots, p,$ along with $\bL_j$ we will use $d\times r_j$ parameters arranged in a $d\times r_j$ matrix, $\bK_j$, to describe the reduced-rank updates $\bC_j^{-1} - \bC_{j-1}^{-1}$. The exact definition of the matrices $\bK_j$ along with the steps for computing the causal VAR likelihood based on the basic constraint-free parameters $\bOmega$, $\bL_j$, and $\bK_j$, $j = 1, \ldots,p$ are given below. We assume the ranks $r_1, \ldots, r_p$ are specified and fixed. Also, unless otherwise specified, we will assume the matrix square roots for pd matrices to be the unique symmetric square-root. Then, given the parameters $\bOmega, \bL_1, \ldots, \bL_p$ and the initialization $\bC_0^{-1} = \bOmega,$ the following steps describe the remaining parameters $\bK_1, \ldots, \bK_p$ recursively  along with recursive computation of the components $f(\bX_j | \bY_{j-1})$ of the VAR likelihood.\\ 

\noindent{\bf Recursive Computation of the Reduced Rank VAR Likelihood}

For $j=1, \ldots, p,$ 
\begin{enumerate}
    \item Since $\bC^{-1}_j - \bC_{j-1}^{-1} = \bL_j \bL_j^T, $  we have $\bC_j^{-1} = \bC_{j-1}^{-1} + \bL_j \bL_j^{\mathrm{T}}.$
 
 \item
 Using the Sherman-Woodbury-Morrison (SWM) formula for partitioned matrices, 
$  
    \bC_j = \bC_{j-1} - \bU_j \bU_j^{\mathrm{T}},
$    
where 
\be \bU_j = \bC_{j-1}\bL_j(\bI + \bL^{\mathrm{T}}_j \bC_{j-1}\bL_j)^{-1/2}
\lb{eq:Uj}
\ee
Note that,  
\beqa  
\bU_j^{\mathrm{T}}\bC_{j-1}^{-1}\bU_j  &=& (\bI + \bL^{\mathrm{T}}_j \bC_{j-1}\bL_j)^{-1/2} \bU_j^{\mathrm{T}}\bC_{j-1}\bC_{j-1}^{-1}\bC_{j-1} (\bI + \bL^{\mathrm{T}}_j \bC_{j-1}\bL_j)^{-1/2}\\
&=& (\bI + \bL^{\mathrm{T}}_j \bC_{j-1}\bL_j)^{-1/2} \bL^{\mathrm{T}}_j \bC_{j-1}\bL_j (\bI + \bL^{\mathrm{T}}_j \bC_{j-1}\bL_j)^{-1/2}.
\eeqa  
Since for a positive definite matrix $\bA$, $\|(\bI + \bA)^{-1/2}\bA(\bI + \bA)^{-1/2}\| = \frac{\|\bA\|}{1 + \|\bA\|} < 1,$ 
$\bU_j$ satisfies the restriction $\|\bU_j^{\mathrm{T}}\bC_{j-1}^{-1}\bU_j \| < 1.$  
\item
Define 
\be  
\bW_j = \bGamma(j)^{\mathrm{T}} - \bxi_{j-1}\bUpsilon_{j-2}^{-1}\bkappa_{j-1},
\lb{eq:Wj}
\ee  
with $\bW_1 = \bGamma(1)^T.$ From \cite{roy2019constrained}, $
\bU_j\bU_j^{\mathrm{T}} = \bW_j\bD_{j-1}^{-1}\bW_j^{\mathrm{T}}.
$
and hence for any $\bV_j$ such that $\bV_j^{\mathrm{T}}\bD_{j-1}^{-1}\bV_j = \bI$ we have 
\begin{align}
\bW_j = \bU_j\bV_j^{\mathrm{T}}.
\end{align}
While $\bU_j$ are determined by $\bL_j$, the $\bV_j$ matrices are determined by the other parameters $\bK_j$. We construct $\bV_j$ from the basic parameters $\bK_j$ as, 
\be 
\bV_j = \bK_j(\bK_j^{\mathrm{T}}\bD_{j-1}^{-1}\bK_j)^{-1/2}.
\lb{eq:Vj}
\ee
Note that $\bV_j^{\mathrm{T}} \bD_{j-1}^{-1} \bV_j = \bI$. 
\item 
Thus, entries of  the covariance matrices can be updated as  
\[
    \bGamma(j)^{\mathrm{T}} = \bU_j\bV_j^{\mathrm{T}} + \bxi_{j-1}^{\mathrm{T}}\bUpsilon_{j-2}^{-1} \bkappa_{j-1},
    \quad \bxi_j^{\mathrm{T}} = (\bxi_{j-1}^{\mathrm{T}}, \;\;\bGamma(j)), \qquad 
    \bkappa_j^{\mathrm{T}} = (\bGamma(j)^{\mathrm{T}}, \;\;\bkappa_{j-1}^{\mathrm{T}}). 
\]
\item
Applying SWM successively, we have
\beqa
    \bD_j &=& \bD_{j-1} - \bR_j\bR_j^{\mathrm{T}},\\
    \bD_j^{-1} &=& \bD_{j-1}^{-1} +  \bP_j\bP_j^{\mathrm{T}},
\eeqa
where $\bR_j = \bV_j(\bU_j^{\mathrm{T}}\bC_{j-1}^{-1}\bU_j)^{1/2}$ and   $ \bP_j = \bD_{j-1}^{-1}\bR_j(\bI - \bR_j^{\mathrm{T}}\bD_{j-1}^{-1}\bR_j)^{-1/2}$. For the last update to hold one needs $(\bI - \bR_j^{\mathrm{T}}\bD_{j-1}^{-1}\bR_j)$ to be positive definite. This follows from the fact that:
\[
\bR_j^{\mathrm{T}}\bD_{j-1}^{-1}\bR_j = (\bU_j^{\mathrm{T}}\bC_{j-1}^{-1}\bU_j)^{1/2}\bV_j^{\mathrm{T}}\bD_{j-1}^{-1} \bV_j (\bU_j^{\mathrm{T}}\bC_{j-1}^{-1}\bU_j)^{1/2}
= \bU_j^{\mathrm{T}}\bC_{j-1}^{-1}\bU_j,
\]
where the final equation holds because $\bV_j^{\mathrm{T}}\bD_{j-1}^{-1}\bV_j = \bI.$
Recalling  $\|\bU_j^{\mathrm{T}}\bC_{j-1}^{-1}\bU_j \| < 1,$
 the result follows. Thus, 
$
\bD_j^{-1}=\bD_{j-1}^{-1} + \bD_{j-1}\bV_j(\bI-\bU_j^{\mathrm{T}}\bC_{j-1}^{-1}\bU_j)^{-1}\bV_j^{\mathrm{T}}\bD_{j-1}.
$

\item
Finally, the $j$th conditional density in the likelihood is updated as 
\[
f(\bX_j \;| \;\bY_{j-1}) = \phi_d(\bX_j|\bxi_{j-1}^{\mathrm{T}}\bUpsilon_{j-2}^{-1}\bY_{j-1}, \bC_{j-1}).
\]
The determinant term can be updated recursively as
$
\det(\bC_j^{-1})  = \det(\bC_{j-2}^{-1})\det(\bI + \bL_{j-1}^{\mathrm{T}}\bC_{j-2}\bL_{j-1}).
$
\end{enumerate}

The subsequent updates $j = (p+1), \dots, T$, required to compute the rest of the factors of the likelihood,  can be obtained simply noting $\bmu_j = \bxi_p \bUpsilon_{p-1}^{-1}\bY_{j-1,p}$ and $\bC_{j-1}^{-1} = \bC_p^{-1}$ where $\bY_{j-1,p} = (\bX_{j-1}^T, \ldots, \bX_{j-p}^T)^T$

The over-parameterization of the reduced rank matrices is intentional; it reduces computational complexity. Consistent posterior inference on the set of identifiable parameters of interest is still possible, and the details are given in subsequent sections. For a general $d\times d$ matrix with rank $r$, the number of free parameters is $d^2 - (d-r)^2 = 2dr - r^2$ where $(d-r)^2$ is the codimension of the rank manifold of $d\times d$ rank $r$ matrices. In the proposed parameterization, we are parameterizing the rank $r_j$ updates $\bC_{j}^{-1} - \bC_{j-1}^{-1}$ using $2d r_j$ parameters in the matrix pair $(\bL_j, \bK_j).$ Thus, we have $r_j^2$ extra parameters that are being introduced for computational convenience. Typically, $r_j$ will be small, and hence the number of extra parameters will be small as well. 

\subsection{Rank-one Updates}

The case when the increments $\bC_{j}^{-1} - \bC_{j-1}^{-1}$ are rank one matrices, i.e. $r_j = 1$ for all $j$,  is of special interest. Then, the parameters $\bU_j, \bV_j$ are fully identifiable if one fixes the sign of the first entry in $\bU_j$. Moreover, in this case, the likelihood can be computed without inverting any matrices or computing square roots.

Specifically, the quantities from the original algorithm can be simplified and recursively computed. The key quantities to be computed at each iteration are the matrices $\bC_j, \bC_j^{-1}, \bD_j, \bD_j^{-1}, \bU_j, \bV_j.$ At the $j$th stage of recursion, Let $l_j$ and $k_j$ be  scalar quantities defined as $l_j = \bL^{\mathrm{T}}_j \bC_{j-1}\bL_j$ and $k_j = \bK_j^{\mathrm{T}} \bD_{j-1}^{-1} \bK_j$. Then, given the vectors $L_j, K_j$ 
\begin{eqnarray*}
\bC_j^{-1} &=& \bC_{j-1}^{-1} +  \bL_j \bL_j^{\mathrm{T}},\\
\bU_j &=& \frac{1}{\sqrt{(1 + z_j)}}\bC_{j-1}\bL_j, \\
\bC_j &=& \bC_{j-1} - \bU_j \bU_j^{\mathrm{T}},\\
   \bV_j &=& \frac{1}{k_j}\bK_j, \\
    \bD_j &=& \bD_{j-1} - \frac{l_j}{(1 + l_j)}\bV_j\bV_j^{\mathrm{T}}, \\
    \bD_j^{-1} &=& \bD_{j-1}^{-1} + \frac{1}{1 + l_j}\bD_{j-1}\bV_j\bV_j^{\mathrm{T}}\bD_{j-1}.
\end{eqnarray*}
When the dimension is large, the rank-one update parameterization can effectively capture a reasonable dependence structure while providing extremely fast updates in the likelihood computation. In the simulation section, we demonstrate the effectiveness of the rank-one update method via numerical illustrations. 

\subsection{Prior Specification}
\label{sec: prior}

If the model is assumed to be causal, the prior should charge only autoregression coefficients complying with causality restrictions while putting prior probability directly on the stationary covariance matrix. Presently, the literature lacks a prior that charges only causal processes with the prior on the precision $\bOmega$ as an independent component of the prior. Using our parametrizations, we achieve that. Recall that in our setup, $p$ is considered given; hence, no prior distribution is assigned to $p$. Also, the ranks $r_1, \ldots, r_p$ are specified. While all of the procedures described in this paper go through where the $p$ ranks are potentially different, we develop the methodology of low-rank updates by fixing $r_1 = \cdots = r_p = r$. 

We consider the modified Cholesky decomposition of $\bOmega=(\bI_p-\bE)\bF(\bI_p-\bE)^{\mathrm{T}}$,  where $\bE$ is strictly lower triangular,  $\bF=\diag(\bm{f})$ is diagonal with $\bm{f}$ the vector of entries.  
For sparse estimation of $\bSigma^{-1}$, we put sparsity inducing prior on $\bE$.
We first define a hard thresholding operator $H_{\lambda}(\bH)=(\!(h_{i,j}\mathbbm{1}\{|h_{i,j}|>\lambda\} )\!)$. Then we set, $\bE=H_{\lambda}(\bE_1)$, where $\bE_1$ is a strictly lower triangular matrix. For off-diagonal entries $\bE_1$, we let $e_{ij,1}\sim\mathrm{N}(0,\sigma^2_{e})$ and $\sigma^2_{e}\sim\IG(c_1,c_1)$ for $i<j$. The components $f_1,\ldots,f_p$ of $\bm{f}$ are independently distributed according to the inverse Gaussian distributions with density function 
$\pi_{d}(t) \propto t^{-3/2} e^{-(t-\xi)^2/(2t)}$, $t > 0$, for some $\xi>0$; see 
\citep{chhikara1988inverse}. This prior has an exponential-like tail near both zero and infinity. We put a weakly informative mean-zero normal prior with a large variance on $\xi$. While employing sparsity using hard thresholding, the modified Cholesky form is easier to work with. Lastly, we put a Uniform prior on $\lambda$.  
\begin{itemize}
    \item Conditional precision updates $\bL_j$ of dimension $d\times r$: We first build the matrix $\bLambda=[\blambda_1;\blambda_2;\ldots\blambda_p]$, where $\blambda_j=\vect(\bL_j)$ placing them one after the other and $p$ is a pre-specified maximum order of the estimated VAR model. Subsequently, we impose a cumulative shrinkage prior on $\bLambda$ to ensure that the higher-order columns are shrunk to zero. Furthermore, for each individual $\bL_j$ too, we impose another cumulative shrinkage prior to shrinking higher-order columns in $\bL_j$. Our prior follows the cumulative shrinkage architecture from \cite{bhattacharya2011sparse}, but with two layers of cumulative shrinkage on $\Lambda$ to shrink its higher-order columns as well as the entries corresponding to $\bL_j$'s with higher $j$: 
$$\lambda_{\ell k}|\phi_{\ell k},\tau_{k}\psi_{k,\lceil{\ell}/{r}\rceil}\sim \mathrm{N}(0,\phi_{\ell k}^{-1}\tau_{k}^{-1}\psi_{k,\lceil\frac{\ell}{r}\rceil}^{-1}),$$
	$$\phi_{lk}\sim \mathrm{Ga}(\nu_1,\nu_1),\quad \tau_{k}=\prod_{i=1}^k\delta_{i},\quad \psi_{k,m} = \prod_{i=1}^{m}\delta^{(k)}_{i}$$
	$$\delta_{1}\sim \mathrm{Ga}(\kappa_{1}, 1),\quad\delta_{i}\sim \mathrm{Ga}(\kappa_{2}, 1)\textrm{ for }i\geq 2.$$
 $$\delta^{(k)}_{1}\sim \mathrm{Ga}(\kappa_{1}, 1),\quad\delta^{(k)}_{i}\sim \mathrm{Ga}(\kappa_{2}, 1)\textrm{ for }i\geq 2,$$
 where $\lceil x\rceil$ stands for the ceiling function that maps to the smallest integer greater than or equal to $x$, and Ga stands for the gamma distribution.
The parameters $\phi_{lk}$ control local shrinkage of the elements in $\Lambda$, whereas $\tau_k$ controls column shrinkage of the $k$-$th$ column. 
Similarly, $\psi_{k,\lceil{\ell}/{r}\rceil}^{-1})$ helps to shrink higher-order columns in $\bL_k$. Following the well-known shrinkage properties of multiplicative gamma prior, we let $\kappa_1=2.1$ and $\kappa_2=3.1$, which work well in all of our simulation experiments. However, in the case of rank-one updates, the $\psi_{k,m}$'s are omitted. Computationally, it seems reasonable to keep $\psi_{k,m}=1$ as long as $r$ is pre-specified to a small number.
    \item $\bK_j$ of dimension $d\times r$: We consider non-informative flat prior for the entries in $\bK_j$. The rank of $\bK_j$ is simultaneously specified with that of $\bL_j$. Hence, the cumulative shrinkage prior for $\bK_j$ is not needed.
\end{itemize}

\subsection{Posterior Computation} 
\label{sec:compuprec}

We use Markov chain Monte Carlo algorithms for posterior inference. The individual sampling strategies are described below. Due to the non-smooth and non-linear mapping between the parameters and the likelihood, it is convenient to
use the Metropolis-within-Gibbs samplers for different parameters. 

Adaptive Metropolis-Hastings (M-H) moves update the lower-triangular entries in the latent $\bE_1$ and the entries in $\bm{f}$. Due to the positivity constraint on $\bm{f}$, we update this parameter in the log scale with the necessary Jacobian adjustment. Specifically, we generate the updates from a multivariate normal, where the associated covariance matrix is computed based on the generated posterior samples. Our algorithm is similar to \cite{haario2001adaptive} with some modifications, as discussed below. The initial part of the chain relies on random-walk updates, as no information is available to compute the covariance. However, after the 3500th iteration, we start computing the covariance matrices based on the last $S$ accepted samples. The choice of 3500 is based on our extensive simulation experimentation. Instead of updating these matrices at each iteration, we update them once at the end of each 100th iteration using the last $S$ accepted samples. The value of $S$ is increased gradually. Furthermore, the constant variance, multiplied by the covariance matrix as in \cite{haario2001adaptive}, is tuned to maintain a pre-specified level of acceptance. The thresholding parameter $\lambda$ is also updated on a log scale with a Jacobian adjustment. \\
\begin{enumerate}
   \item Adaptive Metropolis-Hastings update for each column in $\bLambda$ and full conditional Gibbs updates for the other hyperparameters.
    \item Adaptive M-H update for $\bK_j$.
\end{enumerate}

To speed up the computation, we initialize $\bOmega$ to the graphical lasso \cite{friedman2008sparse} output using  {\tt glasso} R package based on the marginal distribution of multivariate time series at every cross-section, ignoring the dependence for a warm start. From this $\bOmega$, the modified Cholesky parameters $\bE$ and $\bF$ are initialized. We start the chain setting $p=\min(p_m, T/2)$, where $p_m$ is a pre-specified lower bound, and initialize the entries in $\bL_j$'s and $\bK_j$'s from $\Normal(0,1/j)$. At the 1000th iteration, we discard the $j$'s if the entries in $\bL_j$ have very little contribution. In the case of the M-H algorithm, the acceptance rate is maintained between 25\% and 50\% to ensure adequate mixing of posterior samples.

Recently, \cite{heaps2023enforcing} applied the parametrization from \cite{roy2019constrained} with flat priors on the new set of parameters and developed HMC-based MCMC computation using {\tt rstan}. Our priors, however, involve several layers of sparsity structures and, hence, are inherently more complex. We primarily rely on M-H sampling, as direct computation of the required gradients for HMC is difficult. We leave the exploration of more efficient implementation using {\tt rstan} to implement gradient-based samplers as part of possible future investigations.

\section{Posterior Contraction}
 
 In this section, we study the asymptotic properties of the posterior of the posterior distribution under some additional boundedness conditions on the support of the prior for the collection of all parameters $\bzeta:=(\bOmega,\bL_1,\ldots,\bL_p,\bK_1,\ldots,\bK_p)$. The corresponding true value is denoted by $\bzeta_0=(\bOmega_0,\bL_{1,0},\ldots,\bL_{p,0},\bK_{1,0},\ldots,\bK_{p,0})$ with the related true autoregression coefficients $\bA_{1,0},\ldots,\bA_{p,0}$. As the parameters $\bL_1,\ldots,\bL_p,$ $\bK_1,\ldots,\bK_p$ are not identifiable, the above true values may not be unique. Below, we assume that the true model parameters represent the assumed form with some choice of true values satisfying the required conditions. The posterior contraction rate is finally obtained for the identifiable parameters $\bOmega$ and $\bA_1,\ldots,\bA_p$.

\begin{description}
\item [(A1)] The common rank $r$ of $\bL_{1},\ldots,\bL_p,\bK_1,\ldots,\bK_p$ is given. 
\item [(A2)] The prior densities of all entries of $\bE, \bm{f},\bL_{1},\ldots,\bL_p,\bK_1,\ldots,\bK_p$ are positive at the true values $\bE_0, \bm{f}_0,\bL_{1,0},\ldots,\bL_{p,0},\bK_{1,0},\ldots,\bK_{p,0}$ of $\bE, \bm{f},\bL_{1},\ldots,\bL_p,\bK_1,\ldots,\bK_p$ respectively. 
\item [(A3)] The entries of $\bE, \bm{f},\bL_{1},\ldots,\bL_p,\bK_1,\ldots,\bK_p$ are independent, and their densities are positive, continuous. Further, the prior densities of the entries of $\bm{f},\bL_{1},\ldots,\bL_p$ have a common power-exponential tail (that, bounded by $C_0\exp[-c_0|x|^\gamma]$ for some constants $C_0,c_0,\gamma>0$). 
\item [(A4)] The entries of $\bm{f}$ are independent, bounded below by a fixed positive number, and have a common power-exponential upper tail.
\item [(A5)] The entries of $\bm{f}_0,\bE_0,\bL_{1,0},\ldots,\bL_{p,0},\bK_{1,0},\ldots,\bK_{p,0},\bA_{1,0},\ldots,\bA_{p,0}$ and the eigenvalues of $\bOmega_0$ and $\bSigma_0$ lie within a fixed interval, and those of $\bm{f}_0$ lie in the interior of the support of the prior for $\bm{f}$. 
\end{description}

Condition (A1) can be disposed of, and the rate in Theorem~\ref{thm:basic rate} will remain valid, but writing the proof will become more cumbersome. 
We note that the assumed bounds on $\bm{f}$ in Assumption (A4) ensure that the eigenvalues of $\bOmega$ are bounded below by a fixed positive number. This also ensures that the eigenvalues of $\bSigma^{-1}$ are also bounded below by a fixed positive number; that is, the eigenvalues of $\bSigma$ are bounded above, given the representation $\bSigma^{-1}=\bOmega+\sum_{j=1}^p \bL_j \bL_j^{\mathrm{T}}$ and all terms inside the sum are nonnegative definite. The condition that the entries of $\bm{f}$ are bounded below is not essential; it can be removed at the expense of weakening the Frobenius distance or by assuming a skinny tail of the prior density at zero and slightly weakening the rate, depending on the tail of the prior. We forgo the slight generalization in favor of clarity and simplicity. The condition can be satisfied with a minor change in the methodology by replacing the inverse Gaussian prior for $\bm{f}$ with a lower-truncated version, provided that the truncation does not exclude the true value. The last assertion will be valid unless the true precision matrix is close to singularity.

Let $\bY_T=\vect(\bX_1, \bX_2,\ldots,\bX_T)$. Recall that the joint distribution of $\bY_T$ is $dT$-variate normal with mean the zero vector and dispersion matrix $\bUpsilon_T$. 
Thus the likelihood is given by $Q_{\bzeta}=(\sqrt{2\pi\det(\bUpsilon_T)})^{-1/2} \exp[-\bY_T^{\mathrm{T}}\bUpsilon_T^{-1}\bY_T/2]$.  

\begin{theorem} 
\label{thm:basic rate}
Under Conditions (A1)--(A5), the posterior contraction rate for $\bOmega$ at $\bOmega_0$ and for the autoregression parameters $\bA_1,\ldots,\bA_p$ at $\bA_{1,0},\ldots,\bA_{p,0}$ with respect to the Frobenius distance is $d \sqrt{(\log T)/T}$. 
\end{theorem}

If the true precision matrix $\bOmega_0$ has an appropriate lower-dimensional structure and the rank $r$ of the low-rank increment terms $\bL_1,\ldots,\bL_p,\bK_1,\ldots,\bK_p$ in Condition (A1) is bounded by a fixed constant, the contraction rate can be improved by introducing a sparsity-inducing mechanism in the prior for $\bE$, for instance, as in Subsection~\ref{sec: prior}. Let $\bSigma_0$ have the modified Cholesky decomposition $(\bI-\bE_0)\bF_0 (\bI-\bE_0)^{\mathrm{T}}$. 

\begin{theorem} 
\label{thm:sparse rate}
If Conditions (A1)--(A5) hold, $r$ is a fixed constant, and the number of non-zero entries of $\bE_0$ is $s$, then the posterior contraction rate for $\bOmega$ at $\bOmega_0$ and for the autoregression parameters $\bA_1,\ldots,\bA_p$ at $\bA_{1,0},\ldots,\bA_{p,0}$ with respect to the Frobenius distance is $\sqrt{(d+s)(\log T)/T}$. 
\end{theorem}

The proofs for Theorem~\ref{thm:basic rate} and Theorem~\ref{thm:sparse rate} are given in the supplementary material section. 

\section{Differential Connectivity due to the Great Recession Using FRED-QD data}

The main objective for developing a flexible stationary graph estimation method for a VAR process is to evaluate how the connectivity between the economic time series in the FRED-QD database changed during the financial crisis of the Great Recession in the context of a stationary graph constructed based on the joint VAR model for the series. The simulation experiment in the next section provides credence to the claim that for graph estimation, there is benefit in explicitly modeling the temporal dependence structure of the sample. As mentioned earlier, a stationary vector autoregressive model is a reasonable model for the FRED-QD time series after performing the transformations suggested in \cite{mccracken2020fred}. Thus, we estimate a VAR model for the FRED-QD data and the stationary graph obtained from the associated precision matrix. 

As argued earlier, fitting two different stationary models to the data from periods before and after the recession will allow us to compare the presumably stationary and potentially different temporal dynamics of the two periods while circumventing the challenging task of accounting for all the changes and associated epochs during the nearly two and half years long period of the Great Recession. Performing a comparison of the stationary graphs between the before and after period in terms of a combined model is relatively straightforward within a Bayesian framework, where we can study the posterior of functions of parameters from the before and after VAR models to evaluate changes in the stationary graph.  We specifically study differential connectivity between the FRED-QD series due to the Great Recession in terms of the scale-free edge difference parameters, $\theta_g(i,j),$ described in \eqref{eq:edge_diff}. We thus fit our proposed model for each group $g$ and each data fragment from before and after the recession. We obtain the posterior samples of stationary precision matrices $\bOmega_{b,g}$ and $\bOmega_{a,g}$, before and after the period of great recession, respectively, and construct posterior samples of the scale-free edge difference parameters given in \eqref{eq:edge_diff}. As discussed in Section~\ref{sec:compuprec}, we set the maximum VAR order as $\min(p_m,T/2)$, where $p_m=10$ and all the $r_j$'s are initialized at 3. Then, at the 1000th iteration, we drop the columns in $\bL_j$'s with overall contribution less than 0.1 in the sum of squares. In that process, some of the higher-order $\bL_j$'s are dropped, reducing the overall order of the fitted VAR. The computational steps remained stable for higher choices of $p_m$ and $r_j$'s.  All the other tuning details described in Section~\ref{sec:compuprec} are left unchanged.

To designate the edge-difference parameter as significantly different from zero for any particular pair of economic indicators, and hence determine whether there has been a substantial impact of the great recession on that particular connection,  we compute an empirical 95\% credible interval for the parameter and check if zero is included in the interval. The particular edge in the stationary graph is inferred to have undergone significant changes during the Great Recession if the credible interval does not contain zero. 
Let $\{\bOmega^{(s)}_{b,g}\}_{s=1}^S$ be $S$ posterior samples of $\bOmega_{b,g}$ and similarly $\{\bOmega^{(s)}_{a,g}\}_{s=1}^S$ are posterior samples of $\bOmega_{a,g}$. For this exercise,  we obtain $S = 5000$ posterior samples for the before and after precision matrices after burning in the first 5000 samples, thereby obtaining 5000 posterior samples for each $\theta_g(i,j)$ parameter. Based on these samples, we compute the 95\% credible intervals for $\theta_g(i,j)$ as the empirical 2.5 and 97.5 quantiles of the sample of 5000 values. Since there are several hundred pairwise connections, we only report aggregate measures of changes for each of the economic groups in the FRED-QD data. For each group-specific graph, Table~\ref{tab:different} reports the percentage of edge-difference parameters with zeroes not included in their credible intervals. Larger values of this number imply a greater disruption of the interdependencies among the series within that group due to the financial crisis.
Table~\ref{tab:different} values are relatively higher than the ones reported in Table~\ref{tab:EDA}, showing that, unlike the analysis based on sparse VAR modeling,  the proposed method, based on Bayesian analysis of the difference in the two stationary graphs from before and after periods,  detects substantial changes in connectivity within most of the groups. Some of the groups with larger differences in Table~\ref{tab:EDA}, such as Stock Markets and Non-Household Balance Sheets, retain higher differences in their stationary precisions in Table~\ref{tab:different}. However, the inventories, Orders, and Sales series, which showed significantly large differences under he sparse VAR analysis, show only a moderate amount of changes under the proposed detection scheme. Several groups show large differences in their stationary precisions when compared to the estimates in Table~\ref{tab:EDA}, e.g., NIPA, Employment, Housing, Earnings and Productivity, Interest Rates, and Money and Credit. The results are consistent with the literature showing substantial changes in the relationship between series from these groups. 
The groups indicating smaller amount of changes due to the recession under the proposed Bayesian scheme are Industrial Production and Exchange rates. While overall industrial production is significantly affected by the economic downturn, and some relationships are permanently changed, there has been evidence of economic resistance in some parts of the sector that left the interdependencies more comparable between the before and the after recession periods. This is consistent with the literature showing some degree fo resilience and quick recovery in the industrial production series; \citep{kitsos2023industrial}. 
A similar picture emerges in the exchange rate group, with only four series. The relatively lower value for the proportion of changes in the exchange rate series may be due to the stability of the safe-haven currencies like the Japanese Yen or the Swiss Franc, or due to the controlled exchange rate dynamics of the Chinese Yuan achieved through stringent financial controls. 

To further scrutinize the changes estimated in the graphs for each group, and thereby evaluate the impact of the Great Recession on each of the subgroups, we look at the individual change graph, i.e., a graph of the edges deemed significantly perturbed by the Bayesian method. 
Figure~\ref{fig: changegraph} shows edges $(i,j)$ within each network where the credible intervals of the scale-free edge-difference parameters do not include zero. These edges represent the interdependencies that changed during the great recession. 

In each group, several edges are determined to have changed during the Great Recession. This is consistent with the existing knowledge. However, there are series (nodes) for which only a few edges show changes between the before and after recession periods.  

In the NIPA group, the Government Consumption, Expenditure, and Gross Investment series (node 13) has only two edges (with export and import series) that show significant change due to the recession. 
The government consumption and spending worked as a stabilizer during the financial crisis, cushioning the effect of the economic downturn \cite{ramey2019ten}. Governments increased spending to stimulate the economy and ensure faster recovery of some of the other economic indicators. The series showed resilience to the shock, and the interventions by the government helped maintain or recover the relationship with other indicators. However, the interdependency between government consumption, expenditure, and investment, and import/export shifted significantly as the interventions that controlled the domestic economy could not provide control internationally. The collapse of global economic activity led to a sharp decline in trade. Emphasis on boosting  domestic demand and government intervention  promoted a notable decrease in aggregate expenditure, particularly in trade-intensive durable goods, and had a significant effect on imports and exports \cite{bems2013great} 

The housing market bubble fueled by subprime mortgage rates before 2007 and its subsequent collapse is generally attributed to as one of the triggers for the Great Recession; \cite{gorton2009subprime}. However, not all housing sectors were impacted equally. During the Great Recession, single-family housing starts were severely affected. Multifamily housing experienced lower levels of rent decline and shorter recovery  period until rents attained their pre-recession peaks \cite{wang2019housing}
Although multifamily housing starts declined during the Great Recession, the sector's performance in terms of rental rates and recovery time showcased its relative resilience compared to other real estate sectors. This is bourne out in the relatively fewer interdependencies between the multifamily housing and starts and other housing series that were marked as showing significant changes after the recession. 

On the other hand, series such as consumer-related PCE, particularly those related to the health and education sector employment, showed resilience during the recession \cite{Barello}, but their interdependencies with other  series seem to have been disrupted significantly by the financial crisis. This shows that while some sectors were somewhat recession proof, their interrelation with other sectors were not.


\begin{table}[htbp!]
\centering
\caption{Proportions of entries in the stationary precisions where the 95\% credible interval for the associated edge-difference parameter does not include zero.}
\begin{tabular}{lc}
  \hline
 & Proportion of edges that changed \\ 
  \hline
National income and product accounts & 0.72 \\ 
  Industrial Production & 0.13 \\ 
   Employment  & 0.83 \\ 
   Unemployment and Participation & 0.33 \\ 
  Housing & 0.41 \\ 
  Inventories, Orders, and Sales & 0.31 \\ 
   Prices (PCE) & 0.64 \\
   Prices (CPI) & 0.25\\
  Earnings
and Productivity & 0.55 \\ 
  Interest Rates & 0.43 \\ 
  Money and Credit & 0.52 \\ 
  Household Balance Sheets & 0.62 \\ 
  Exchange
Rates & 0.17 \\ 
  Stock Markets & 0.57 \\ 
  Non-Household Balance Sheets & 0.78 \\ 
   \hline
\end{tabular}
\label{tab:different}
\end{table}


\begin{figure}[htbp!]
\centering
\subfigure{\includegraphics[width = 0.45\textwidth, trim=1cm 2cm 0cm 0cm, clip=true]{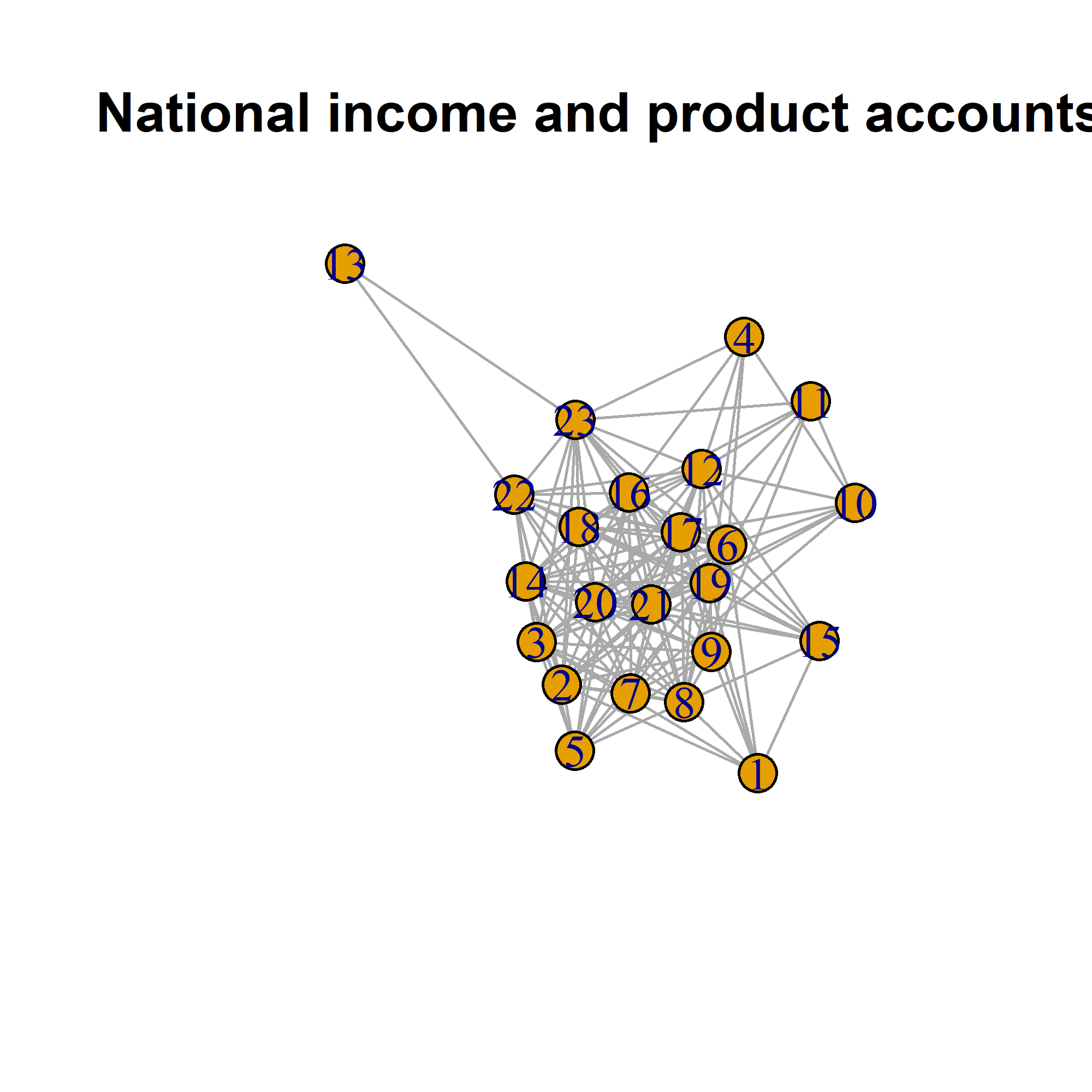}}
\subfigure{\includegraphics[width = 0.45\textwidth, trim=1cm 2cm 0cm 0cm, clip=true]{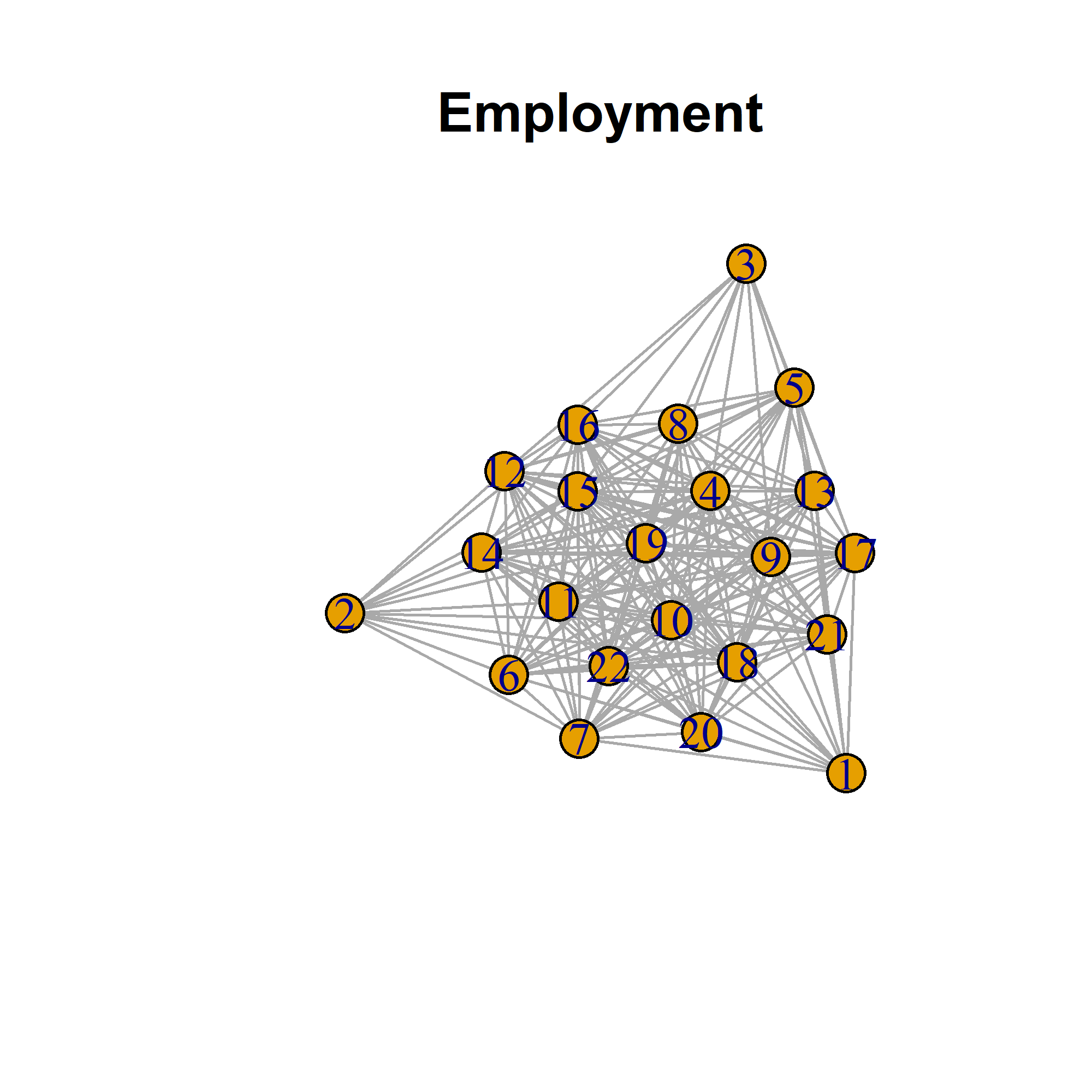}}
\subfigure{\includegraphics[width = 0.45\textwidth, trim=1cm 2cm 0cm 0cm, clip=true]{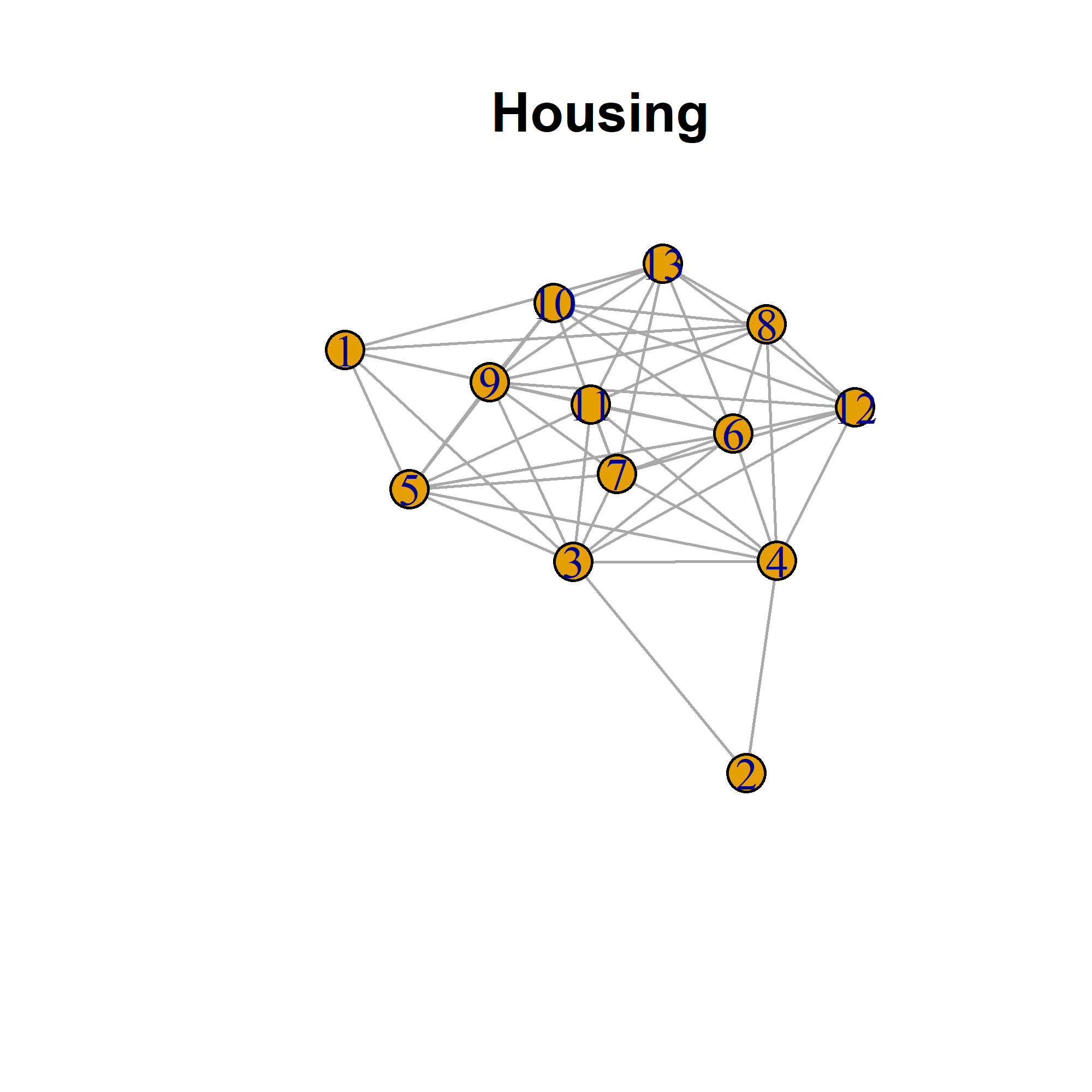}}
\subfigure{\includegraphics[width = 0.45\textwidth, trim=1cm 2cm 0cm 0cm, clip=true]{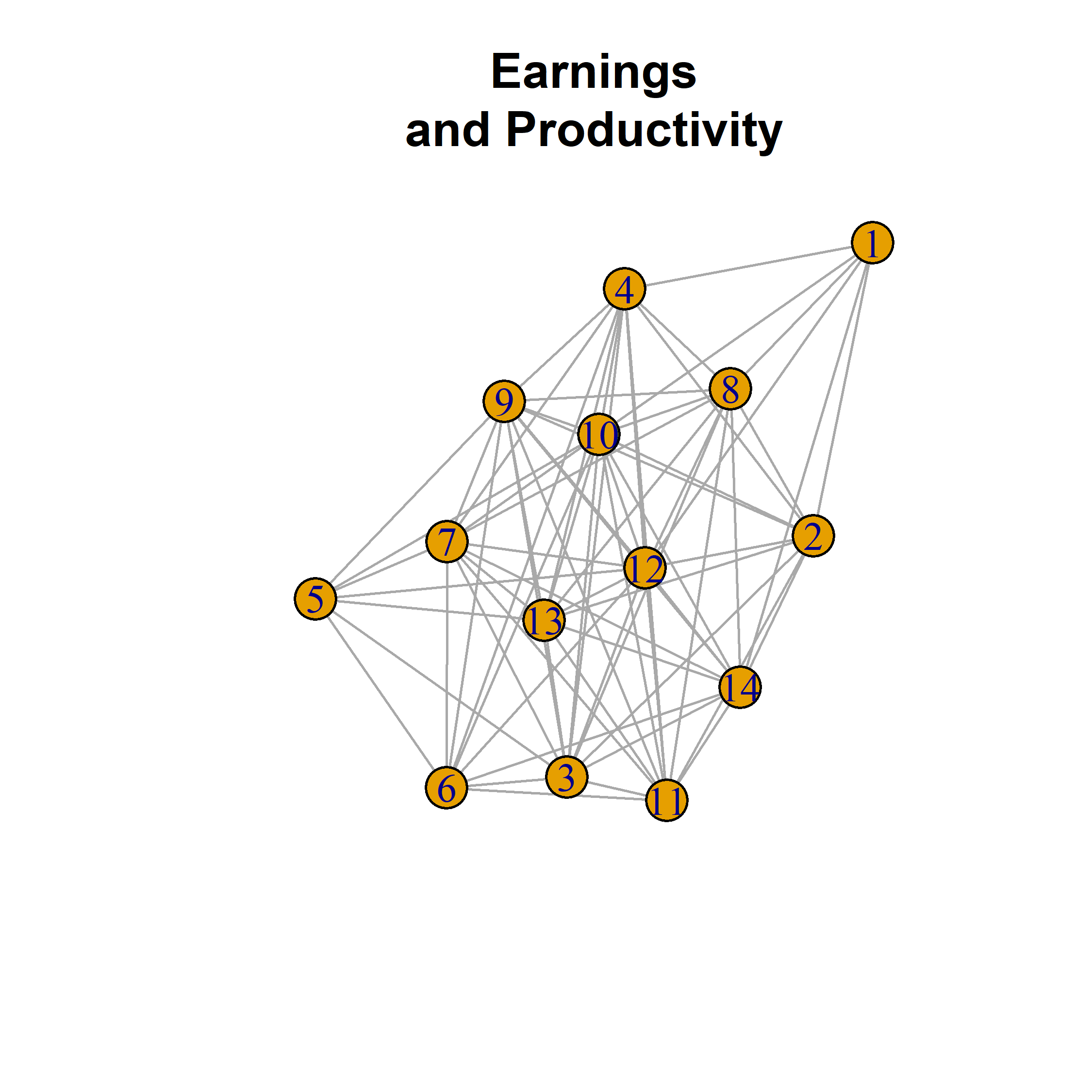}}
\subfigure{\includegraphics[width = 0.45\textwidth, trim=1cm 2cm 0cm 0cm, clip=true]{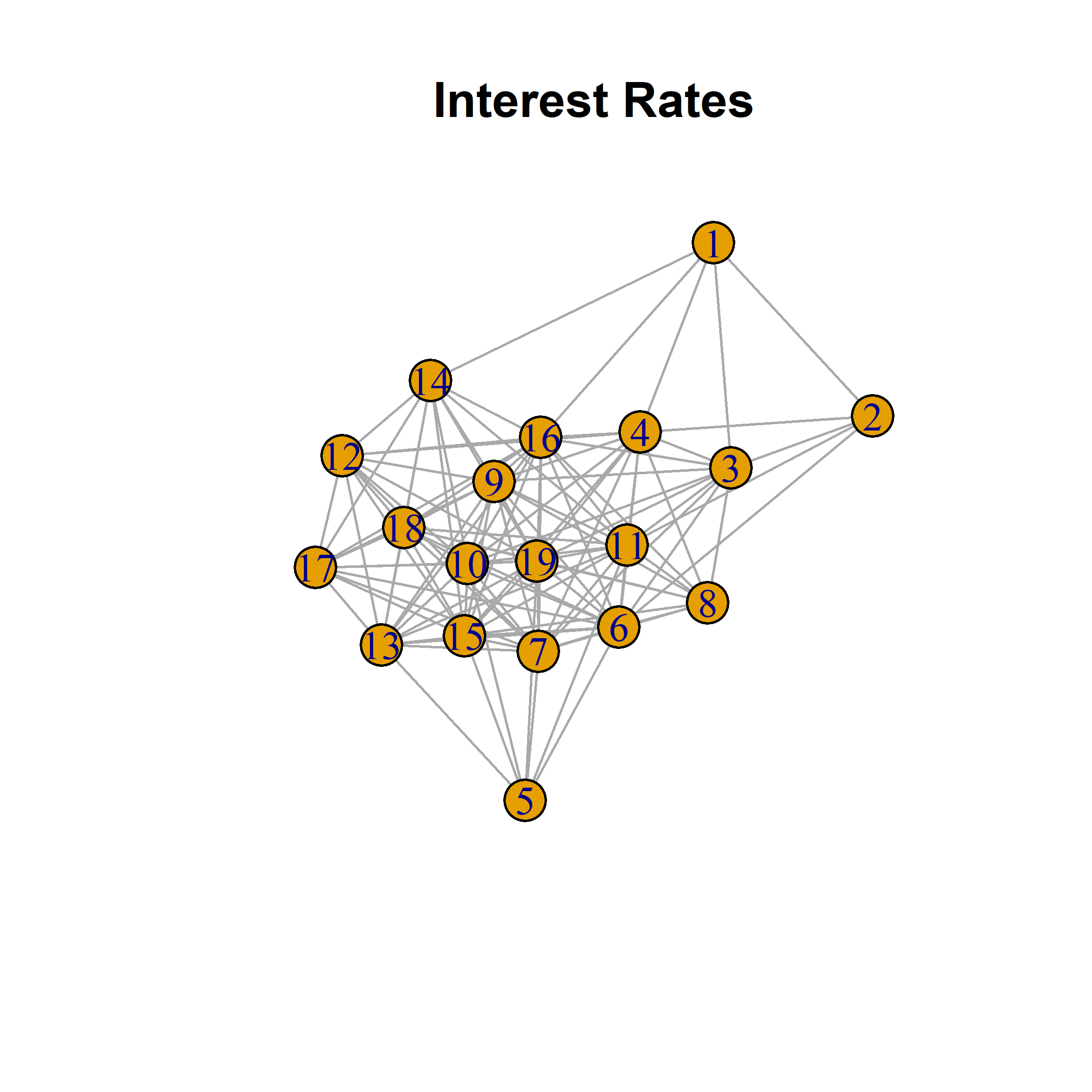}}
\subfigure{\includegraphics[width = 0.45\textwidth, trim=1cm 2cm 0cm 0cm, clip=true]{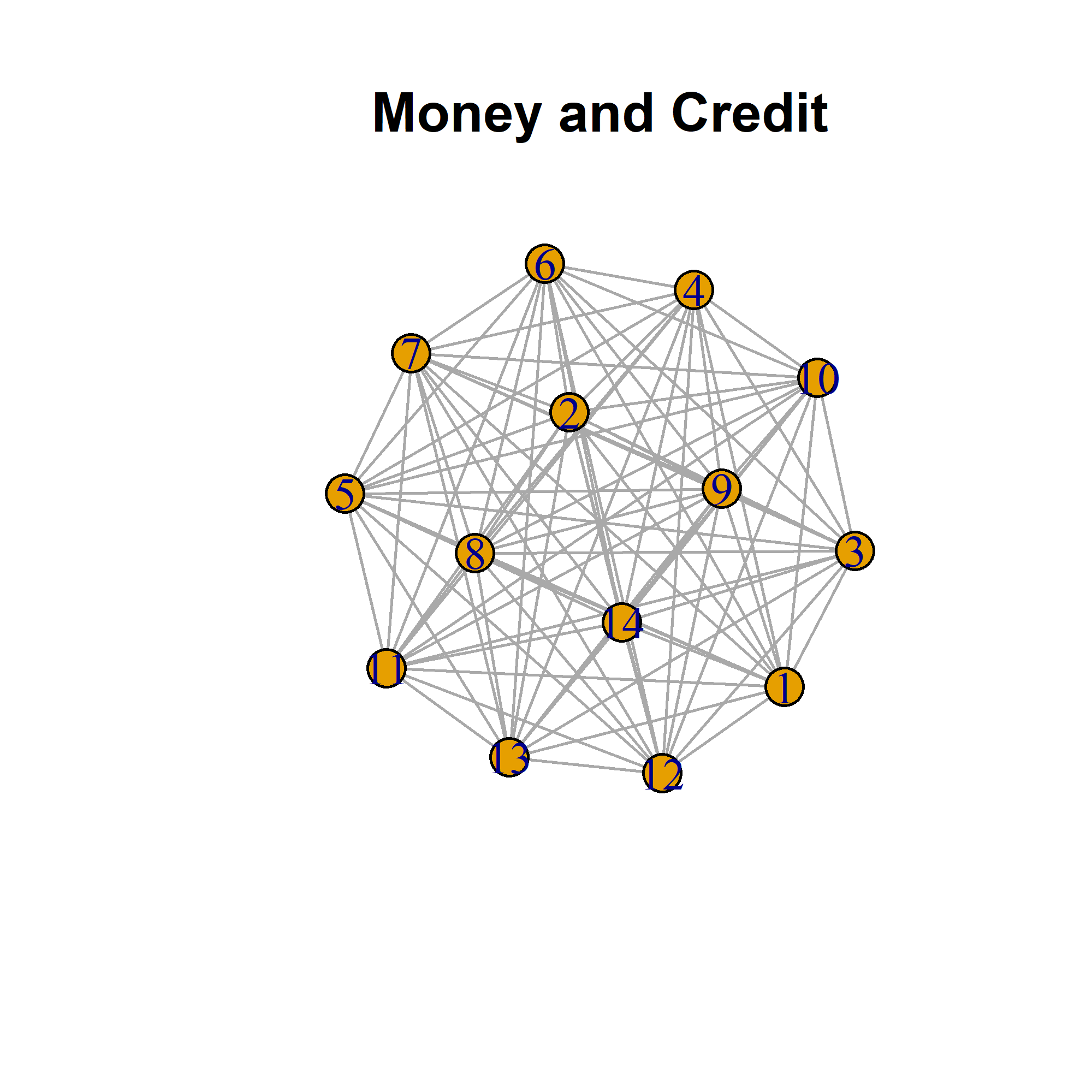}}
\caption{Estimated change networks. The node numbers correspond to the variable names as described in \cite{mccracken2020fred}.}
\label{fig: changegraph}
\end{figure}

\section{Simulation}
\label{sec:simu}

To analyze the performance of the proposed method when the ground truth is known, we perform further simulation experiments. We primarily study the impact of the sample size and the sparsity level on estimation accuracy for the stationary precision matrix for a first order VAR process. We compare the accuracy of estimating the stationary precision matrix and the associated partial correlation graph for the proposed method with two other popular graphical model estimation methods, the Gaussian Graphical Model (GGM) and the Gaussian Copula Graphical Model (GCGM).

The data are generated from a VAR(1) model with a fixed marginal precision matrix, $\bOmega_1$. 
Specifically, the data $\bX_1, \ldots, \bX_T$ are generated from the model 
\be
\bX_t = \bA_1 \bX_{t-1} + \bZ_t
\label{eq:model_simu}
\ee
with the stationary precision matrix chosen as $\bOmega  = \bOmega_1$ and the parameters associated with the update of the conditional precision, $(\bL_1, \bK_1)$, are chosen as random $30\times 30$ matrices with entries generated independently from $\mathrm{N}(0, 2.5^2).$ The VAR coefficient $\bA_1$ and the innovation variance $\bSigma_1$ are solved from the specified parameters, $\bOmega_1, \bL_1, \bK_1,$ following the steps defined in section~\ref{sec:likelihood_comp}. The initial observation, $\bX_1$, is generated from the stationary distribution, and subsequent observations are generated via the iteration in \eqref{eq:model_simu}. Two different sample sizes are used, $T = 40$ and $T = 60.$   For the sparse precision matrix, we use two different levels of sparsity, $15\%$ and $25\%$. The sparse precision $\bOmega_1$ is generated using the following method.
\begin{enumerate}
    \item[(1)] Adjacency matrix: three small-world networks are generated, each containing 10 disjoint sets of nodes with two different choices, 5 and 10,  for {\tt nei} variable in {\tt sample\_smallworld} of {\tt igraph} \citep{igraphR}. Then,  nodes are randomly connected across the small worlds with probability $q$. The parameters in the small world distribution and $q$ are adjusted to attain the desired sparsity levels in the precision matrix. 
    \item[(2)] Precision matrix: Using the above adjacency matrix, the precision matrix $\bOmega$ is generated from a G-Wishart with scale equal to 6, and off-diagonal entries smaller than 1 in magnitude are set to zero.
\end{enumerate}
The simulated model is a full-rank 30-dimensional first-order VAR. Thus, the proposed method of fitting low-rank matrices for the conditional precision updates is merely an approximation and is fitting a possibly misspecified model. Posterior computation for the proposed method uses the steps in section~\ref{sec:compuprec}, where a higher-order VAR with low-rank updates is used to fit the data. The upper bound on the order of the progress is chosen to be $p_m = 10.$ Thus, we fit a misspecified model and evaluate the robustness characteristics of the method along with estimation accuracy.
Table~\ref{tab:predsimu11} shows the median MSE for estimating the entries of the 30$\times$30 precision matrix for model \eqref{eq:model_simu}. The proposed method has higher estimation accuracy than the GCGM and GGM methods that ignore the temporal dependence. Explicit modeling of the dependence, even under incorrect rank specification, is beneficial and can provide substantial gains in MSE. The comparative performance is better for the proposed method when the sparsity level is $25\%.$ The Gaussian Graphical Model seems sensitive to the specification of the simulation model. The MSE for the GGM increases with increasing sample size, a phenomenon presumably an artifact of ignoring the dependence in the sample. When the simulated parameters of the conditional precision update, $\bL_1, \bK_1,$ are generated independently from $\mathrm{N}(0, 1)$ (not reported here), the MSE for the GGM decreases monotonically with the sample size. This is because  $\mathrm{N}(0, 1)$ being more concentrated around zero, makes the simulated model closely mimic an independent model. The MSEs for all the methods go up with a decreased sparsity level. A higher number of nonzero entries in the precision matrix makes the problem harder for sparse estimation methods and leads to higher estimation errors. The results for the proposed method presented in this section are based on rank-1 updates. The results for rank-3 updates (not shown here)  show only marginal improvement over rank-1 updates for precision matrix estimation. The rank-3 updates are computationally more expensive. 

We also investigate how well the partial correlation graph is estimated under different methods when the samples are dependent and arise from a VAR process. 
Figure~\ref{fig::simumagARMA} shows the ROC curves for the three estimation methods for different settings; two different sample sizes $T = 40, 60$ and two different sparsity levels $15\%, 25\%.$ In every case, the proposed method performs better than GCGM and GGM, with the GCGM performing better than the GGM. The AUC values for the proposed method are bigger than those for the GCGM and GGM.

\begin{table}[ht]
\centering
\caption{Median estimation MSE in estimating the precision matrix of dimension $30\times 30$ when the data is generated from VAR(1).}
\begin{tabular}{|c|ccc|ccc|}
  \hline
    Time points &\multicolumn{3}{|c|}{15\% Non-zero} & \multicolumn{3}{|c|}{25\% Non-zero} \\
    \hline
  \hline
 & Causal VAR & GCGM & GGM & Causal VAR & GCGM & GGM \\ 
  \hline
40 & 7.01 & 6.30 & 21.87 & 8.95 & 14.92 & 17.83 \\ 
  60 & 5.37 & 6.46 & 34.33 & 7.18 & 10.68 & 34.10 \\ 
   \hline
\end{tabular}
\label{tab:predsimu11}
\end{table}

\begin{figure}[htbp]
\centering
\subfigure{\includegraphics[width = 0.42\textwidth]{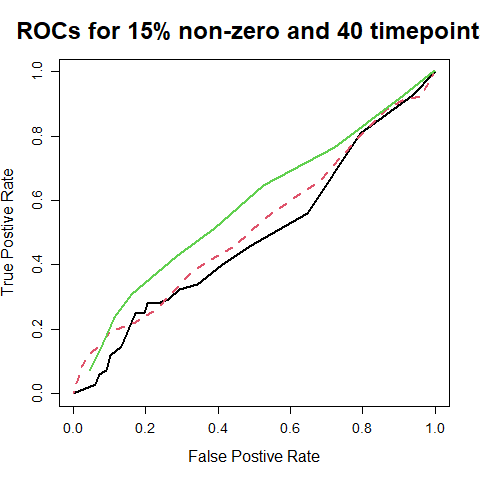}}
\subfigure{\includegraphics[width = 0.42\textwidth]{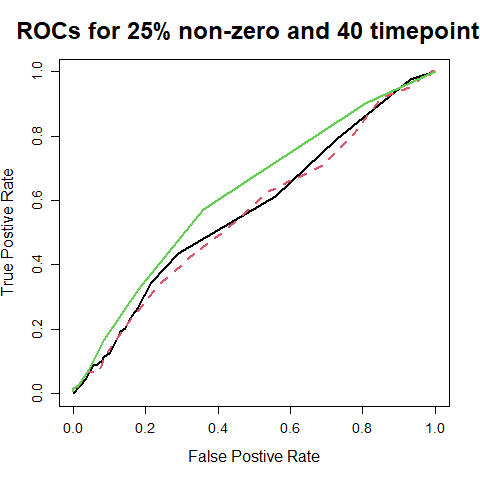}}
\subfigure{\includegraphics[width = 0.42\textwidth]{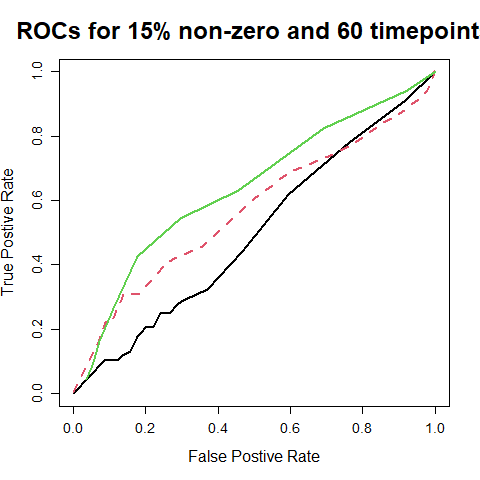}}
\subfigure{\includegraphics[width = 0.42\textwidth]{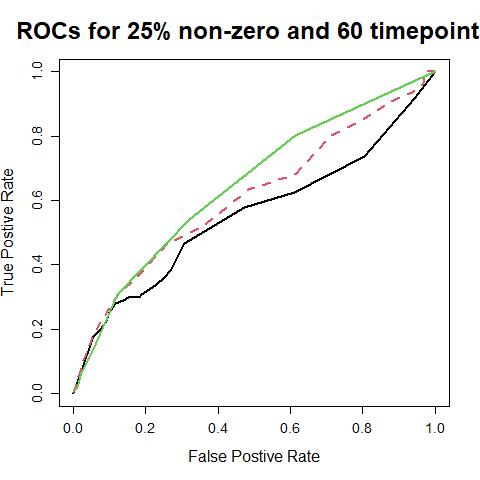}}
\caption{ROC comparison for different cases: Black = GCGM, Red = GGM, Green = Causal VAR when the true data is generated using VAR(1) model.}
\label{fig::simumagARMA}
\end{figure}

\section{Discussion}

In this article, we propose a new Bayesian method for estimating the stationary precision matrix of a high-dimensional VAR process under stability constraints and subsequently estimating the contemporaneous stationary graph for the components of the VAR process. The method is then used to evaluate possible changes during the financial crisis of the Great Recession of 2007--2009 in terms of the interdependencies among various seasonally adjusted quarterly time series collected in the FRED-QD database. 
The method has several natural advantages that will be useful in analyzing similar time series databases. The methodology introduces a parameterization that allows for the fast computation of the stationary likelihood of a reduced-rank high-dimensional Gaussian VAR process. The new parameterization also gives a natural way to directly model the sparse stationary precision matrix of a high-dimensional VAR process, which is the quantity needed to construct contemporaneous stationary graphs for the VAR time series. Another distinct advantage of the proposed method over popular competitors is that the proposed methodology uses causality of the VAR process as a hard constraint so that any estimated VAR will be restricted to the causal VAR space. This property is helpful in applications, such as stable dynamical system identification, where stationarity of the process is paramount. 
Most estimation methods for high-dimensional VAR processes fail to impose causality on the solution, making estimating the stationary precision matrix less meaningful.  In our preliminary analysis of the FRED-QD data using sparse VAR, the estimated VAR polynomials are often unstable; the estimated VAR process is non-causal. However, estimates based on the proposed parameterization are always stable by construction.

We also establish posterior consistency for our Bayesian estimation scheme when priors are defined through the proposed parameterization. The consistency of $\bOmega$ posterior indicates accurate inference on change parameters $\theta_g(i,j)$  to determine the effect of the Great Recession on the individual edges. In the future, we plan to investigate graph consistency for the partial correlation graph obtained by thresholding the entries of the stationary precision matrix of a high-dimensional VAR.

The focus of the present article is on the estimation of the stationary precision matrix of a high-dimensional stationary VAR. Hence, the direct parameterization of the precision, independent from other parameters involved in the temporal dynamics, is critical. However, a similar scheme can provide a direct parameterization of the VAR autocovariance matrices and, hence, the VAR coefficients under reduced rank and causality constraints for a high-dimensional VAR. This is part of a future investigation. Another possibility is to estimate directional relationships among the variables in each group of the FRED-QD data through directed acyclic graph (DAG)-based characterization.

Finally, the methodology developed here was motivated by the desire to estimate and compare the stationary graphs between the economic indicators before and after the Great Recession. But the methods are broadly applicable for doing a similar before-and-after analysis of VAR time series graphs in the context of other extreme events, such as the COVID-19 pandemic, provided the event is a prolonged and complex event that would be hard to analyze using change point analysis. In addition, the methods apply to any problem where the main interest is the estimation of the partial correlation graph of a VAR process.


\section*{Funding}
The authors would like to thank the National Science Foundation collaborative research grants DMS-2210280 (Subhashis Ghosal) / 2210281 (Anindya Roy) / 2210282 (Arkaprava Roy).

\Appendix

\renewcommand{\theequation}{A.\arabic{equation}}
\renewcommand{\thetable}{A.\arabic{table}}

\setcounter{equation}{0}
\setcounter{table}{0}

\section*{Appendix: Identification of Parameters}

In the  recursive algorithm described above, the parameters $\bOmega, (\bL_1,\bK_1), \ldots, (\bL_p, \bK_p)$  are mapped to the modified parameters $\bOmega,$ $ (\bU_1, \bV_1), \ldots, (\bU_p, \bV_p)$. The mapping is not one-to-one since each $\bV_j$ has the restriction $\bV_j^{\mathrm{T}}\bD_{j-1}^{-1} \bV_j = \bI$ while the basic parameters $\bK_j$ do not have any restrictions. The intermediate parameters in the mapping, $\bOmega, (\bU_1, \bV_1), \ldots, (\bU_p, \bV_p)$ are  identifiable only up to rotations. Consider the equivalence classes of pair of $d\times r_j$ matrices $(\bU, \bV)$ defined through the relation $(\bU, \bV) \equiv (\bU\bQ, \bV\bQ)$ for any $r_j\times r_j$ orthogonal matrix $\bQ$. Let for each $r_j$, $\mcC_{r_j}$ be the set of such equivalence classes. Also let $\mcS^{++}_{d,s}$ be a subset of the positive definite cone $\mcS^{++}_{d}$ with specified sparsity level $s$. The following proposition shows that there is a bijection from the set $\mcS^{++}_{d,s} \times\mcC_{r_1} \times \cdots \times \mcC_{r_p}$ to the causal VAR$(p)$ parameter space defined by the VAR coefficients $\bA_1, \ldots, \bA_p$ and innovation variance $\bSigma$,  where polynomials with coefficients $\bA_1, \ldots, \bA_p$ satisfy low-rank restrictions and the stationary precision satisfies $s$-sparsity. The definition of sparsity is made clearer in the prior specification section. For simplicity, we use a mean-zero VAR($p$) process. Before formally stating the identification result, we state and prove a short lemma used to show identification.

\begin{lemma}
Given an $n \times n$ matrix $\bA$ of rank $1 \leq r \leq n$ for some positive integer $n$, and a positive definite $n\times n$ matrix $\bB$, there exists $n \times r$ matrices $\bC$ and $\bR$ such that $\bA = \bC\bR^{\mathrm{T}}$, $\mathrm{rank}(\bC) = \mathrm{rank}(\bR) = r$ and $\bR^{\mathrm{T}}\bB\bR = \bI$. 
\lb{lem:rankfac}
\end{lemma}
\begin{proof}
 Given a rank factorization $\bA = \bU\bV$  and the symmetric square root $\bB^{1/2}$, let $\bS = \bQ\bZ$ be the Q-R decomposition of $\bS = \bB^{1/2}\bV^{\mathrm{T}}$ where $\bQ$ is $n\times r$ semi-orthogonal and $\bZ$ is $r\times r$ nonsingular matrix. Then 
 $$\bA = \bU\bV = \bU\bV\bB^{1/2}\bB^{-1/2} = \bU\bV\bB^{1/2}\bB^{-1/2} = \bU\bZ^{\mathrm{T}}\bQ^{\mathrm{T}}\bB^{-1/2} = \bC\bR,$$ 
 where $\bC = \bU\bZ^{\mathrm{T}}$ and $\bR = \bB^{-1/2}\bQ$. For this choice, $\bR^{\mathrm{T}} \bB \bR = \bI.$
\end{proof}

\begin{proposition}
Let $\bOmega >0$ be a $d\times d$ positive definite matrix, which is $s$-sparse. Let $(\bU_1, \bV_1), \ldots$, $(\bU_p, \bV_p)$ be given full column rank matrix pairs, of order  $d\times r_1, \dots, d\times r_p$, respectively. Then 
there is a unique causal $\mathrm{VAR}(p)$ process with stationary variance matrix $\bGamma(0) = \bOmega^{-1}$ and autcovariances $\bGamma(1), \ldots, \bGamma(p)$ (and hence the coefficients $\bA_1, \ldots, \bA_p$ as in the Yule-Walker solution) uniquely determined recursively from  $\bW_j = \bU_j\bV_j^{\mathrm{T}}$.  The associated increments in the conditional precision matrices $\bC_j^{-1} - \bC_{j-1}^{-1}$ will be of rank $r_j.$

Conversely, let $\bX_t$ be a zero mean causal $\mathrm{VAR}(p)$ process such that $\bOmega = \bGamma(0)^{-1}$  is $s$-sparse and the increments in the conditional precision $\bC_j^{-1} - \bC_{j-1}^{-1}$ are of rank $r_j.$ Let $\bW_j$ be defined recursively as in \eqref{eq:Wj}. Then there are $d\times r_j$ matrices $(\bU_j, \bV_j)$, determined up to rotation, such that $\bW_j = \bU_j \bV_j^{\mathrm{T}}$ and  $\bV_j^{\mathrm{T}} \bD_{j-1}^{-1} \bV_j = \bI$ for each $j$; 
\label{prop:identi}
\end{proposition}

\begin{proof}
  Given the parameters $\bOmega, (\bU_1,\bV_1), \ldots, (\bU_p, \bV_p)$, the autocovariances are obtained as $
  \bGamma(0)  = \bOmega^{-1}$ and $\bGamma(j)^T = \bU_j\bV_j^{\mathrm{T}} + \bxi_{j-1}\bGamma^{-1}_{j-2}\bkappa_{j-1}$ with $\bxi_j$ and $\bkappa_j$ determined recursively. The block Toeplitz matrix $\bUpsilon_p$ will be positive definite since the associated $\bC_{j-1} - \bC_j = \bU_j\bU_j^{\mathrm{T}}$
  are all nonnegative definite. Hence, the associated VAR polynomial whose coefficients are obtained as $(\bA_1, \ldots, \bA_p) = \bxi_p^{\mathrm{T}}\bUpsilon_{p-1}^{-1}$ will be Schur stable. Also, $\bSigma = (\bOmega + \sum_{j=1}^p \bC_j^{-1}\bU_j(\bI + \bU_j^T\bC_j^{-1}\bU_j)^{-1}\bU_j^{\mathrm{T}}\bC_j^{-1})$ will be positive definite, thereby making the associated VAR process a causal process.

  For the converse, note that once $\bW_j$ are obtained, by Lemma~\ref{lem:rankfac} they can be factorized as $\bW_j = \bU_j \bV_j^{\mathrm{T}}$ such that $\bV_j$ satisfy $\bV_j^{\mathrm{T}} \bD_{j-1} \bV_j = \bI.$ 
  However, the pairs $(\bU_j, \bV_j)$ are unique only up to rotation with orthogonal matrices of order $r_j \times r_j.$
  The increments $\bC_{j-1} - \bC_j$ are equal to $\bW_j \bD_{j-1}^{-1}\bW_j^{\mathrm{T}}$ and hence are nonnegative definite of rank $r_j$. Thus,  the increments 
$\bC_{j}^{-1} - \bC_{j-1}^{-1}$ are also nonnegative definite of rank $r_j$. 
\end{proof}


\section*{Supplementary materials}

\setcounter{section}{0}

\makeatletter
\renewcommand \thesection{S\@arabic\c@section}
\renewcommand\thetable{S\@arabic\c@table}
\renewcommand \thefigure{S\@arabic\c@figure}
\makeatother

\section{Proof of the Main Theorems}
\label{sec:proof main}

We follow arguments in the general posterior contraction rate result from Theorem~8.19 of \cite{GhosalBook} on the joint density of the entire multivariate time series by directly constructing a likelihood ratio test satisfying a condition like (8.17) of \cite{GhosalBook}. We also use the simplified prior concentration condition (8.22) and global entropy condition (8.23) of \cite{GhosalBook}. It is more convenient to give direct arguments than to fit within the notations of the theorem. 

\subsection{Prior concentration: pre-rate}
\label{prior concentration}

Recall that the true parameter is denoted by $\bzeta_0$, and the corresponding dispersion matrix is $\bUpsilon_{T,0}$. We shall obtain the pre-rate $\bar\epsilon_T$ satisfying the relation 
\begin{align}
\label{KL proprty}
-\log \Pi (K(Q_{\bzeta_0}, Q_{\bzeta})\le T\bar\epsilon_T^2,  V(Q_{\bzeta_0}, Q_{\bzeta})\le T\bar\epsilon_T^2)\lesssim T\bar\epsilon_T^2, 
\end{align}
where $K$ and $V$, respectively, stand for the Kullback-Leibler divergence and Kullback-Leibler variation. To proceed, we show that the above event contains $\{\bzeta: \|\bzeta-\bzeta_0\|_\infty\le \eta_T\}$ for some $\eta_T$ and estimate the probability of the latter, when $\eta_T$ is small. 

 We note that by Relation (iii) of Lemma~\ref{norm comparison}, $\|\bUpsilon_{T,0}\|^2_{\mathrm{op}}\le \|\bSigma_0\|^2_{\mathrm{op}}$, while by Relation (i), $\|\bUpsilon_{T,0}^{-1}\|^2_{\mathrm{op}}\le \|\bSigma_0^{-1}\|^2_{\mathrm{op}}$, which are bounded by a constant by Condition (A5). 
 Next observe that by Relation (iii) applied to $\bUpsilon_T$, $\op{\bUpsilon_T}\le \op{\bSigma}\le \op{\bOmega^{-1}}=\|1/\bm{f}\|_\infty$, where $1/\bm{f}$ refers to the vector of the reciprocals of the entries of $\bm{f}$. Since the entries of $\bm{f}_0$ lie between two fixed positive numbers, so do the entries of $\bm{f}$ and $1/\bm{f}$ when $\|\bm{f}-\bm{f}_0\|_\infty$ is small. Then it is immediate that (see, e.g., Lemma~4 of \cite{roy2024bayesian}) that $\Frob{\bOmega-\bOmega_0}\lesssim d\eta$. Applying Lemma~\ref{lem:around_the_truth}, we obtain the relation $\Frob{\bUpsilon_{T,0}-\bUpsilon_T}\lesssim drT^2 \eta$, which is small for $\eta$ sufficiently small. Thus, $\Frob{\bUpsilon_{T,0}^{-1/2}(\bUpsilon_{T,0}-\bUpsilon_T) \bUpsilon_{T,0}^{-1/2}}$ is also small. Since the operator norm is weaker than the Frobenius norm, this implies all eigenvalues of $\bUpsilon_{T,0}^{-1/2}(\bUpsilon_{T,0}-\bUpsilon_T) \bUpsilon_{T,0}^{-1/2}$ are close to zero. 
Now, expanding the estimates in Parts (ii) and (iii) of Lemma~\ref{normal divergences} in a Taylor series, the Kullback-Leibler divergence and variation between $Q_{\bzeta_0}$ and $Q_{\bzeta}$ are bounded by a multiple of 
$\Frob{\bUpsilon_{T,0}^{-1/2}(\bUpsilon_{T,0}-\bUpsilon_T) \bUpsilon_{T,0}^{-1/2}}^2\le d^2r^2 T^4 \eta^2= T \bar\epsilon_T^2$ if $\eta$ is chosen to be $d^{-1} T^{-3/2} \epsilon_T$. The prior probability of $\|\bzeta-\bzeta_0\|_\infty \le \eta$ is, in view of the assumed a priori independence of all components of $\bzeta$, bounded below by $(\bar c \eta)^{\mathrm{dim}(\bzeta)}$ for some constant $\bar c>0$. 
Hence $-\log \Pi (\|\bzeta-\bzeta_0\|_\infty \le \eta)\lesssim (d^2+d+2pdr)\log (1/\eta)\lesssim d^2 \log(1/\eta)$. Thus \eqref{KL proprty} for $T\bar \epsilon_T^2\asymp d^2 \log T$, that is for $\bar\epsilon_T\asymp d \sqrt{(\log T)/T}$, or any larger sequence.

\subsection{Test construction}

Recall that the true parameter is denoted by $\bzeta_0$, and the corresponding dispersion matrix is $\bUpsilon_{T,0}$. Let $\bzeta_1$ be another point in the parameter space such that the corresponding dispersion matrix $\bUpsilon_{T,1}$ satisfies $\Frob{\bUpsilon_{T,1}-\bUpsilon_{T,0}}>\sqrt{T} \epsilon_T$, where $\epsilon_T=M\bar\epsilon_T$ is a large constant multiple of the pre-rate $\bar{\epsilon}_T$ obtained in Subsection~\ref{prior concentration}. We first obtain a bound for the type I and type II error probabilities for testing the hypothesis $\bzeta=\bzeta_0$ against $\bzeta=\bzeta_1$. 

Let $\phi_T=\mathbbm{1}\{Q_{\bzeta_1}/Q_{\bzeta_0}>1\}$ stand for the likelihood ratio test. Then, by the Markov inequality applied to the square root of the likelihood ratio, both error probabilities are bounded by $e^{-R(Q_{\bzeta_1},Q_{\bzeta_0})}$, where $R$ stands for the Reyni divergence $-\log \int \sqrt{Q_{\bzeta_1} Q_{\bzeta_0}}$. 

Let $\rho_1,\ldots,\rho_p$ stand for the eigenvalues of $\bUpsilon_{T,0}^{-1/2}(\bUpsilon_{T,1}-\bUpsilon_{T,0})\bUpsilon_{T,0}^{-1/2}$. By Lemma~\ref{normal divergences}, the Reyni divergence is given by $\frac14 \sum_{j=1}^T [2\log (1+\rho_j)-\log(1+\rho_j)]$. By arguing as in the proof of Lemma~1 of \cite{roy2024bayesian}, each term inside the sum can be bounded below by $c_2 \min(\rho_j^2,c_1)$ for some constants $c_1,c_2>0$, where $c_1$ can be chosen as large as we like at the expense of making $c_2$ appropriately smaller. Under Condition (A4), using Relation (iii) of Lemma~\ref{norm comparison}, we obtain that the $\rho_j$ are bounded by some constant not changing with $T$. Thus, with a sufficiently large $c_1$, the minimum operation in the estimate is redundant; that is, the Reyni divergence is bounded below by 
$\frac14 c_2 \Frob{\bUpsilon_{T,0}^{-1/2}(\bUpsilon_{T,1}-\bUpsilon_{T,0})\bUpsilon_{T,0}^{-1/2}}^2\ge c_2' T\epsilon_T^2$ for some constant $c_2'>0$. This follows since $\Frob{\bUpsilon_{T,0}^{-1/2}(\bUpsilon_{T,1}-\bUpsilon_{T,0})\bUpsilon_{T,0}^{-1/2}}\ge \Frob{\bUpsilon_{T,1}-\bUpsilon_{T,0}}/\|\bUpsilon_{T,0}\|^2_{\mathrm{op}}$ and by Relation (iii) of Lemma~\ref{norm comparison}, $\|\bUpsilon_{T,0}\|^2_{\mathrm{op}}\le \|\bOmega^{-1}_0\|^2_{\mathrm{op}}$, which is bounded by a constant by Condition (A5). Thus both error probabilities are bounded by $e^{-c_2' T \epsilon_T^2}$ for some $c_2'>0$. 

Now let $\bzeta_2$ be another parameter with the associated dispersion matrix $\bUpsilon_{T,2}$ such that $\Frob{\bUpsilon_{T,2}-\bUpsilon_{T,1}}<\delta_T$, where $\delta_T$ is to be chosen sufficiently small. Then, using the Cauchy-Schwarz inequality and Part (iv) of Lemma~\ref{normal divergences}, the probability of type II error of $\phi_T$ at $\bzeta_2$ is bounded by 
\begin{align}
    \EE_{\bzeta_2}(1-\phi_T) &\le [\EE_{\bzeta_1}(1-\phi_T)]^{1/2} [\EE_{\bzeta_1} (Q_{\bzeta_2}/Q_{\bzeta_1})^2 ]^{1/2} \nonumber\\ 
    &\le \exp\{-c_2'T\epsilon_T^2/2+\op{\bUpsilon_{T,1}^{-1}}^2\Frob{\bUpsilon_{T,2}-\bUpsilon_{T,1}}^2/2\}.
    \label{type II error}
\end{align}
Consider a sieve $\mathcal{F}_T$ consisting of all alternative parameter points $\bzeta$ such that $\|\bm{f}\|_\infty, \|1/\bm{f}\|_\infty\le C ({T} \bar\epsilon_T^2)^{1/\gamma}$ and $\max( \|\bL_1\|_\infty ,\ldots,\|\bL_p\|_\infty)\le C ({T} \bar \epsilon_T^2)^{1/\gamma}$, $\max( \|\bK_1\|_\infty ,\ldots,\|\bK_p\|_\infty)\le C ({T} \bar \epsilon_T^2)^{1/\gamma}$, $\|\bE\|_\infty \le C ({T} \bar \epsilon_T^2)^{1/\gamma}$ 
for some constant $C$ to be chosen later, where $\gamma$ is as in Condition (A2). By Part (i) of Lemma~\ref{norm comparison}, for $\bzeta_1\in \mathcal{F}_T$, 
$$\op{\bUpsilon_{T,1}^{-1}}\le \op{\bSigma_1^{-1}}\le \op{\bOmega_1}+\sum_{j=1}^p \op{\bL_{j,1}\bL_{j,1}\trans}\le \|\bm{f}_1\|_\infty+ \sum_{j=1}^p \tr(\bL_{j,1}\bL_{j,1}\trans), $$
which is bounded by 
$$\|\bm{f}_1\|_\infty+pdr\max_j(\|\bL_{j,1}\|_\infty^2)\le C({T} \bar\epsilon_T^2)^{1/\gamma}+pdr (T\bar\epsilon_T^2)^{2/\gamma}\le C' dr (T\bar\epsilon_T^2)^{2/\gamma}$$  
for some constant $C'>0$ as $p$ is fixed and $T\bar\epsilon_T^2\to \infty$. Thus the expression in \eqref{type II error} is bounded by $\exp[-c_2 T\epsilon_T^2+ (C'dr)^2 (T\bar\epsilon_T^2)^{4/\gamma}\delta_T^2]\le \exp[-c_3 T\epsilon_T^2]$ for some constant $c_3>0$ by choosing $\delta_T$ a sufficiently small constant multiple of $(C'dr)^{-1} (T\bar\epsilon_T^2)^{1/2-2/\gamma}$ and $M>0$ large enough.

 Under the Lemma~\ref{lem:for_sieve}, within the sieve, we have $O_T=G_T\lesssim (T\bar\epsilon_T^2)^{2/\gamma}$ and $C_T=C^2dr(T\bar\epsilon_T^2)^{2/\gamma}$. Then $M_{\ell,U}\lesssim d(T\bar\epsilon_T^2)^{4/\gamma}, M_{\ell,V}\lesssim d(T\bar\epsilon_T^2)^{2/\gamma}$ and $M_{\ell,C,V,U}\lesssim d^p(T\bar\epsilon_T^2)^{8p/\gamma}$. Then $M_{P,1},M_{P,2}\lesssim d^{3+p}(T\bar\epsilon_T^2)^{(8+8p)/\gamma}$. Thus, $$\|\bUpsilon_{1,T}-\bUpsilon_{2,T}\|^2_{\mathrm{F}}\leq T^3\zeta_T^2(T\bar\epsilon_T^2)^{2/\gamma}d^{3+p}(T\bar\epsilon_T^2)^{(8+8p)/\gamma} = \delta^2_T,$$
when $\max(\|\bOmega_1-\bOmega_2\|^2_{\mathrm{F}},\|\bK_{1,j}-\bK_{2,j}\|^2_{\mathrm{F}},\|\bL_{1,k}-\bL_{2,k}\|^2_{\mathrm{F}})\lesssim \zeta_T^2$

We have, 
\begin{itemize}
\item [(1)]  $\|\bK_{1,j}-\bK_{2,j}\|^2_{\mathrm{F}}\le dr\|\bK_{1,j}-\bK_{2,j}\|^2_{\infty}$;
\item [(2)] $\|\bL_{1,j}-\bL_{2,j}\|^2_{\infty}\le dr\|\bL_{1,j}-\bL_{2,j}\|^2_{\mathrm{\infty}}$; 
\item [(3)] $\|\bOmega_1-\bOmega_2\|^2_{\mathrm{F}}$ is bounded by $$\|\bm{f}_1-\bm{f}_2\|^2_{\mathrm{F}} + 2\|\bm{f}_1\|_{\infty}\|\bE-\bE_0\|_{\mathrm{F}}^2\le d\|\bm{f}_1-\bm{f}_2\|^2_{\infty} + 2\|\bm{f}_1\|_{\infty}\|\bE_1\|_0\|\bE_1-\bE_2\|_{\infty}^2.$$
\end{itemize} 

Using (1), (2), and (3) above, we have 
\begin{itemize} 
\item [(i)] $\|\bK_{1,j}-\bK_{2,j}\|^2_{\infty},\|\bL_{1,j}-\bL_{2,j}\|^2_{\infty}\le \zeta_T^2/dr$;
\item [(ii)] $ \|\bm{f}_1-\bm{f}_2\|^2_{\infty} \le \zeta_T^2/d$; 
\item [(iii)] $ \|\bE_1-\bE_2\|_{\infty}^2 \le \max\{ ({T} \bar \epsilon_T^2)^{-2/\gamma}\log d, 1/d^2\}\zeta_T^2$ implies 
$$\{\|\bOmega_1-\bOmega_2\|^2_{\mathrm{F}},\|\bK_{1,j}-\bK_{2,j}\|^2_{\mathrm{F}},\|\bL_{1,k}-\bL_{2,k}\|^2_{\mathrm{F}}\}\lesssim \zeta_T^2.$$ 
\end{itemize}
Using the result from the above display, we let 
$$\zeta_T^2=\delta_T^2T^{-3}(T\bar\epsilon_T^2)^{-2/\gamma}d^{-3-p}(T\bar\epsilon_T^2)^{-(8+8p)/\gamma}.$$
Thus, at $\bzeta_2 \in \mathcal{F}_T$ with $\Frob{\bUpsilon_{T,2}-\bUpsilon_{T,1}}<\delta_T$, 
$$ \EE_{\bzeta_2}(1-\phi_T)\le \exp[-c_3 T\epsilon_T^2].$$ 

\subsection{Rest of the proof of Theorem~2}

To show that the posterior probability of $\{\bzeta: \Frob{\bGamma-\bUpsilon_0}>\epsilon_T\}$ converges to zero in probability, it remains to show that the prior probability of the complement of the sieve is exponentially small: $\Pi (\mathcal{F}_T^c)\le e^{-c' T\bar \epsilon_T^2}$ for some sufficiently large $c'>0$. 

Observe that 
\begin{align*}
\Pi  (\mathcal{F}_T^c) &\le \sum_{j=1}^d \Pi (f_j>C ({T}\bar\epsilon_T^2)^{1/\gamma})+\sum_{k=1}^p\sum_{j=1}^d \sum_{l=1}^r \Pi (|L_{1,kl}| >C ({T}\bar\epsilon_T^2)^{1/\gamma})\\&\quad+\sum_{k=1}^p\sum_{j=1}^d \sum_{l=1}^r \Pi (|L_{1,kl}| >C ({T}\bar\epsilon_T^2)^{1/\gamma}) \\
&\le p C_0 e^{-c_0 C^\gamma T \bar\epsilon_T^2}+2pdr C_0 e^{-c_0 C^\gamma T \bar\epsilon_T^2}.
\end{align*}
Since $\log d\ll T\bar\epsilon_T^2$, the above expression is bounded by $B_0 e^{-b_0 T \epsilon_T^2}$ for some $B_0,b_0>0$ and $b_0$ can be chosen as large as we please by making $C$ in the definition of the sieve $\mathcal{F}_T$ large enough.
Hence, by arguing as in the proof of Theorem~8.19 of \cite{GhosalBook}, it follows that the posterior contraction rate for $\bUpsilon_T$ at $\bUpsilon_{T,0}$ in terms of the Frobenius distance is $\sqrt{T}\sqrt{d^2 (\log T)/T}=d\sqrt{\log T}$. 

Let the collection of such $\bzeta_2$ with $\Frob{\bUpsilon_{T,2}-\bUpsilon_{T,1}}<\delta_T$ be denoted by 
\begin{align*} 
\mathcal{B}_T &=\{\|\bK_{1,j}-\bK_{2,j}\|^2_{\infty}, \|\bL_{1,j}-\bL_{2,j}\|^2_{\infty}\le \zeta_T^2/dr,\\
& \qquad\qquad
\|\bm{f}_1-\bm{f}_2\|^2_{\infty} \le \zeta_T^2/d, \|\bE_1-\bE_2\|_{\infty}^2 \le \zeta_T^2/d^2\}, 
\end{align*}
where $\zeta_T^2=\delta_T^2T^{-3}(T\bar\epsilon_T^2)^{-2/\gamma}d^{-3-p}(T\bar\epsilon_T^2)^{-(8+8p)/\gamma}$. 
The number of such sets $\mathcal{B}_T$ needed to cover $\mathcal{F}_T$ is at most $N_T\le T^{B}$ following the similar lines of arguments as in Lemma 8 of \cite{roy2024bayesian} setting $\delta_T$ a negative polynomial in $T$. Therefore $\log N_T \lesssim T\epsilon_T^2 $. For each $\mathcal{B}_T$ containing a point $\bzeta_1$ such that $\Frob{\bUpsilon_{T,1}-\bUpsilon_{T,0}}>\epsilon_T$, we choose a representative and construct the likelihood ratio test $\phi_T$. The final test for testing the hypothesis $\bzeta=\bzeta_0$ against the alternative $\{\bzeta: \Frob{\bUpsilon_{T,1}-\bUpsilon_{T,0}}>\sqrt{T}\epsilon_T\}$ is the maximum of these tests. Then, the resulting test has the probability of type II error bounded by $\exp[-c_3 T\epsilon_T^2]$ while the probability of type I error is bounded by $\exp[\log N_T-c_2' T \epsilon_T^2]\le \exp[-c_2' T \epsilon_T^2/2]$ if we choose $M>0$ sufficiently large.

Finally, we convert the notion of convergence to those for $\bOmega$ and $\bA_1,\ldots,\bA_p$. 
Recall that the eigenvalues of $\Gamma_{T,0}$ are bounded between two fixed positive numbers, because by Lemma~\ref{norm comparison}, $\op{\bUpsilon_0}= \op{\bOmega^{-1}_0}$ and $\op{\bUpsilon_0^{-1}}\le \op{\bSigma_0^{-1}}$ and $\bSigma_0$ and $\Omega_0$ are assumed to have eigenvalues bounded between two fixed positive numbers. Now, applying Lemmas~\ref{lem:boundtransfer1} and \ref{lem:boundtransfer2} respectively,  
$\Frob{\bOmega-\bOmega_0}^2\le C_1 T^{-1}\Frob{\bUpsilon_T-\bUpsilon_{T,0}}^2$ and for all $j=1,\ldots,p$, $\Frob{\bA_j-\bA_{j,0}}^2\le C_1 T^{-1}\Frob{\bUpsilon_T-\bUpsilon_{T,0}}^2$ for some constant $C_1>0$. These relations yield the rate $d\sqrt{(\log T)/T}$ for the precision matrix $\bOmega$ and all regression coefficients at their respective true values in terms the Frobenius distance. 

\subsection{Proof of Theorem~3}

When $\bOmega_0$ is sparse and the prior for $\bOmega$ imposes sparsity as described in Section~3.4, we show the improved rate $\sqrt{(d+s)(\log T)/T}$. To this end, we refine the estimate of the prior concentration in Subsection~\ref{prior concentration}. The arguments used for the test construction, bounding the prior probability of the complement of the sieve and for converting the rate result in terms of the Frobenius distance on $\bOmega$ and $\bR_1,\ldots,\bR_p$ remain the same. Except to control the number of nonzero entries in $\bE$, we add $\|\bE\|_0\le C ({T} \bar \epsilon_T^2)^{1/\gamma}/\log d$ in the sieve and let 
\begin{align*} 
\mathcal{B}_T &=\{\|\bK_{1,j}-\bK_{2,j}\|^2_{\infty}\le \zeta_T^2/dr,\|\bL_{1,j}-\bL_{2,j}\|^2_{\infty}\le \zeta_T^2/dr, \\ 
&\qquad \|\bm{f}_1-\bm{f}_2\|^2_{\infty} \le \zeta_T^2/d, \|\bE_1-\bE_2\|_{\infty}^2 \le ({T} \bar \epsilon_T^2)^{-2/\gamma}\log d\zeta_T^2\}. 
\end{align*}
Covering number calculation will remain identical to Lemma 8 of \cite{roy2024bayesian}.

Since $r$ is assumed to be a fixed constant, $-\log \Pi (\|\bL_j-\bL_{j,0}\|\le \eta)\lesssim d \log (1/\eta)$, $j=1,\ldots,p$, and so is the corresponding estimate for $\bK_1,\ldots,\bK_p$. For $\bm{f}$, the corresponding estimate $d\log (1/\eta)$ does not change from the previous scenario. However, $\bE_0$ is now $s$-dimensional instead of $\binom{d}{2}$. The estimate for $-\log \Pi (\|\bE-\bE_{0}\|_\infty\le \eta)$ under the shrinkage prior improves to a multiple of $(d+s)\log (1/\eta)$ by the arguments given in the proof of Theorem~4.2 of \cite{shi2021bayesian}. Thus 
$$-\log \Pi (\|\bzeta-\bzeta_0\|_\infty\le \eta)\lesssim (d+s) \log (1/\eta),$$ 
leading to the asserted pre-rate $\bar \eta=\sqrt{(d+s)(\log T)/T}$. We also note that the additional sparsity imposed on $\bL_1,\ldots,\bL_p$ through the prior does not further improve the rate when $r$ is a fixed constant and it may be replaced by a fully non-singular prior.

\section{Auxiliary lemmas and proofs}
\label{sec:lemmas}


The following lemma expresses certain divergence measures in the family of centered multivariate normal distributions. The proofs follow from direct calculations. 

\begin{lemma}
\label{normal divergences}
   Let $f_1$ and $f_2$ be probability densities of $k$-dimensional normal distributions with mean zero and dispersion matrices $\bm{\Delta}_1$ and  $\bm{\Delta}_2$ respectively. Let $\bm{R}=\bm{\Delta}_1^{-1/2}(\bm{\Delta}_2- \bm{\Delta}_1)\bm{\Delta}_1^{-1/2}$ with eigenvalues $-1<\rho_1,\ldots,\rho_k<\iy$. Then the following assertions hold: 
   \begin{enumerate}
       \item [{\rm (i)}] The Reyni divergence $R(f_1,f_2) =-\log \int \sqrt{f_1 f_2}$ is given by 
       \begin{eqnarray*}
       \lefteqn{\frac12 \log \det ((\bm{\Delta}_1+\bm{\Delta}_2)/2) -\frac14 \log \det (\bm{\Delta}_1)-\frac14\log\det (\bm{\Delta}_2)}\\
       &&=\frac12 \log \det (\bI+\bR/2)-\frac14 \log\det (\bI+\bR)\\
       &&=\frac14 \sum_{j=1}^k [2\log (1+\rho_j/2)-\log (1+\rho_j)].
       \end{eqnarray*}
       \item [{\rm (ii)}] The Kullback-Leibler divergence $K(f_1,f_2) =\int f_1 \log (f_2/f_1)$ is given by 
       \begin{eqnarray*}
       \lefteqn{\frac12 \log \det (\bm{\Delta}_2)-\frac12 \log\det (\bm{\Delta}_1) +\frac12 \tr(\bm{\Delta}_2^{-1/2}(\bm{\Delta}_2- \bm{\Delta}_1) \bm{\Delta}_2^{-1/2}) }\\
       &&=\frac12 \log \det (\bI+\bR)+\frac12 \tr((\bI+\bR)^{-1}-\bI)\\
       &&=\frac12 \sum_{j=1}^k [\log (1+\rho_j)-\rho_j/(1+\rho_j)].
       \end{eqnarray*}
       \item [{\rm (iii)}] The Kullback-Leibler variation $V(f_1,f_2) =\int f_1 (\log (f_2/f_1)-K(f_1,f_2))^2$ is given by 
       \begin{align*}
       \frac12 \tr((\bm{\Delta}_2^{-1/2}(\bm{\Delta}_2- \bm{\Delta}_1) \bm{\Delta}_2^{-1/2})^2) 
       =\frac12 \tr[((\bI+\bR)^{-1}-\bI)^2]
       =\frac12 \sum_{j=1}^k \rho_j^2/(1+\rho_j)^2.
       \end{align*}
\item [{\rm (iv)}] If $2\bm{\Delta}_1-\bm{\Delta}_2$ is positive definite, then expected squared-likelihood ratio $\int (f_2/f_1)^2 f_1$ is given by 
\begin{align*} 
\frac{\det (\bm{\Delta}_1)}{\sqrt{\det(\bm{\Delta}_2) \det (2\bm{\Delta_1}-\bm{\Delta}_2)}}=\frac{1}{\sqrt{\det(\bI+\bR)\det(\bI-\bR)}}=\exp[\sum_{j=1}^k \log (\rho_j^2-1)/2]
\end{align*}
which is bounded by 
$$ \exp[\sum_{j=1}^k \rho_j^2/2]=\exp[ \Frob{\bR}^2/2]\le \exp [\op{\bDelta_1^{-1}}^2\|\Frob{\bDelta_2-\bDelta_1}^2/2].$$
   \end{enumerate}
\end{lemma}

\begin{lemma}
\label{norm comparison}
For all $j$, we have
\begin{itemize}
    \item [{\rm (i)}] $\|\bUpsilon_{j}^{-1}\|^2_{\mathrm{op}}\leq \|\bC_p^{-1}\|^2_{\mathrm{op}}$,
    \item [{\rm (ii)}] $\|\bD_j^{-1}\|^2_{\mathrm{op}}\leq \|\bC_p^{-1}\|^2_{\mathrm{op}}$, 
    \item [{\rm (iii)}] $\|\bUpsilon_{j}\|^2_{\mathrm{op}}\leq \|\bGamma(0)\|^2_{\mathrm{op}}$,
    \item [{\rm (iv)}] $\|\bD_j\|^2_{\mathrm{op}}\leq \|\bGamma(0)\|^2_{\mathrm{op}}=\|\bOmega^{-1}\|^2_{\mathrm{op}}$,
      \item [{\rm (v)}] $\|\bUpsilon_{j}^{-1}\|^2_{\mathrm{op}}\leq \max\{\|\bC^{-1}_j\|^2_{\mathrm{op}},\|\bUpsilon_{j-1}\|^2_{\mathrm{op}}\}\leq \|\bC_p^{-1}\|^2_{\mathrm{op}}$ 
\end{itemize}
\end{lemma}

\begin{lemma}
   We have $\|\bOmega_1-\bOmega_2\|^2_{\mathrm{F}}$ is small if $T^{-1}\|\bUpsilon_{1,T}-\bUpsilon_{2,T}\|^2_{\mathrm{F}}$ is small.
    \label{lem:boundtransfer1}
\end{lemma}

\begin{lemma}
    If $T^{-1}\|\bUpsilon_{1,T}-\bUpsilon_{2,T}\|^2_{\mathrm{F}}$ is small, $\|\bA_1-\bA_2\|^2_{\mathrm{F}}$ is small, where $\bA$ is the autoregressive coefficient.
    \label{lem:boundtransfer2}
\end{lemma}






The following Lemma is for around the truth case. 
\begin{lemma}
    With $\{\|\bOmega_1-\bOmega_0\|^2_{\mathrm{F}},\|\bK_{1,j}-\bK_{0,j}\|^2_{\mathrm{F}},\|\bL_{1,k}-\bL_{0,k}\|^2_{\mathrm{F}}\}\lesssim \epsilon^2$, we have
\begin{align*}
   \|\bUpsilon_{1,T}-\bUpsilon_{0,T}\|^2_{\mathrm{F}}\lesssim \epsilon^2 T^3,
\end{align*}
assuming fixed lower and upper bounds on eigenvalues of $\bOmega_0$ and an upper bound for eigenvalues of $\bL_{0,k}$.
\label{lem:around_the_truth}
\end{lemma}

\begin{lemma}
    With $\{\|\bOmega_1-\bOmega_2\|^2_{\mathrm{F}},\|\bK_{1,j}-\bK_{2,j}\|^2_{\mathrm{F}},\|\bL_{1,k}-\bL_{2,k}\|^2_{\mathrm{F}}\}\lesssim \epsilon^2$, and $\|\bOmega_{\ell}\|^2_{\mathrm{op}}\leq O_T, \|\bUpsilon_{\ell}(0)\|^2_{\mathrm{op}}\leq G_T, \|\bL_{\ell,k}\|^2_{\infty}\leq L_T$ and $\|\bC_{\ell,p}^{-1}\|^2_{\mathrm{op}}\leq C_T=O_T+drL_T$  for $\ell=1,2$, Setting $M_{1,U}=G_TL_Tdr$, $M_{1,V}= C_T$ and $M_{1,C,V,U}=(1+C^2_TG_TL_Tdr)^p$, we have, 
\begin{align*}
   \|\bUpsilon_{1,T}-\bUpsilon_{2,T}\|^2_{\mathrm{F}}\leq \epsilon^2[T^2M_{P,1}+T^3G_TM_{P,2}],
\end{align*}
where $M_{P,1}$ is 
\begin{align*} 
\lefteqn{(p+1) + p^3\{M_{1,U}M_{1,V}M_{2,U}M_{2,V} + 2C_T(M_{1,V}+M_{2,U})\} }\\
&+ 2G_T p^2\{p^2 M_{1,U}M_{1,V}M_{1,C,U,V} + 2C_TM_{2,V}M_{1,V}M_{1,C,U,V}  +2C_TM_{2,U}M_{1,V}M_{1,C,U,V}\}
\end{align*}
and $M_{P,2}$ is 
$$2p^2(M_{1,U}M_{1,V}M_{1,C,U,V}+p) + 2pC_TM_{2,V}M_{1,V}M_{1,C,U,V} +2pC_TM_{2,U}M_{1,V}M_{1,C,U,V}.$$
\label{lem:for_sieve}
\end{lemma}

The proof is based on the next five lemmas and results.

\begin{lemma}
   The quantity $\|\bUpsilon_{1,T}-\bUpsilon_{2,T}\|^2_{\mathrm{F}}$ is bounded by $$ T^2\|\bUpsilon_{1,p-1}-\bUpsilon_{2,p-1}\|^2_{\mathrm{F}}+T^2(T-p+1)\|\bUpsilon_{1}(0)\|^2_{\mathrm{op}}\|\bkappa_{1,p}\bUpsilon_{1,p-1}^{-1}-\bkappa_{2,p}\bUpsilon_{2,p-1}^{-1}\|^2_{\mathrm{F}}.$$
   \label{lem:bigsigma}
\end{lemma}

\begin{lemma}[Auxiliary recurrences]
  The following bounds hold:   
\begin{align}
    \|\bUpsilon_{1,j}-\bUpsilon_{2,j}\|^2_{\mathrm{F}} &\leq \|\bUpsilon_{1,j-1}-\bUpsilon_{2,j-1}\|^2_{\mathrm{F}}+\|\bD_{1,j}-\bD_{2,j}\|^2_{\mathrm{F}}\nonumber \\ 
    & \quad +(\|\bUpsilon_1(0)\|^2_{\mathrm{op}}+\|\bUpsilon_2(0)\|^2_{\mathrm{op}}) \nonumber\\
    &\qquad \times \|\bUpsilon_{1,j-1}^{-1}\bkappa_{1,j}^{\mathrm{T}}-\bUpsilon_{2,j-1}^{-1}\bkappa_{2,j}^{\mathrm{T}}\|^2_{\mathrm{F}}; \\
   \|\bD_{1,j}-\bD_{2,j}\|^2_{\mathrm{F}} &\leq\|\bD_{1,j-1}-\bD_{2,j-1}\|^2_{\mathrm{F}} \nonumber\\ 
   &\quad + \|\bW_{1,j}\|^2_{\mathrm{op}}\|\bW_{2,j}\|^2_{\mathrm{op}}\|\bC_{1,j-1}^{-1}-\bC_{2,j-1}^{-1}\|^2_{\mathrm{F}} \nonumber \\ 
   &\quad + (\|\bC_{1,j-1}^{-1}\|^2_{\mathrm{op}}+\|\bC_{2,j-1}^{-1}\|^2_{\mathrm{op}})\|\bW_{1,j}-\bW_{2,j}\|^2_{\mathrm{F}};\\
    \|\bW_{1,j}-\bW_{2,j}\|^2_{\mathrm{F}} &\leq \|\bU_{1,j}-\bU_{2,j}\|^2_{\mathrm{F}}\|\bV_{1,j}\|^2_{\mathrm{op}} \nonumber\\ 
    & \quad +\|\bV_{1,j}-\bV_{2,j}\|^2_{\mathrm{F}}\|\bU_{2,j}\|^2_{\mathrm{op}};\\
        \|\bUpsilon_{1,j-1}^{-1}\bkappa_{1,j}^{\mathrm{T}}-\bUpsilon_{2,j-1}^{-1}\bkappa_{2,j}^{\mathrm{T}}\|^2_{\mathrm{F}} &\leq\|\bUpsilon_{1,j-2}^{-1}\bkappa_{1,j-1}^{\mathrm{T}}-\bUpsilon_{2,j-2}^{-1}\bkappa_{2,j-1}^{\mathrm{T}}\|^2_{\mathrm{F}} \nonumber \\
        &\qquad \times(1+\|\bC_{1,j}^{-1}\|^2_{\mathrm{op}}\|\bU_{1,j}\|^2_{\mathrm{op}}\|\bV_{1,j}\|^2_{\mathrm{op}}) \nonumber\\
        &\quad+2\|\bC_{1,j}^{-1}\bU_{1,j}\bV_{1,j}-\bC_{2,j}^{-1}\bU_{2,j}\bV_{2,j}\|^2_{\mathrm{F}}. 
    \end{align}

\label{lem:rec}
\end{lemma}

\begin{lemma}[Bounding the first term in Lemma \ref{lem:bigsigma}]
The quantity $\|\bUpsilon_{1,p-1}-\bUpsilon_{2,p-1}\|^2_{\mathrm{F}}$ is bounded by 
    \begin{align}
   \lefteqn{\|\bUpsilon_{1}(0)-\bUpsilon_{2}(0)\|^2_{\mathrm{F}}+\sum_{j=0}^{p-1}\|\bD_{1,j}-\bD_{2,j}\|^2_{\mathrm{F}}} \label{eq:ineq1}\\
   &\qquad +(\|\bUpsilon_1(0)\|^2_{\mathrm{op}}+\|\bUpsilon_2(0)\|^2_{\mathrm{op}})\sum_{j=0}^{p-1}\|\bUpsilon_{1,j-1}^{-1}\bkappa_{1,j}^{\mathrm{T}}-\bUpsilon_{2,j-1}^{-1}\bkappa_{2,j}^{\mathrm{T}}\|^2_{\mathrm{F}},\nonumber
\end{align}
and 
\begin{align} 
\sum_{j=0}^{p-1}\|\bD_{1,j}-\bD_{2,j}\|^2_{\mathrm{F}} 
& \le {p\|\bUpsilon_1(0)-\bUpsilon_2(0)\|^2_{\mathrm{F}} } \nonumber\\
     &\quad + p\sum_{k=1}^{p}\big[M_{1,U}M_{1,V}M_{2,U}M_{2,V}\|\bC_{1,k-1}^{-1}-\bC_{2,k-1}^{-1}\|^2_{\mathrm{F}}\nonumber\\
     &\qquad + (\|\bC_{1,p}^{-1}\|^2_{\mathrm{op}}+\|\bC_{2,p}^{-1}\|^2_{\mathrm{op}})\nonumber \\
     &\qquad\qquad \times (\|\bU_{1,k}-\bU_{2,k}\|^2_{\mathrm{F}}M_{1,V}+\|\bV_{1,k}-\bV_{2,k}\|^2_{\mathrm{F}}M_{2,U})\big],\label{eq:ineq2}
\end{align}
    where, for $\ell\in\{1,2\}$, 
    $$M_{\ell,C,U,V}=(1+\|\bC_{\ell,p}^{-1}\|^2_{\mathrm{op}}\|M_{\ell,U}M_{\ell,V})^p, \quad M_{\ell,U}=\max_j\|\bU_{\ell,j}\|^2_{\mathrm{op}},M_{\ell,V}=\max_j\|\bV_{\ell,j}\|^2_{\mathrm{op}}.$$
\label{lem:finalbd}
\end{lemma}

\begin{lemma}[Bounding the second term in Lemma \ref{lem:bigsigma}]
We have   
  \begin{align}
\|\bUpsilon_{1,p-1}^{-1}\bkappa_{1,p}^{\mathrm{T}}-\bUpsilon_{2,p-1}^{-1}\bkappa_{2,p}^{\mathrm{T}}\|^2_{\mathrm{F}}
 & \le {2M_{1,U}M_{1,V}M_{1,C,U,V}\sum_{k=0}^p\|\bC^{-1}_{1,k}-\bC^{-1}_{2,k}\|^2_{\mathrm{F}} }\nonumber\\
 &\quad + 2\|\bC_{2,p}^{-1}\|^2_{\mathrm{op}}M_{2,V}M_{1,V}M_{1,C,U,V}\sum_{k=1}^p\|\bU_{1,k}-\bU_{2,k}\|^2_{\mathrm{F}}\nonumber\\
&\quad +2\|\bC_{2,p}^{-1}\|^2_{\mathrm{op}}\|^2_{\mathrm{F}}M_{2,U}M_{1,V}M_{1,C,U,V}\sum_{k=1}^p\|\bV_{1,k}-\bV_{2,k}\|^2_{\mathrm{F}}, 
\label{eq:ineq3}
\end{align}
    where $$M_{\ell,U}=\max_j\|\bU_{\ell,j}\|^2_{\mathrm{op}}, \quad M_{\ell,V}=\max_j\|\bV_{\ell,j}\|^2_{\mathrm{op}},\quad  M_{\ell,C,U,V}=(1+\|\bC_{\ell,p}^{-1}\|^2_{\mathrm{op}}\|M_{\ell,U}M_{\ell,V})^p.$$
\label{lem:finalbdpre}
\end{lemma}

The following steps transform the bounds for $\bU_j$'s and $\bV_j$'s to $\bL_j$'s and $\bK_j$'s.

\begin{lemma}
The quantity $\|\bV_{j,1}-\bV_{j,2}\|^2_{\mathrm{F}}$ is bounded by 
\begin{align*}
        \lefteqn{ r\|\bK_{j,1}\|^2_{\mathrm {op}}\|(\bK_{j,1}^{\mathrm{T}}\bD_{j-1,1}^{-1}\bK_{1,1})^{-1}\|_{\mathrm{op}}\|(\bK_{j,2}^{\mathrm{T}}\bD_{j-1,2}^{-1}\bK_{j,2})^{-1}\|_{\mathrm{op}} }\\
        &
        \quad\times \{\|\bK_{j,1}\|_{\mathrm{op}}\|\bK_{j,2}\|_{\mathrm{op}}\|\bD_{j-1,1}^{-1}-\bD_{j-1,2}^{-1}\|_{\mathrm{op}}\\
        &+|\bK_{j,1}-\bK_{j,2}\|^2_{\mathrm{op}}(\|\bD_{j-1,1}^{-1}\|_{\mathrm{op}}\|\bK_{j,1}\|_{\mathrm{op}}+\|\bD_{j-1,2}^{-1}\|_{\mathrm{op}}\|\bK_{j,2}\|_{\mathrm{op}})\}\\&+ r\|(\bK_{j,2}^{\mathrm{T}}\bD_{j-1,2}^{-1}\bK_{j,2})^{-1/2}\|^2_{\mathrm{op}}\|\bK_{j,1}-\bK_{j,2}\|^2_{\mathrm{op}}, 
    \end{align*}
and    $\|\bU_{j,1}-\bU_{j,2}\|^2_{\mathrm{op}}$ is bounded by 
\begin{align*}
       \lefteqn{ \|\bC_{j-1,1}\|^2_{\mathrm{op}}\|\bL_{j,1}\|^2_{\mathrm{op}}\|(\bI + \bL^{\mathrm{T}}_{j,1} \bC_{j-1,1}\bL_{j,1})^{-1}\|_{\mathrm{op}} }\\
       &\quad\times \|(\bI + \bL^{\mathrm{T}}_{j,2} \bC_{j-1,2}\bL_{j,2})^{-1}\|_{\mathrm{op}}\|\bL^{\mathrm{T}}_{j,1} \bC_{j-1,1}\bL_{j,1}       -\bL^{\mathrm{T}}_2 \bC_{j-1,2}\bL_{j,2}\|_{\mathrm{op}}\\
       &+\|(\bI + \bL^{\mathrm{T}}_{j,2} \bC_{j-1,2}\bL_{j,2})^{-1/2}\|^2_{\mathrm{op}}\|\bC_{j-1,1}\bL_{j,1}-\bC_{j-1,2}\bL_{j,2}\|^2_{\mathrm{op}}.
    \end{align*}
    \label{lem:UjVj}
\end{lemma}

\subsection{Proofs of the Auxiliary Lemmas}

\begin{proof}[Proof of Lemma~\ref{norm comparison}]
    Since,
$$\bM={\begin{bmatrix}\bA&\bC^{\mathrm{T}}\\\bC&\bB\end{bmatrix}}={\begin{bmatrix}\bI_{p}&\bzero\\\bC\bA^{-1}&\bI_{q}\end{bmatrix}}{\begin{bmatrix}\bA&\bzero\\\bzero&\bB-\bC\bA^{-1}C^{\mathrm{T}}\end{bmatrix}}{\begin{bmatrix}\bI_{p}&
\bA^{-1}\bC^{\mathrm{T}}\\\bzero&\bI_{q}\end{bmatrix}},$$
we have 
$$\bUpsilon_{j}={\begin{bmatrix}\bI_{d(j-1)}&\bzero\\\bkappa_j\bUpsilon_{j-1}^{-1}&\bI_{d}\end{bmatrix}}{\begin{bmatrix}\bUpsilon_{j-1}&\bzero\\\bzero&\bGamma(0)-\bkappa_j^{\mathrm{T}}\bUpsilon_{j-1}^{-1}\bkappa_j\end{bmatrix}}{\begin{bmatrix}\bI_{d(j-1)}&\bUpsilon_{j-1}^{-1}\bkappa_j^{\mathrm{T}}\\\bzero&\bI_{d}\end{bmatrix}}=\bP\bQ\bP^{\mathrm{T}}.$$

The operator norm of $\bP$ is 1 and $\bD_j=\bGamma(0)-\bkappa_j^{\mathrm{T}}\bUpsilon_{j-1}^{-1}\bkappa_j$, giving 
$$\|\bUpsilon_{j}\|^2_{\mathrm{op}}\leq\|\bQ\|^2_{\mathrm{op}}\leq \max\{\|\bD_j\|^2_{\mathrm{op}},\|\bUpsilon_{j-1}\|^2_{\mathrm{op}}\}\leq \|\bGamma(0)\|^2_{\mathrm{op}},$$ by applying the first inequality recursively as $\|\bD_j\|^2_{\mathrm{op}}\leq \|\bGamma(0)\|^2_{\mathrm{op}}$.

Applying the above, we also have,
$$\bUpsilon_{j}^{-1}={\begin{bmatrix}\bI_{d(j-1)}&\bzero\\-\bkappa_j\bUpsilon_{j-1}^{-1}&\bI_{d}\end{bmatrix}}{\begin{bmatrix}\bUpsilon_{j-1}^{-1}&0\\0&\bD_{j}^{-1}\end{bmatrix}}{\begin{bmatrix}\bI_{d(j-1)}&-\bUpsilon_{j-1}^{-1}\bkappa_j^{\mathrm{T}}\\\bzero&\bI_{d}\end{bmatrix}}=\bG\bH\bG^{\mathrm{T}}.$$
The operator norm of $\bG$ is 1, and so  $$\|\bUpsilon_{j}^{-1}\|^2_{\mathrm{op}}\leq \|\bH\|^2_{\mathrm{op}}\leq \max\{\|\bD^{-1}_j\|^2_{\mathrm{op}},\|\bUpsilon_{j-1}^{-1}\|^2_{\mathrm{op}}\}\leq \|\bC_p^{-1}\|^2_{\mathrm{op}},$$ by applying the first inequality recursively since $\|\bD_j^{-1}\|^2_{\mathrm{op}}\leq \|\bC_p^{-1}\|^2_{\mathrm{op}}$ and $\bC_p^{-1}=\bOmega+\sum_{k=1}^p \bL_K\bL_k^{\mathrm{T}}$. 
Thus,
\begin{align*} 
\|\bUpsilon_{1,j}-\bUpsilon_{2,j}\|^2_{\mathrm{F}} & \leq \|\bQ_1-\bQ_2\|_{\mathrm{F}}+(\|\bQ_1\|^2_{\mathrm{op}}+\|\bQ_2\|^2_{\mathrm{op}})\|\bP_1-\bP_2\|^2_{\mathrm{F}} \\ 
& \leq \|\bUpsilon_{1,j-1}-\bUpsilon_{2,j-1}\|^2_{\mathrm{F}}+\|\bD_{1,j}-\bD_{2,j}\|^2_{\mathrm{F}}
\\
&\quad +(\|\bUpsilon_1(0)\|^2_{\mathrm{op}}+\|\bUpsilon_2(0)\|^2_{\mathrm{op}})\|\bUpsilon_{1,j-1}^{-1}\bkappa_{1,j}^{\mathrm{T}}-\bUpsilon_{2,j-1}^{-1}\bkappa_{2,j}^{\mathrm{T}}\|^2_{\mathrm{F}}.
\end{align*}
Thus we have 
$$\|\bkappa_{j}\|^2_{\mathrm{op}}\leq \|\bUpsilon_{j-1}^{-1}\bkappa_{j}^{\mathrm{T}}\|^2_{\mathrm{op}}\|\bUpsilon_{j-1}\|^2_{\mathrm{op}}\leq \|\bUpsilon_{j-1}\|^2_{\mathrm{op}}\leq \|\bGamma(0)\|^2_{\mathrm{op}} ,$$
since $\|\bUpsilon_{j-1}^{-1}\bkappa_{j}^{\mathrm{T}}\|^2_{\mathrm{op}}\leq 1$, 
\end{proof}

\begin{proof}[Proof of Lemma~\ref{lem:boundtransfer1}]

There are $T$ many diagonal blocks of $\bUpsilon_1(0)$ and $\bUpsilon_{2}(0)$ in $\bUpsilon_{1,T}$ and $\bUpsilon_{2,T}$, respectively. Thus, $  \|\bUpsilon_1(0)-\bUpsilon_2(0)\|^2_{\mathrm{F}} \leq T^{-1}\|\bUpsilon_{1,T}-\bUpsilon_{2,T}\|^2_{\mathrm{F}}$ and 
 \begin{align*} 
\|\bUpsilon_2(0)^{-1}\|^2_{\mathrm{op}} & \leq \|\bUpsilon_1(0)^{-1}\|^2_{\mathrm{op}}+\|\bUpsilon_2(0)^{-1}\|^2_{\mathrm{op}}\|\bUpsilon_1(0)^{-1}\|^2_{\mathrm{op}}\|\bUpsilon_{1,T}-\bUpsilon_{2,T}\|^2_{\mathrm{op}}, 
\end{align*} 
which leads to the bound 
$$\|\bUpsilon_2(0)^{-1}\|^2_{\mathrm{op}}\leq \frac{\|\bUpsilon_1(0)^{-1}\|^2_{\mathrm{op}}}{1-\|\bUpsilon_1(0)^{-1}\|^2_{\mathrm{op}}\|\bUpsilon_{1}(0)-\bUpsilon_{2}(0)\|^2_{\mathrm{op}}}$$ whenever $\|\bUpsilon_{1}(0)-\bUpsilon_{2}(0)\|^2_{\mathrm{op}}$ is sufficiently small. 
 Hence, 
 \begin{align*}
 \|\bOmega_1-\bOmega_2\|^2_{\mathrm{F}} & \le \|\bUpsilon_1(0)-\bUpsilon_2(0)\|^2_{\mathrm{F}}\|\bUpsilon_1(0)^{-1}\|^2_{\mathrm{op}}\|\bUpsilon_2(0)^{-1}\|^2_{\mathrm{op}}\\
& \leq \|\bUpsilon_1(0)-\bUpsilon_2(0)\|^2_{\mathrm{F}}\frac{\|\bUpsilon_1(0)^{-1}\|^2_{\mathrm{op}}}{1-\|\bUpsilon_1(0)^{-1}\|^2_{\mathrm{op}}\|\bUpsilon_{1}(0)-\bUpsilon_{2}(0)\|^2_{\mathrm{op}}}. 
 \end{align*}
 This completes the proof.
\end{proof}

\begin{proof}[Proof of Lemma~\ref{lem:boundtransfer2}]
Note that 
    $$\|\bA_1-\bA_2\|^2_{\mathrm{F}}=\|\bxi_{1,p}-\bxi_{2,p}\|^2_{\mathrm{F}}\|\bGamma^{-1}_{1,p-1}\|^2_{\mathrm{op}}+\|\bGamma^{-1}_{1,p-1}-\bGamma^{-1}_{2,p-1}\|^2_{\mathrm{F}}\|\bxi_{2,p-1}\|^2_{\mathrm{op}}.$$
    As there are $T-p$ blocks of $\bxi_{1,p}$ in $\bUpsilon_{1,T}$, $(T-p)\|\bxi_{1,p}-\bxi_{2,p}\|^2_{\mathrm{F}}\leq \|\bUpsilon_{1,T}-\bUpsilon_{2,T}\|^2_{\mathrm{F}}$ which implies that $(1-p/T)\|\bxi_{1,p}-\bxi_{2,p}\|^2_{\mathrm{F}}\leq T^{-1}\|\bUpsilon_{1,T}-\bUpsilon_{2,T}\|^2_{\mathrm{F}}$. Moreover, we have the bound $\|\bxi_{2,p-1}\|^2_{\mathrm{op}}\leq \|\bxi_{1,p-1}\|^2_{\mathrm{op}}+\|\bxi_{1,p-1}-\bxi_{2,p-1}\|^2_{\mathrm{op}}$.

We assume without loss of generality that $p$ divides $T$; the contribution from the remainder may be neglected since $p/T\to 0$. 
As there are $T/p$ diagonal blocks of $\bxi_{1,p}$ in $\bUpsilon_{1,p-1}$, 
$$ p^{-1}\|\bUpsilon_{1,p-1}-\bUpsilon_{2,p-1}\|^2_{\mathrm{F}}\leq T^{-1} (T/p)\|\bUpsilon_{1,T}-\bUpsilon_{2,T}\|^2_{\mathrm{F}} \leq T^{-1} \|\bUpsilon_{1,T}-\bUpsilon_{2,T}\|^2_{\mathrm{F}},$$
establishing the claim. 
\end{proof}

\begin{proof}[Proof of Lemma~\ref{lem:around_the_truth}]
Let  
$$\max\{\|\bOmega_1-\bOmega_2\|_{\mathrm{F}}, \|\bL_{1,j}-\bL_{2,j}\|_{\mathrm{F}},\|\bU_{1,j}-\bU_{2,j}\|_{\mathrm{F}}, \|\bV_{1,j}-\bV_{2,j}\|_{\mathrm{F}}\}\leq \epsilon.$$ 
Since, $\|\bU_{\ell,j}\|^2_{\mathrm{op}}\leq \|\bC_{j-1}\|^2_{\mathrm{op}}\|\bL_j\|^2_{\mathrm{op}}$ and $\|\bV_{\ell,j}\|^2_{\mathrm{op}}\leq \|\bD_{j-1}^{-1}\|^2_{\mathrm{op}}\leq \|\bC_p^{-1}\|_{\mathrm{op}}$, we have
\begin{align*} 
\max_j\|\bU_{\ell,j}\|^2_{\mathrm{op}} &\leq \|\bUpsilon_1(0)\|^2_{\mathrm{op}}\|\bL_{\ell,j}\|^2_{\mathrm{op}}\leq  M_{\ell,U},\\ \max_j\|\bV_{\ell,j}\|^2_{\mathrm{op}} &\leq\|\bC_{\ell,p}^{-1}\|_{\mathrm{op}}= M_{\ell,V}, \\ (1+\|\bC^{-1}_{1,p}\|_{\mathrm{op}}\|\bUpsilon_1(0)\|^2_{\mathrm{op}}\|\bL_{1,j}\|^2_{\mathrm{op}}\|\bC_p^{-1}\|_{\mathrm{op}})^p &\leq M_{\ell,C,V,U}, 
\end{align*} 
we have, 
\begin{align*}
    &\|\bUpsilon_{1,p-1}-\bUpsilon_{2,p-1}\|^2_{\mathrm{F}}\\&\quad\leq \epsilon[(p+1) + p^3\{M_{1,U}M_{1,V}M_{2,U}M_{2,V} + (\|\bC_{1,p}^{-1}\|^2_{\mathrm{op}}+\|\bC_{2,p}^{-1}\|^2_{\mathrm{op}})(M_{1,V}+M_{2,U})\} \\&\qquad+ (\|\bUpsilon_1(0)\|^2_{\mathrm{op}}+\|\bUpsilon_2(0)\|^2_{\mathrm{op}})\{p^4M_{1,U}M_{1,V}M_{1,C,U,V}\\
    &\qquad + 2p^2\|\bC_{2,p}^{-1}\|^2_{\mathrm{op}}M_{2,V}M_{1,V}M_{1,C,U,V}  +2p^2\|\bC_{2,p}^{-1}\|^2_{\mathrm{op}}\|^2_{\mathrm{F}}M_{2,U}M_{1,V}M_{1,C,U,V}\}]
\end{align*}
and 
\begin{align*}
    \lefteqn{\|\bUpsilon_{1,p-1}^{-1}\bkappa_{1,p}^{\mathrm{T}}-\bUpsilon_{2,p-1}^{-1}\bkappa_{2,p}^{\mathrm{T}}\|^2_{\mathrm{F}}} \\
    &
\leq \epsilon[2p^2(M_{1,U}M_{1,V}M_{1,C,U,V}+p) + 2p\|\bC_{2,p}^{-1}\|^2_{\mathrm{op}}M_{2,V}M_{1,V}M_{1,C,U,V} \\&\qquad+2\|\bC_{2,p}^{-1}\|^2_{\mathrm{op}}\|^2_{\mathrm{F}}M_{2,U}M_{1,V}M_{1,C,U,V}p].
\end{align*}
Finally,
\begin{align}
   \|\bUpsilon_{1,T}-\bUpsilon_{2,T}\|_{\mathrm{F}}\leq \epsilon[T^2M_{P,1}+(T-p+1)\|\bUpsilon_{1,p-1}\|^2_{\mathrm{op}}M_{P,2}], 
   \label{eq:mainbd}
\end{align}
where 
\begin{align*} 
M_{P,1} &=[(p+1) + p^3\{M_{1,U}M_{1,V}M_{2,U}M_{2,V} \\
&\quad + (\|\bC_{1,p}^{-1}\|^2_{\mathrm{op}}+\|\bC_{2,p}^{-1}\|^2_{\mathrm{op}})(M_{1,V}+M_{2,U})\} \\
&\quad + (\|\bUpsilon_1(0)\|^2_{\mathrm{op}}+\|\bUpsilon_2(0)\|^2_{\mathrm{op}})\{p^4M_{1,U}M_{1,V}M_{1,C,U,V}\\
&\quad + 2p^2\|\bC_{2,p}^{-1}\|^2_{\mathrm{op}}M_{2,V}M_{1,V}M_{1,C,U,V}  \\
&\quad +2p^2\|\bC_{2,p}^{-1}\|^2_{\mathrm{op}}\|^2_{\mathrm{F}}M_{2,U}M_{1,V}M_{1,C,U,V}\}]
\end{align*} 
and 
\begin{align*} 
M_{P,2} &=[2p^2(M_{1,U}M_{1,V}M_{1,C,U,V}+p) \\ 
&\quad + 2p\|\bC_{2,p}^{-1}\|^2_{\mathrm{op}}M_{2,V}M_{1,V}M_{1,C,U,V} +2p\|\bC_{2,p}^{-1}\|^2_{\mathrm{op}}\|^2_{\mathrm{F}}M_{2,U}M_{1,V}M_{1,C,U,V}].
\end{align*}

While bounding around the truth, we let the eigenvalues of $\bOmega_{1}$ be bounded between two fixed constants. Also let $\|\bL_{1,j}\|^2_{\mathrm{op}}$ and $\|\bK_{1,j}\|^2_{\mathrm{op}}$ be also bounded by a fixed constant. This puts an upper bound on $\bC_p^{-1}$ as well. Then we have 
$$\|\bU_{1,j}-\bU_{2,j}\|^2_{\mathrm{F}}\lesssim \|\bOmega_1-\bOmega_2\|^2_{\mathrm{op}}+\sum_{k=1}^{j}  \|\bL_{1,k}\bL_{1,k}^{\mathrm{T}}-\bL_{2,k}\bL_{2,k}^{\mathrm{T}}\|^2_{\mathrm{op}}.$$ 
We can similarly show that 
$$\|\bV_{1,j}-\bV_{2,j}\|^2_{\mathrm{F}}\lesssim \|\bD_{j-1,1}-\bD_{j-1,2}\|^2_{\mathrm{op}}+ \|\bK_{1,j}-\bK_{2,j}\|^2_{\mathrm{op}}.$$ 
The recurrence relations leads to 
\begin{align*} 
\lefteqn{\|\bD_{1,j-1}-\bD_{2,j-1}\|^2_{\mathrm{op}} }\\
&\leq\|\bD_{1,j-2}-\bD_{2,j-2}\|^2_{\mathrm{op}} \\
&\quad + \|\bU_{1,j-1}\|^2_{\mathrm{op}}\|\bU_{2,j-1}\|^2_{\mathrm{op}}\|\bV_{1,j-1}\|^2_{\mathrm{op}}\|\|\bV_{2,j-1}\|^2_{\mathrm{op}}\|\bC_{1,j-2}^{-1}-\bC_{2,j-2}^{-1}\|^2_{\mathrm{op}}\\
&\quad + (\|\bC_{1,p}^{-1}\|^2_{\mathrm{op}}+\|\bC_{2,p}^{-1}\|^2_{\mathrm{op}})(\|\bU_{1,j-1}-\bU_{2,j-1}\|^2_{\mathrm{op}}\|\bV_{1,j-1}\|^2_{\mathrm{op}}\\
& \qquad \qquad \qquad \qquad \qquad \qquad \qquad \qquad 
+\|\bV_{1,j-1}-\bV_{2,j-1}\|^2_{\mathrm{op}}\|\bU_{2,j-1}\|^2_{\mathrm{op}}).
\end{align*} 
Then $\|\bV_{1,k}-\bV_{2,k}\|^2_{\mathrm{F}}\lesssim \epsilon$ and $\|\bU_{1,k}-\bU_{2,k}\|^2_{\mathrm{F}}\lesssim \epsilon$ if 
$$\|\bOmega_1-\bOmega_2\|^2_{\mathrm{op}}\lesssim \ep, \; \|\bK_{1,j}-\bK_{2,j}\|^2_{\mathrm{op}}\lesssim \epsilon, \; \|\bL_{1,k}-\bL_{2,k}\|^2_{\mathrm{op}}\lesssim \epsilon.$$ The implicit constants of proportionality will depend on the bounds for the truth.

\end{proof}

\begin{proof}[Proof of Lemma~\ref{lem:for_sieve}]
    In the sieve, we have 
    $$\|\bU_{1,j}-\bU_{2,j}\|^2_{\mathrm{F}}\leq u_T\{\|\bOmega_1-\bOmega_2\|^2_{\mathrm{op}}+\sum_{k=1}^{j}  \|\bL_{1,k}\bL_{1,k}^{\mathrm{T}}-\bL_{2,k}\bL_{2,k}^{\mathrm{T}}\|^2_{\mathrm{op}}\}$$ 
    for some polynomial $u_T$ in $T$. 
We can similarly show that 
$$\|\bV_{1,j}-\bV_{2,j}\|^2_{\mathrm{F}}\leq v_T\{\|\bD_{j-1,1}-\bD_{j-1,2}\|^2_{\mathrm{op}}+ \|\bK_{1,j}-\bK_{2,j}\|^2_{\mathrm{op}}\}$$ 
for some polynomial $v_T$ in $T$.

Then we have $\|\bV_{1,k}-\bV_{2,k}\|^2_{\mathrm{F}}\leq \epsilon,\|\bU_{1,k}-\bU_{2,k}\|^2_{\mathrm{F}}\leq \epsilon$ if 
\begin{align*} 
\max\{\|\bOmega_1-\bOmega_2\|^2_{\mathrm{op}},\|\bK_{1,j}-\bK_{2,j}\|^2_{\mathrm{op}}\} &\leq \min\big\{\frac{\epsilon}{u_T},\frac{\epsilon}{v'_T}\big\},\\
\|\bL_{1,k}\bL_{1,k}^{\mathrm{T}}-\bL_{2,k}\bL_{2,k}^{\mathrm{T}}\|^2_{\mathrm{op}} &\leq \min\big\{\frac{\epsilon}{pu_T},\frac{\epsilon}{pv'_T}\big\}. 
\end{align*}
The required condition will hold if 
$\|\bL_{1,k}-\bL_{2,k}\|^2_{\mathrm{op}}\leq o_T\epsilon$ for some polynomial $o_T\in\G_T$.

Finally, $\|\bD_{1,0}-\bD_{2,0}\|^2_{\mathrm{op}} \leq c_T\|\bOmega_1-\bOmega_2\|^2_{\mathrm{op}}$ for some $c_T$ which is polynomial in $T$ and the bound for $\|\bU_{1,j-1}-\bU_{2,j-1}\|^2_{\mathrm{op}}$ is established above. 
Thus 
$$\|\bV_{1,j}-\bV_{2,j}\|^2_{\mathrm{F}}\leq v'_T\{\|\bOmega_1-\bOmega_2\|^2_{\mathrm{op}}+ \|\bK_{1,j}-\bK_{2,j}\|^2_{\mathrm{op}}+\sum_{k=1}^{j-1}  \|\bL_{1,k}\bL_{1,k}^{\mathrm{T}}-\bL_{2,k}\bL_{2,k}^{\mathrm{T}}\|^2_{\mathrm{op}}\}$$  
for some $v'_T$ which is polynomial in $T$

If $a_{T,1},a_{T,2},b_{T},c_T$ are all polynomials in $T$ belonging to $\G_T$, we have $M_{P,1}=c_T$ and $M_{P,2}=e_T$ also polynomials in $T$. 

Setting, $\epsilon=\delta_T/[T^2c_T+(T-p+1)a_{T,1}e_T]$, we have $\|\bUpsilon_{1,T}-\bUpsilon_{2,T}\|_{\mathrm{F}}\leq \delta_T$, where 
$$ \max(\|\bV_{1,k}-\bV_{2,k}\|_{\mathrm{F}},\|\bU_{1,k}-\bU_{2,k}\|_{\mathrm{F}}, \|\bL_{1,k}\bL_{1,k}^{\mathrm{T}}-\bL_{2,k}\bL_{2,k}^{\mathrm{T}}\|_{\mathrm{F}}, \|\bUpsilon_{1}(0)-\bUpsilon_{2}(0)\|_{\mathrm{F}})\leq \epsilon.$$
Then, we reach the conclusion using \eqref{eq:mainbd}.
\end{proof}

\begin{proof}[Proof of Lemma~\ref{lem:bigsigma}]
We first compute a bound for VAR(1) case. Then use it for VAR(p).
    From \cite{lutkepohl2005new}, we get $$\bGamma(j)=\sum_{i=0}^{\infty}\bA^{j+i}\bC_{1}(\bA^{i})^{\mathrm{T}}=\bA^{j}\sum_{i=0}^{\infty}\bA^{i}\bC_{1}(\bA^{i})^{\mathrm{T}}$$ 
    for a VAR(1) process, where $\bA=\bGamma(1)\bOmega$. 


Using the telescopic sum setting $\bGamma(0)-\bA\bGamma(0)\bA=\bC_1$, we have $$\sum_{i=0}^{\infty}\bA^{i}\bC_{1}(\bA^{i})^{\mathrm{T}}=\bGamma(0)-\lim_{k\rightarrow\infty} \bA^k\bGamma(0)(\bA^k)^{\mathrm{T}}= \bGamma(0)$$ under our parametrization. Hence, $\bGamma(j)=\bA^{j}\bGamma(0)$. 

As $\bGamma(1)=\bU_1\bV_1^{\mathrm{T}}$ and $\bC_1=\bOmega^{-1}-\bU_1\bU_1^{\mathrm{T}}$, we have $\bV_1\bOmega\bV_1^{\mathrm{T}}=\bI$, and hence $\bOmega^{-1}\geq \bU_1\bU_1^{\mathrm{T}}$. Alternatively, $\lambda_{\min}(\bOmega^{-1}-\bU_1\bU_1^{\mathrm{T}})\geq 0$.

Since, $\bV_1^{\mathrm{T}}\bOmega\bV_1=\bI$, we have $\|\bV_1^{\mathrm{T}}\bOmega^{1/2}\|^2_{\mathrm{op}}= 1$. Hence $\|\bV_1\|^2_{\mathrm{op}}\leq \|\bOmega^{-1/2}\|^2_{\mathrm{op}}$. 

Note that for any two matrices, 
$\bA^k - \bB^k = \bA^{k-1}(\bA-\bB)+(\bA^{k-1}-\bB^{k-1})\bB$. Recursively using this relation and applying norm, we obtain 
$$\|\bA^k - \bB^k\|^2_{\mathrm{F}}\leq \|\bA-\bB\|^2_{\mathrm{F}}(\sum_{i=1}^{k}\|\bA\|_{\mathrm{op}}^{2(k-i)}\|\bB\|_{\mathrm{op}}^{2(i-1)}),$$ 
using the triangle inequality and the submultiplicative property of operator norm.

After some simplification using GP-series sum results, 
\begin{align*}
\sum_{i=1}^{k}\|\bA\|_{\mathrm{op}}^{2(k-i)}\|\bB\|_{\mathrm{op}}^{2(i-1)} &=\frac{\|\bA\|_{\mathrm{op}}^{2k}-\|\bB\|_{\mathrm{op}}^{2k}}{\|\bA\|^2_{\mathrm{op}}-\|\bB\|^2_{\mathrm{op}}}\\
&=\frac{\{\|\bB\|^2_{\mathrm{op}}+(\|\bA\|^2_{\mathrm{op}}-\|\bB\|^2_{\mathrm{op}})\}^{k}-\|\bB\|_{\mathrm{op}}^{2k}}{\|\bA\|^2_{\mathrm{op}}-\|\bB\|^2_{\mathrm{op}}}\\
&=\sum_{i=0}^{k-1} {k\choose i} \|\bB\|^{i}_{\mathrm{op}}(\|\bA\|^2_{\mathrm{op}}-\|\bB\|^2_{\mathrm{op}})^{k-1-i}.
\end{align*}
Hence, 
\begin{align*} 
\|\bA^k - \bB^k\|^2_{\mathrm{F}} &\leq \|\bA-\bB\|^2_{\mathrm{F}}\left\{\sum_{i=0}^{k-1} {k\choose i} \|\bB\|^{i}_{\mathrm{op}}(\|\bA\|^2_{\mathrm{op}}-\|\bB\|^2_{\mathrm{op}})^{k-1-i}\right\} \\
&\leq k\|\bA-\bB\|^2_{\mathrm{F}} \min\{ \|\bA\|_{\mathrm{op}}^{2(k-1)}, |\|\bA\|^2_{\mathrm{op}}-\|\bB\|^2_{\mathrm{op}}|^{k-1} \}.
\end{align*}
Hence 


\begin{align*}
\|\bUpsilon_{1,T}-\bUpsilon_{2,T}\|^2_{\mathrm{F}} &\leq T\|\bUpsilon_1(0)-\bUpsilon_2(0)\|_{\mathrm{F}} + \sum_{k=1}^{T-1}(T-k)\|\bUpsilon_1(k)-\bUpsilon_{2}(k)\|^2_{\mathrm{F}};\\
\|\bUpsilon_1(k)-\bUpsilon_{2}(k)\|^2_{\mathrm{F}} &\leq \|\bUpsilon_1(0)\|^2_{\mathrm{op}}\|\bA_1^k-\bA_2^k\|_{\mathrm{F}} + \|\bA_2\|^{k}_{\mathrm{op}}\|\bUpsilon_1(0)-\bUpsilon_2(0)\|^2_{\mathrm{F}}; \\
\|\bA_1-\bA_2\|_{\mathrm{F}} &\leq \|\bUpsilon_1(1)\|^2_{\mathrm{op}}\|\bOmega_1-\bOmega_2\|^2_{\mathrm{F}}+\|\bOmega_2\|^2_{\mathrm{op}}\|\bUpsilon_1(1)-\bUpsilon_2(1)\|^2_{\mathrm{F}}\\
&\leq \|\bU_1\|^2_{\mathrm{op}}\|\bV_1\|^2_{\mathrm{op}}\|\bOmega_1-\bOmega_2\|^2_{\mathrm{F}}\\
&\quad+\|\bOmega_2\|^2_{\mathrm{op}}\{\|\bU_{1,1}\|^2_{\mathrm{op}}\|\bV_{1,1}-\bV_{1,2}\|^2_{\mathrm{F}}+\|\bV_{1,2}\|^2_{\mathrm{op}}\|\bU_{1,1}-\bU_{1,2}\|^2_{\mathrm{F}}\},
\end{align*}
and 
\begin{align} 
\lefteqn{\|\bUpsilon_{1,T}-\bUpsilon_{2,T}\|^2_{\mathrm{F}}}\nonumber \\ &\leq T\|\bUpsilon_1(0)-\bUpsilon_2(0)\|^2_{\mathrm{F}} \nonumber\\
 &\quad + \sum_{k=1}^{T-1}(T-k)\{\|\bUpsilon_1(0)\|^2_{\mathrm{op}}\|\bA_1^k-\bA_2^k\|^2_{\mathrm{F}} + \|\bA_2\|^{k}_{\mathrm{op}}\|\bUpsilon_1(0)-\bUpsilon_2(0)\|^2_{\mathrm{F}}\} \nonumber\\
 &\leq T\|\bUpsilon_1(0)-\bUpsilon_2(0)\|_{\mathrm{F}} + \sum_{k=1}^{T-1}(T-k)\{\|\bUpsilon_1(0)\|^2_{\mathrm{op}}k\|\bA_1-\bA_2\|^2_{\mathrm{F}} a^{k-1} \nonumber\\
 &\quad+ \|\bA_2\|^{2k}_{\mathrm{op}}\|\bUpsilon_1(0)-\bUpsilon_2(0)\|^2_{\mathrm{F}}\} \nonumber\\
 &\leq (T\frac{b-b^{\mathrm{T}}}{1-b} -\frac{b\{(T-1)b^{\mathrm{T}}-Tb^{T-1}+1\}}{(1-b)^2} + T)\|\bUpsilon_1(0)-\bUpsilon_2(0)\|^2_{\mathrm{F}} \nonumber\\
 &\qquad + T\|\bUpsilon_1(0)\|^2_{\mathrm{op}}\|\bA_1-\bA_2\|^2_{\mathrm{F}}\sum_{k=1}^{T-1} ka^{k-1}\nonumber\\
 &\leq (\frac{Tb-Tb^2+b^{T+1}-b}{(1-b)^2} + T)\|\bUpsilon_1(0)-\bUpsilon_2(0)\|^2_{\mathrm{F}} \nonumber\\
 &\quad+ T\|\bUpsilon_1(0)\|^2_{\mathrm{op}}\|\bA_1-\bA_2\|^2_{\mathrm{F}}\frac{1+(T-2)a^{T-1}-(T-1)a^{T-2}}{(1-a)^2}
 \label{eq:bound}
\end{align}
where $a = \|\bA_1\|^2_{\mathrm{op}}$ and $b=\|\bA_2\|^2_{\mathrm{op}}$. Since, $a<1$, we have $$\frac{1+(T-2)a^{T-1}-(T-1)a^{T-2}}{(1-a)^2}<\frac{1}{(1-a)^2}.$$ This comes in terms of the model parameters.

It is easy to verify that the functions $(1-a)^{-2}\{1+(T-2)a^{T-1}-(T-1)a^{T-2}\}$ and $(1-b)^{-2}\{Tb-Tb^2+b^{T+1}-b\}$ are increasing in $a$ and, $b$ respectively. Following are the limits,
$$
\lim_{a\rightarrow 1} \frac{1+(T-2)a^{T-1}-(T-1)a^{T-2}}{(1-a)^2} = {T-1\choose 2}.
$$

$$
\lim_{b\rightarrow 1} \frac{Tb-Tb^2+b^{T+1}-b}{(1-b)^2} + T = {T \choose 2}.
$$
Due to a VAR(1) representation of VAR$(p)$ \citep{ghosh2019high}, Equation \eqref{eq:bound} and $\bA_{\ell}=\bkappa_{\ell,p}\bUpsilon_{\ell,p-1}^{-1}$ for $\ell=1,2$, we have 
\begin{align*}
    \lefteqn{\|\bUpsilon_{1,T}-\bUpsilon_{2,T}\|^2_{\mathrm{F}}}\\
    &\leq T^2\|\bUpsilon_{1,p-1}-\bUpsilon_{2,p-1}\|^2_{\mathrm{F}}+T^3\|\bUpsilon_{1,p-1}\|^2_{\mathrm{op}}\|\bkappa_{1,p}\bUpsilon_{1,p-1}^{-1}-\bkappa_{2,p}\bUpsilon_{2,p-1}^{-1}\|^2_{\mathrm{F}}\\&\leq T^2\|\bUpsilon_{1,p-1}-\bUpsilon_{2,p-1}\|^2_{\mathrm{F}}+T^3\|\bUpsilon_{1}(0)\|^2_{\mathrm{op}}\|\bkappa_{1,p}\bUpsilon_{1,p-1}^{-1}-\bkappa_{2,p}\bUpsilon_{2,p-1}^{-1}\|^2_{\mathrm{F}},
\end{align*}
completing the proof. 
\end{proof}

\begin{proof}[Proof of Lemma~\ref{lem:rec}]
   From the relation $\bD_j=\bD_{j-1}-\bW_{j}^{\mathrm{T}}\bC_{j-1}^{-1}\bW_{j}$ with $\bW_j=\bU_j\bV_j$.\\ 
\begin{align*} 
\|\bD_{1,j}-\bD_{2,j}\|^2_{\mathrm{F}} &\leq\|\bD_{1,j-1}-\bD_{2,j-1}\|^2_{\mathrm{F}} + \|\bW_{1,j}\|^2_{\mathrm{op}}\|\bW_{2,j}\|^2_{\mathrm{op}}\|\bC_{1,j-1}^{-1}-\bC_{2,j-1}^{-1}\|^2_{\mathrm{F}} \\
&+ (\|\bC_{1,j-1}^{-1}\|^2_{\mathrm{op}}+\|\bC_{2,j-1}^{-1}\|^2_{\mathrm{op}})\|\bW_{1,j}-\bW_{2,j}\|^2_{\mathrm{F}}.
\end{align*}
We have the bounds 
$$\|\bW_{1,j}-\bW_{2,j}\|^2_{\mathrm{F}}\leq \|\bU_{1,j}-\bU_{2,j}\|^2_{\mathrm{F}}\|\bV_{1,j}\|^2_{\mathrm{op}}+\|\bV_{1,j}-\bV_{2,j}\|^2_{\mathrm{F}}\|\bU_{2,j}\|^2_{\mathrm{op}},$$ 
$\|\bC_{1,j-1}^{-1}\|^2_{\mathrm{op}}\leq 
\|\bC_{1,p}^{-1}\|^2_{\mathrm{op}}$, and 
$\|\bC_{2,j-1}^{-1}\|^2_{\mathrm{op}}\leq \|\bC_{2,p}^{-1}\|^2_{\mathrm{op}}$.
Using the quadratic factorization $\bG\bH\bG^T$ for $\bUpsilon^{-1}$ given in the proof of Lemma 2, we have the following,
\begin{align*} 
\|\bUpsilon^{-1}_{1,j}-\bUpsilon^{-1}_{2,j}\|^2_{\mathrm{F}} &\leq \|\bG_1-\bG_2\|_{\mathrm{F}}+(\|\bG_1\|^2_{\mathrm{op}}+\|\bG_2\|^2_{\mathrm{op}})\|\bH_1-\bH_2\|^2_{\mathrm{F}}\\
&\leq \|\bGamma^{-1}_{1,j-1}-\bGamma^{-1}_{2,j-1}\|^2_{\mathrm{F}}+\|\bD^{-1}_{1,j}-\bD^{-1}_{2,j}\|^2_{\mathrm{F}}\\
&\quad +(\|\bC_{1,p}^{-1}\|^2_{\mathrm{op}}+\|\bC_{2,p}^{-1}\|^2_{\mathrm{op}})\|\bUpsilon_{1,j-1}^{-1}\bkappa_{1,j}^{\mathrm{T}}-\bUpsilon_{2,j-1}^{-1}\bkappa_{2,j}^{\mathrm{T}}\|^2_{\mathrm{F}}.
\end{align*}
Now, 
\begin{align*} 
\|\bD^{-1}_{1,j}-\bD^{-1}_{2,j}\|^2_{\mathrm{F}} &\leq \|\bD^{-1}_{1,j}\|^2_{\mathrm{op}}\|\bD^{-1}_{2,j}\|^2_{\mathrm{op}}\|\bD_{1,j}-\bD_{2,j}\|^2_{\mathrm{F}}\\
&\leq \|\bC_{1,p}^{-1}\|^2_{\mathrm{op}}\|\bC_{2,p}^{-1}\|^2_{\mathrm{op}}\|\bD_{1,j}-\bD_{2,j}\|^2_{\mathrm{F}}
\end{align*} 
since $\|\bUpsilon_{1,j-1}^{-1}\|^2_{\mathrm{op}}\leq\|\bUpsilon_{1,p}^{-1}\|^2_{\mathrm{op}}, \|\bUpsilon_{2,j-1}^{-1}\|^2_{\mathrm{op}}\leq\|\bC_{2,p}^{-1}\|^2_{\mathrm{op}}$.
Let 
\begin{align*} 
\bm {P} ={\begin{bmatrix}\bm{A}_1 &\bm{A}_2 \\\bm{A}_3 &\bm{A}_4 \end{bmatrix}}^{-1} ={\begin{bmatrix}\bm{B}_1 &\bm{B}_2 \\\bm{B}_3 &\bm{B}_4, \end{bmatrix}}
\end{align*} 
where 
\begin{align*} 
\bB_1 &=(\bm{A}_1 -\bm{A}_2\bm{A}_4^{-1}\bm{A}_3)^{-1},\\
\bB_2&=-(\bA_1 -\bA_2\bA_4^{-1}\bA_3)^{-1}\bA_2\bA_4^{-1},\\ 
\bB_3&=-\bA_4^{-1}\bA_3(\bA_1 -\bA_2\bA_4^{-1}\bA_3)^{-1},\\ \bB_4&=\bA_4^{-1}+\bA_4^{-1}\bA_3 (\bA_1 -\bA_2\bA_4^{-1}\bA_3)^{-1}\bA_2\bA_4^{-1}.
\end{align*}
Setting $\bP=\bUpsilon_{j-1}, \bA_1=\bGamma(0), \bA_2=\bxi_{j-1}^{\mathrm{T}},\bA_3=\bxi_{j-1},\bA_4=\bUpsilon_{j-2}$, we have 
$$\bUpsilon_{j-1}^{-1}\bkappa_{j}^{\mathrm{T}}=\{\bC_{j-1}^{-1}\bU_j\bV_j, \bUpsilon_{j-2}^{-1}\bkappa_{j-1}^{\mathrm{T}}-\bUpsilon_{j-2}^{-1}\bxi_{j-1}^{\mathrm{T}}\bC_{j-1}^{-1}\bU_j\bV_j\},$$ since 
$\bGamma(j) = \bU_j\bV_j^{\mathrm{T}} + \bxi_{j-1}\bGamma^{-1}_{j-2}\bkappa_{j-1}$. 

Now, $\|\bUpsilon_{1,j-2}^{-1}\bkappa_{1,j-1}^{\mathrm{T}}\|^2_{\mathrm{op}},\|\bUpsilon_{2,j-2}^{-1}\bkappa_{2,j-1}^{\mathrm{T}}\|^2_{\mathrm{op}}\leq 1$, and hence the relation $$\|\bUpsilon_{1,j-1}^{-1}\bxi_{1,j}^{\mathrm{T}}-\bUpsilon_{2,j-1}^{-1}\bxi_{2,j}^{\mathrm{T}}\|^2_{\mathrm{F}}=\|\bUpsilon_{1,j-1}^{-1}\bkappa_{1,j}^{\mathrm{T}}-\bUpsilon_{2,j-1}^{-1}\bkappa_{2,j}^{\mathrm{T}}\|^2_{\mathrm{F}}$$ holds. 
\end{proof}

\begin{proof}[Proof of Lemma~\ref{lem:finalbd}]
We have the relations 
\begin{align*}
    \|\bUpsilon_{1,j}-\bUpsilon_{2,j}\|^2_{\mathrm{F}}\nonumber &\leq \|\bUpsilon_{1,j-1}-\bUpsilon_{2,j-1}\|^2_{\mathrm{F}}+\|\bD_{1,j}-\bD_{2,j}\|^2_{\mathrm{F}}\nonumber \\
    &\quad  +(\|\bUpsilon_1(0)\|^2_{\mathrm{op}}+\|\bUpsilon_2(0)\|^2_{\mathrm{op}})\|\bUpsilon_{1,j-1}^{-1}\bkappa_{1,j}^{\mathrm{T}}-\bUpsilon_{2,j-1}^{-1}\bkappa_{2,j}^{\mathrm{T}}\|^2_{\mathrm{F}},\\
    \|\bD_{1,j}-\bD_{2,j}\|^2_{\mathrm{F}}
    &\leq\|\bD_{1,j-1}-\bD_{2,j-1}\|^2_{\mathrm{F}} \nonumber \\
& \quad    + \|\bU_{1,j}\bV_{1,j}\|^2_{\mathrm{op}}\|\bU_{2,j}\bV_{2,j}\|^2_{\mathrm{op}}\|\bC_{1,j-1}^{-1}-\bC_{2,j-1}^{-1}\|^2_{\mathrm{F}} 
\nonumber \\
& \quad 
+ (\|\bC_{1,j-1}^{-1}\|^2_{\mathrm{op}}+\|\bC_{2,j-1}^{-1}\|^2_{\mathrm{op}})\\
&\qquad \times (\|\bU_{1,j}-\bU_{2,j}\|^2_{\mathrm{F}}\|\bV_{1,j}\|^2_{\mathrm{op}}+\|\bV_{1,j}-\bV_{2,j}\|^2_{\mathrm{F}}\|\bU_{2,j}\|^2_{\mathrm{op}})\nonumber\\
    &\le \|\bUpsilon_1(0)-\bUpsilon_2(0)\|^2_{\mathrm{F}}\\
    &\quad+ \sum_{k=1}^{j}\big[\|\bU_{1,k}\bV_{1,k}\|^2_{\mathrm{op}}\|\bU_{2,k}\bV_{2,k}\|^2_{\mathrm{op}}\|\bC_{1,k-1}^{-1}-\bC_{2,k-1}^{-1}\|^2_{\mathrm{F}} \\
    &\quad + (\|\bC_{1,k-1}^{-1}\|^2_{\mathrm{op}}+\|\bC_{2,k-1}^{-1}\|^2_{\mathrm{op}})\\
    &\qquad \times (\|\bU_{1,k}-\bU_{2,k}\|^2_{\mathrm{F}}\|\bV_{1,k}\|^2_{\mathrm{op}}+\|\bV_{1,k}-\bV_{2,k}\|^2_{\mathrm{F}}\|\bU_{2,k}\|^2_{\mathrm{op}})\big],
    \end{align*} 
    \begin{align*} 
\|\bUpsilon_{1,j-1}^{-1}\bkappa_{1,j}^{\mathrm{T}}-\bUpsilon_{2,j-1}^{-1}\bkappa_{2,j}^{\mathrm{T}}\|^2_{\mathrm{F}} & \leq \bB_1(1+\|\bC_{1,p}^{-1}\|^2_{\mathrm{op}}\max_j\|\bU_{1,j}\|^2_{\mathrm{op}}\max_j\|\bV_{1,j}\|^2_{\mathrm{op}})^j+\cdots 
\end{align*} 
and hence 
\begin{align*}
    \sum_{j=0}^{p-1}\|\bD_{1,j}-\bD_{2,j}\|^2_{\mathrm{F}} 
    &\leq p\|\bUpsilon_1(0)-\bUpsilon_2(0)\|^2_{\mathrm{F}}\\
    &\quad + \sum_{j=0}^{p-1}\sum_{k=1}^{j}\big[\|\bU_{1,k}\bV_{1,k}\|^2_{\mathrm{op}}\|\bU_{2,k}\bV_{2,k}\|^2_{\mathrm{op}}\|\bC_{1,k-1}^{-1}-\bC_{2,k-1}^{-1}\|^2_{\mathrm{F}} \\ 
    &\qquad + (\|\bC_{1,k-1}^{-1}\|^2_{\mathrm{op}}+\|\bC_{2,k-1}^{-1}\|^2_{\mathrm{op}})(\|\bU_{1,k}-\bU_{2,k}\|^2_{\mathrm{F}}\|\bV_{1,k}\|^2_{\mathrm{op}}\\
    &\qquad +\|\bV_{1,k}-\bV_{2,k}\|^2_{\mathrm{F}}\|\bU_{2,k}\|^2_{\mathrm{op}})\big]\nonumber\\
    &\leq p\|\bUpsilon_1(0)-\bUpsilon_2(0)\|^2_{\mathrm{F}} \\
    & \quad + \sum_{j=1}^{p-1}\sum_{k=1}^{j}\big[M_{1,U}M_{1,V}M_{2,U}M_{2,V}\|\bC_{1,k-1}^{-1}-\bC_{2,k-1}^{-1}\|^2_{\mathrm{F}} \\
    &\qquad + (\|\bC_{1,p}^{-1}\|^2_{\mathrm{op}}+\|\bC_{2,p}^{-1}\|^2_{\mathrm{op}})\\
    &\qquad \qquad \times (\|\bU_{1,k}-\bU_{2,k}\|^2_{\mathrm{F}}M_{1,V}+\|\bV_{1,k}-\bV_{2,k}\|^2_{\mathrm{F}}M_{2,U})\big]\nonumber\\
     &\leq p\|\bUpsilon_1(0)-\bUpsilon_2(0)\|^2_{\mathrm{F}}\\
     &\quad+ p\sum_{k=1}^{p}\big[M_{1,U}M_{1,V}M_{2,U}M_{2,V}\|\bC_{1,k-1}^{-1}-\bC_{2,k-1}^{-1}\|^2_{\mathrm{F}}\\
     &\qquad+ (\|\bC_{1,p}^{-1}\|^2_{\mathrm{op}}+\|\bC_{2,p}^{-1}\|^2_{\mathrm{op}})\\
     &\qquad \qquad \times (\|\bU_{1,k}-\bU_{2,k}\|^2_{\mathrm{F}}M_{1,V}+\|\bV_{1,k}-\bV_{2,k}\|^2_{\mathrm{F}}M_{2,U})\big],
\end{align*}
completing the proof. 
\end{proof}

\begin{proof}[Proof of Lemma~\ref{lem:finalbdpre}]
We have 
\begin{align*}
    \lefteqn{\|\bUpsilon_{1,j-1}^{-1}\bkappa_{1,j}^{\mathrm{T}}-\bUpsilon_{2,j-1}^{-1}\bkappa_{2,j}^{\mathrm{T}}\|^2_{\mathrm{F}}}\nonumber\\
    &\leq\|\bUpsilon_{1,j-2}^{-1}\bkappa_{1,j-1}^{\mathrm{T}}-\bUpsilon_{2,j-2}^{-1}\bkappa_{2,j-1}^{\mathrm{T}}\|^2_{\mathrm{F}}(1+\|\bC_{1,j}^{-1}\|^2_{\mathrm{op}}\|\bU_{1,j}\|^2_{\mathrm{op}}\|\bV_{1,j}\|^2_{\mathrm{op}})\\
    &\quad+2\|\bC_{1,j}^{-1}\bU_{1,j}\bV_{1,j}-\bC_{2,j}^{-1}\bU_{2,j}\bV_{2,j}\|^2_{\mathrm{F}}\\
    &\leq \|\bUpsilon_{1,j-2}^{-1}\bkappa_{1,j-1}^{\mathrm{T}}-\bUpsilon_{2,j-2}^{-1}\bkappa_{2,j-1}^{\mathrm{T}}\|^2_{\mathrm{F}}(1+\|\bC_{1,j}^{-1}\|^2_{\mathrm{op}}\|\bU_{1,j}\|^2_{\mathrm{op}}\|\bV_{1,j}\|^2_{\mathrm{op}})\\
    &\quad+2\|\bC^{-1}_{1,j}-\bC^{-1}_{2,j}\|^2_{\mathrm{F}}\|\bU_{1,j}\|^2_{\mathrm{op}}\|\bV_{1,j}\|^2_{\mathrm{op}}\\
    &\quad+
    2\|\bC_{2,j}^{-1}\|^2_{\mathrm{op}}\bigg[\|\bU_{1,j}-\bU_{2,j}\|^2_{\mathrm{F}}\|\bV_{2,j}\|^2_{\mathrm{op}}+\|\bV_{1,j}-\bV_{2,j}\|^2_{\mathrm{F}}\|\bU_{1,j}\|^2_{\mathrm{op}}\bigg]\\
    &\leq \|\bUpsilon_{1,j-2}^{-1}\bkappa_{1,j-1}^{\mathrm{T}}-\bUpsilon_{2,j-2}^{-1}\bkappa_{2,j-1}^{\mathrm{T}}\|^2_{\mathrm{F}}(1+\|\bC_{1,p}^{-1}\|^2_{\mathrm{op}}\|M_{1,U}M_{1,V})\\
    &\quad+2\|\bC^{-1}_{1,j}-\bC^{-1}_{2,j}\|^2_{\mathrm{F}}M_{1,U}M_{1,V}\\&\qquad+
    2\|\bC_{2,p}^{-1}\|^2_{\mathrm{op}}\big[\|\bU_{1,j}-\bU_{2,j}\|^2_{\mathrm{F}}M_{2,V}+\|\bV_{1,j}-\bV_{2,j}\|^2_{\mathrm{F}}M_{1,U}\big],
\end{align*}
where $M_{\ell,U}=\max_j\|\bU_{\ell,j}\|^2_{\mathrm{op}}$, and $M_{\ell,V}=\max_j\|\bV_{\ell,j}\|^2_{\mathrm{op}}$. 

Let $M_{\ell,C,U,V}=(1+\|\bC_{\ell,p}^{-1}\|^2_{\mathrm{op}}\|M_{\ell,U}M_{\ell,V})^p$, then 
\begin{align*}
    \|\bUpsilon_{1,p-1}^{-1}\bkappa_{1,p}^{\mathrm{T}}-\bUpsilon_{2,p-1}^{-1}\bkappa_{2,p}^{\mathrm{T}}\|^2_{\mathrm{F}} &\leq 2M_{1,U}M_{1,V}M_{1,C,U,V}\sum_{k=0}^p\|\bC^{-1}_{1,k}-\bC^{-1}_{2,k}\|^2_{\mathrm{F}} \\
    &\quad + 2\|\bC_{2,p}^{-1}\|^2_{\mathrm{op}}M_{2,V}M_{1,V}M_{1,C,U,V}\sum_{k=1}^p\|\bU_{1,k}-\bU_{2,k}\|^2_{\mathrm{F}}\nonumber\\
    &\quad+2\|\bC_{2,p}^{-1}\|^2_{\mathrm{op}}\|^2_{\mathrm{F}}M_{2,U}M_{1,V}M_{1,C,U,V}\sum_{k=1}^p\|\bV_{1,k}-\bV_{2,k}\|^2_{\mathrm{F}}.
\end{align*}
Combining \eqref{eq:ineq2} and \eqref{eq:ineq3}, we get an upper bound for \eqref{eq:ineq1}.
\end{proof}

\begin{proof}[Proof of Lemma~\ref{lem:UjVj}]
   We have
$$\bU_j = \bC_{j-1}\bL_j(\bI + \bL^{\mathrm{T}}_j \bC_{j-1}\bL_j)^{-1/2}, \quad  
\bV_j=\bK_j(\bK_j^{\mathrm{T}}\bD_{j-1}^{-1}\bK_j)^{-1/2},$$
and the bounds 
$\|\bU_{\ell,j}\|^2_{\mathrm{op}}\leq \|\bC_{j-1}\|^2_{\mathrm{op}}\|\bL_j\|^2_{\mathrm{op}}$ and $\|\bV_{\ell,j}\|^2_{\mathrm{op}}\leq \|\bD_{j-1}^{-1}\|^2_{\mathrm{op}}$. 

Using the inequality $\|\bA^{1/2}-\bB^{1/2}\|^2_{\mathrm{op}}\leq\|\bA-\bB\|^{1/2}_{\mathrm{op}}$, we can bound 
\begin{align*}
   \|\bV_{j,1}-\bV_{j,2}\|^2_{\mathrm{F}} &\leq r\|\bV_{j,1}-\bV_{j,2}\|^2_{\mathrm{op}}\\
   &\leq \|\bK_{1,1}\|^2_{\mathrm{op}}\|(\bK_{j,1}^{\mathrm{T}}\bD_{j-1,1}^{-1}\bK_{1,1})^{-1}-(\bK_{j,2}^{\mathrm{T}}\bD_{j-2,2}^{-1}\bK_{1,2})^{-1}\|_{\mathrm{op}}\\
   &\quad + r\|(\bK_{j,2}^{\mathrm{T}}\bD_{0,2}^{-1}\bK_{j,2})^{-1/2}\|^2_{\mathrm{op}}\|\bK_{j,1}-\bK_{j,2}\|^2_{\mathrm{op}} .
\end{align*}
Similarly, 
\begin{align*}
    \|\bU_{j,1}-\bU_{j,2}\|^2_{\mathrm{op}}&\leq \|\bC_{j-1,1}\|^2_{\mathrm{op}}\|\bL_{j,1}\|^2_{\mathrm{op}}\|(\bI + \bL^{\mathrm{T}}_{j,1} \bC_{j-1,1}\bL_{1,1})^{-1}\|_{\mathrm{op}}\\
    &\qquad \times \|(\bI + \bL^{\mathrm{T}}_{j,2} \bC_{j-1,2}\bL_{j,2})^{-1}\|_{\mathrm{op}}\|\bL^{\mathrm{T}}_{j,1} \bC_{j-1,1}\bL_{j,1}
    -\bL^{\mathrm{T}}_{2,j} \bC_{j-1,2}\bL_{j,2}\|_{\mathrm{op}}\\
    &\quad +\|(\bI + \bL^{\mathrm{T}}_{j,2} \bC_{j-1,2}bL_{j,2})^{-1/2}\|^2_{\mathrm{op}}\|\bC_{j-1,1}\bL_{j,1}-\bC_{j-1,2}\bL_{j,2}\|^2_{\mathrm{op}},
\end{align*}
establing the assertion.
\end{proof}

\newpage 
\section{Additional Figures}

We present some additional figures from the data analysis. 

\begin{figure}[htbp]
\centering
\subfigure{\includegraphics[width = 1\textwidth, trim=1cm 1cm 1cm 0cm, clip=true]{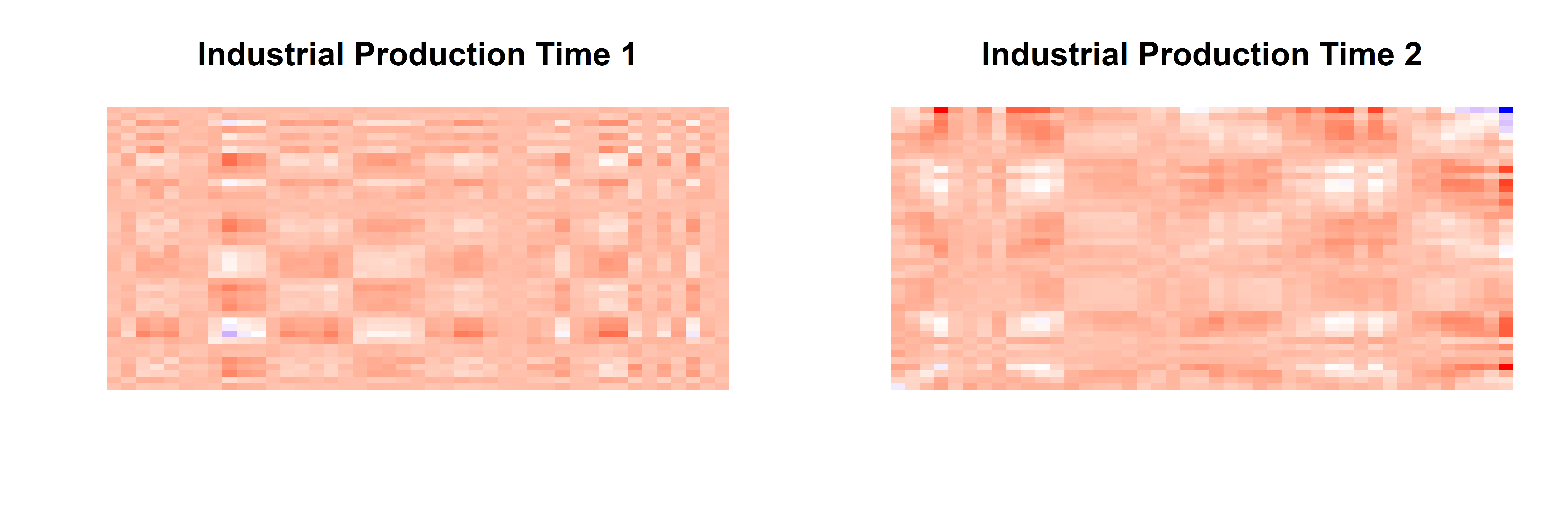}}
\\
\subfigure{\includegraphics[width = 1\textwidth, trim=1cm 1cm 1cm 0cm, clip=true]{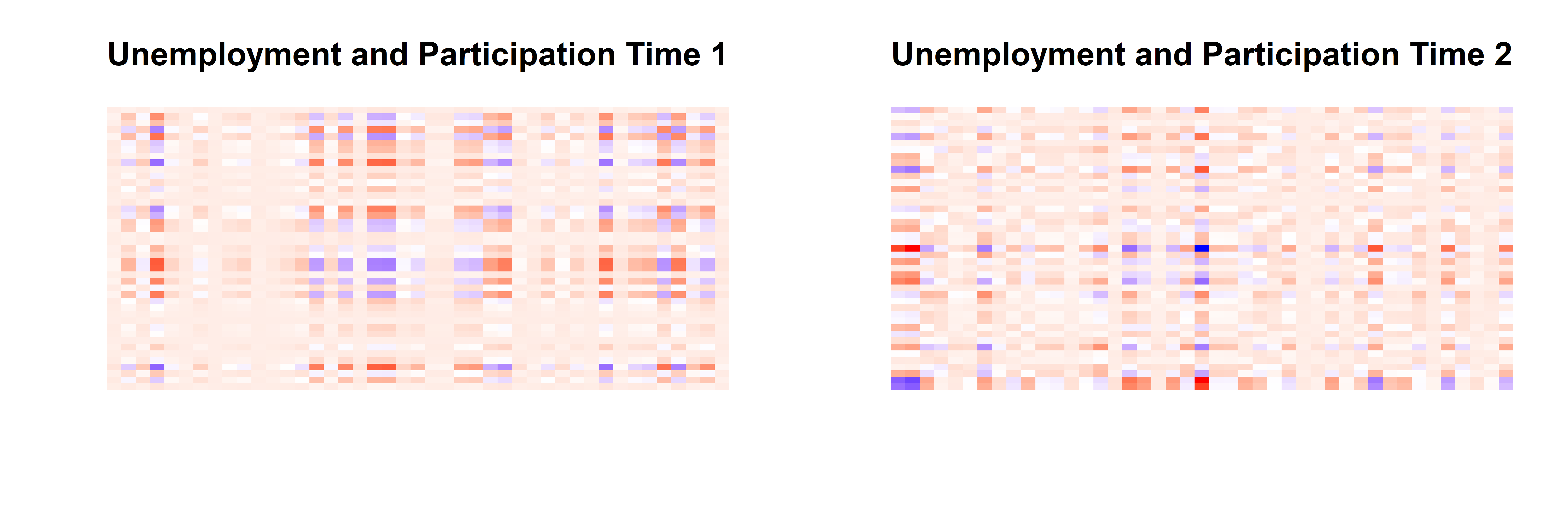}}
\end{figure}

\begin{figure}[htbp]
\subfigure{\includegraphics[width = 1\textwidth, trim=1cm 1cm 1cm 0cm, clip=true]{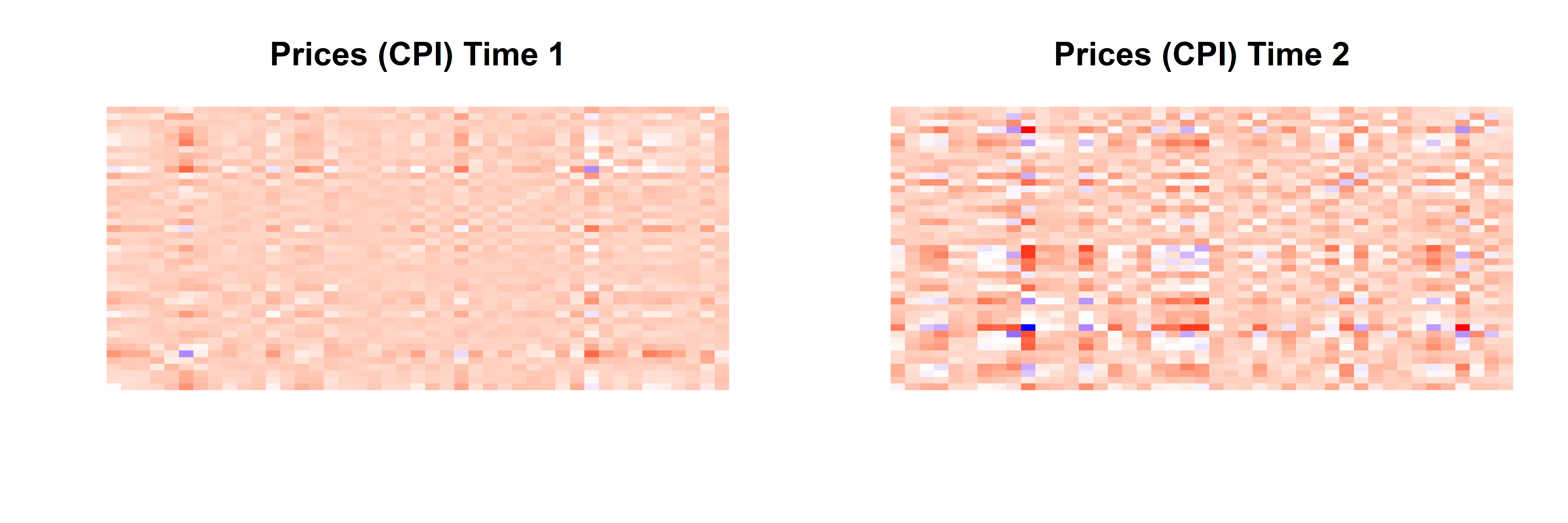}}
\\
\subfigure{\includegraphics[width = 1\textwidth, trim=1cm 1cm 1cm 0cm, clip=true]{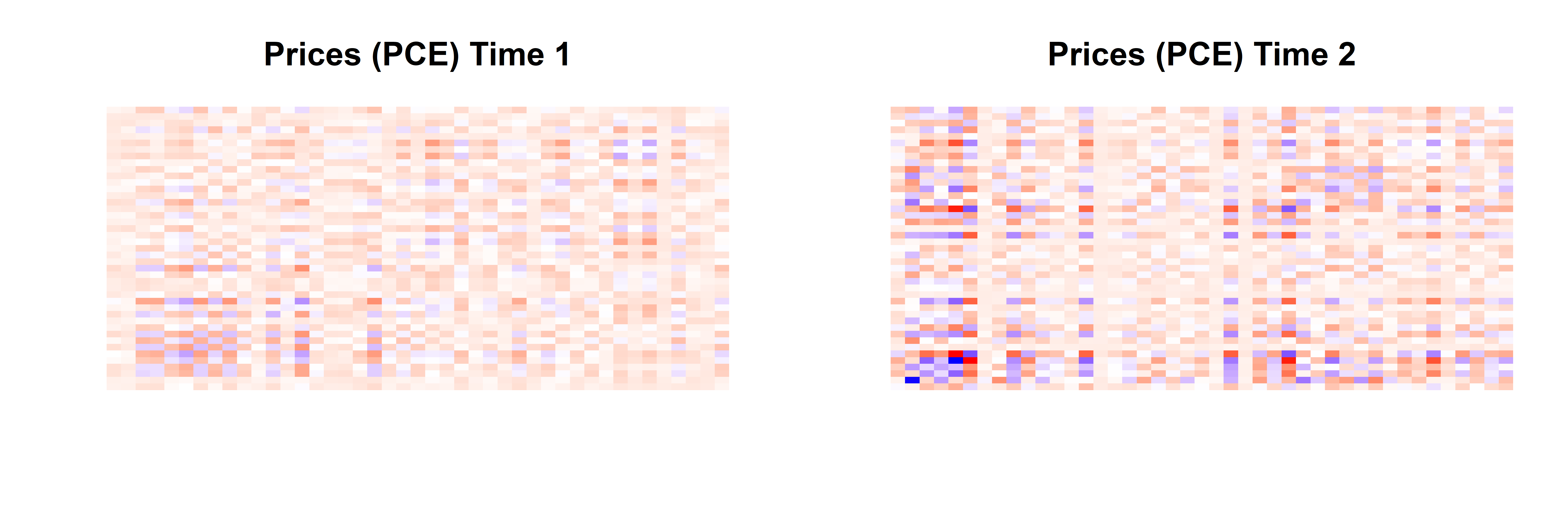}}
\\
\subfigure{\includegraphics[width = 1\textwidth, trim=1cm 1cm 1cm 0cm, clip=true]{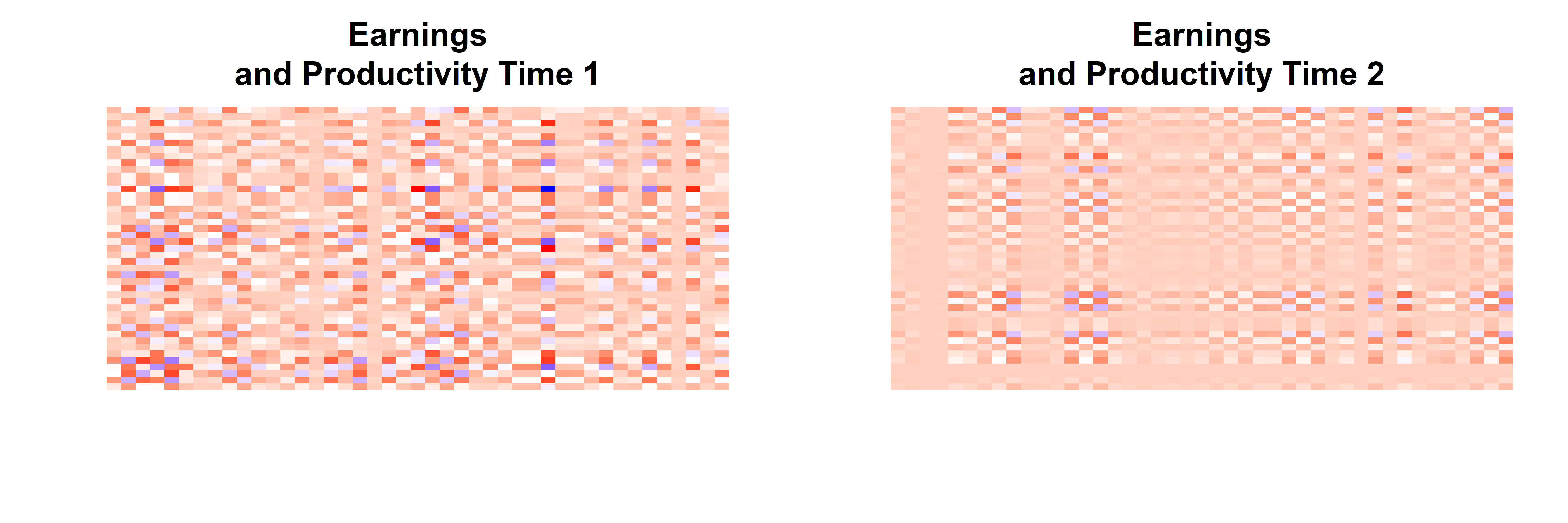}}
\label{fig: precision2}
\end{figure}

\begin{figure}[htbp]
\centering
\subfigure{\includegraphics[width = 1\textwidth, trim=1cm 1cm 1cm 0cm, clip=true]{Figures/EDAprecisions/9EDAprecision.png}}
\\
\subfigure{\includegraphics[width = 1\textwidth, trim=1cm 1cm 1cm 0cm, clip=true]{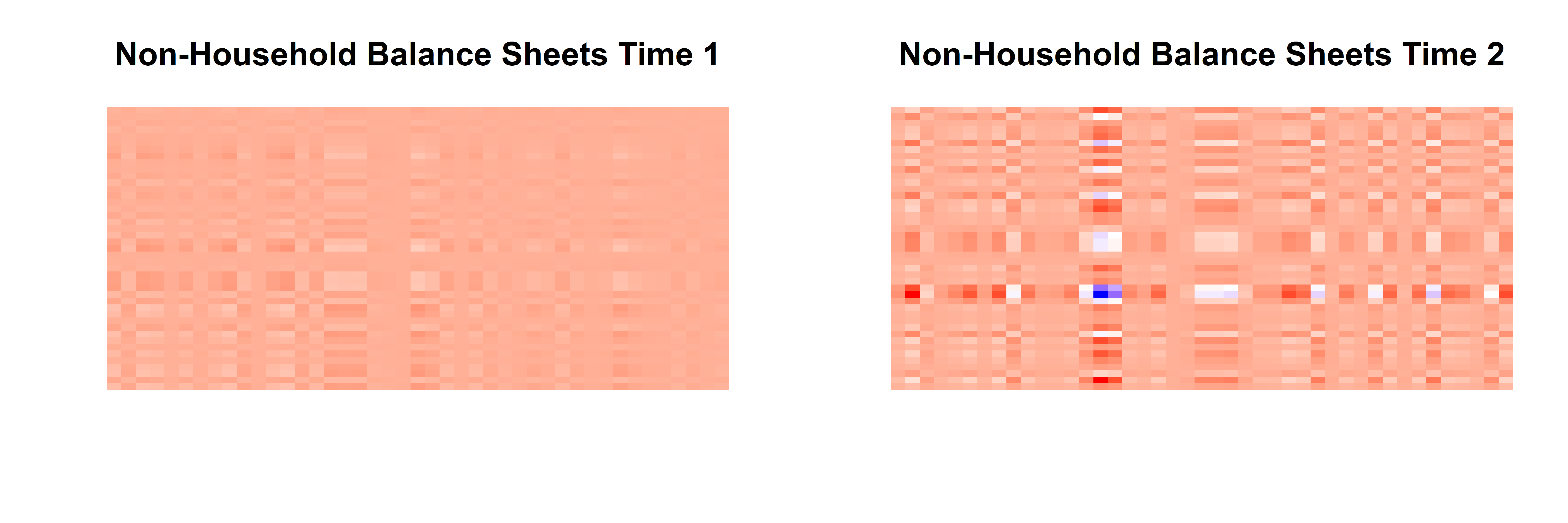}}
\label{fig: precision4}
\end{figure}

\begin{figure}[htbp!]
\centering
\subfigure{\includegraphics[width = 0.45\textwidth, trim=1cm 2cm 0cm 0cm, clip=true]{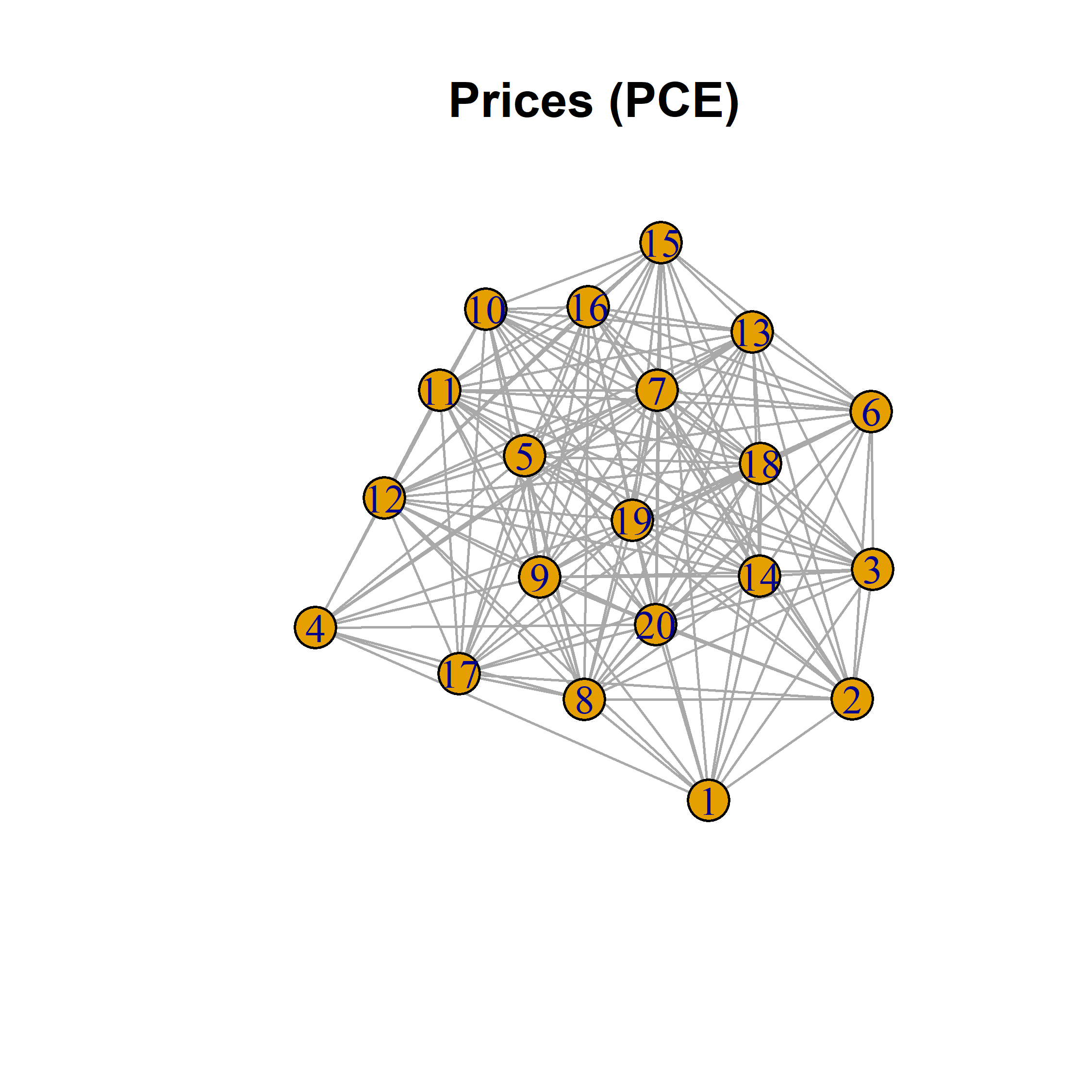}}
\subfigure{\includegraphics[width = 0.45\textwidth, trim=1cm 2cm 0cm 0cm, clip=true]{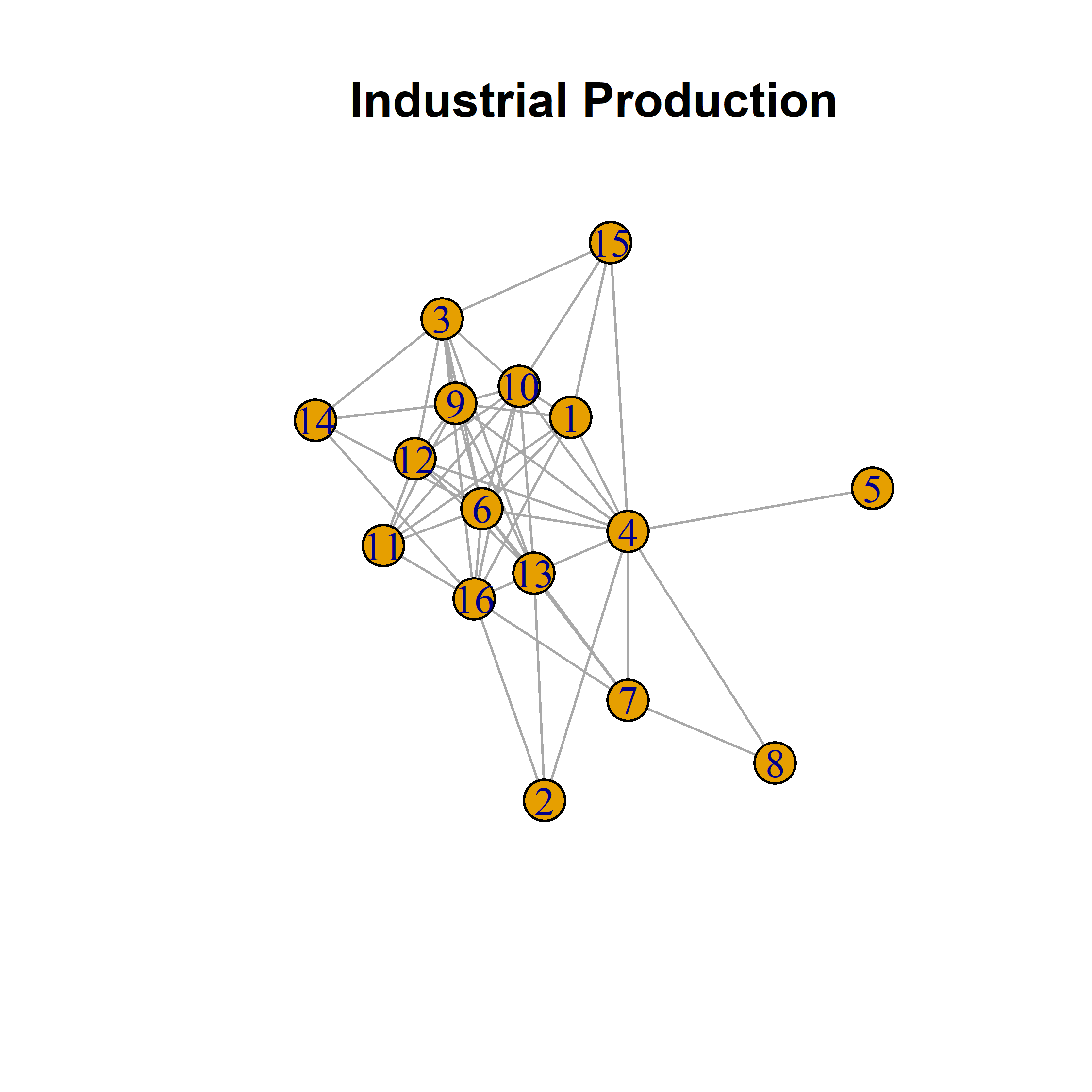}}
\subfigure{\includegraphics[width = 0.45\textwidth, trim=1cm 2cm 0cm 0cm, clip=true]{Figures/changegraphs/7changegraph.png}}
\subfigure{\includegraphics[width = 0.45\textwidth, trim=1cm 2cm 0cm 0cm, clip=true]{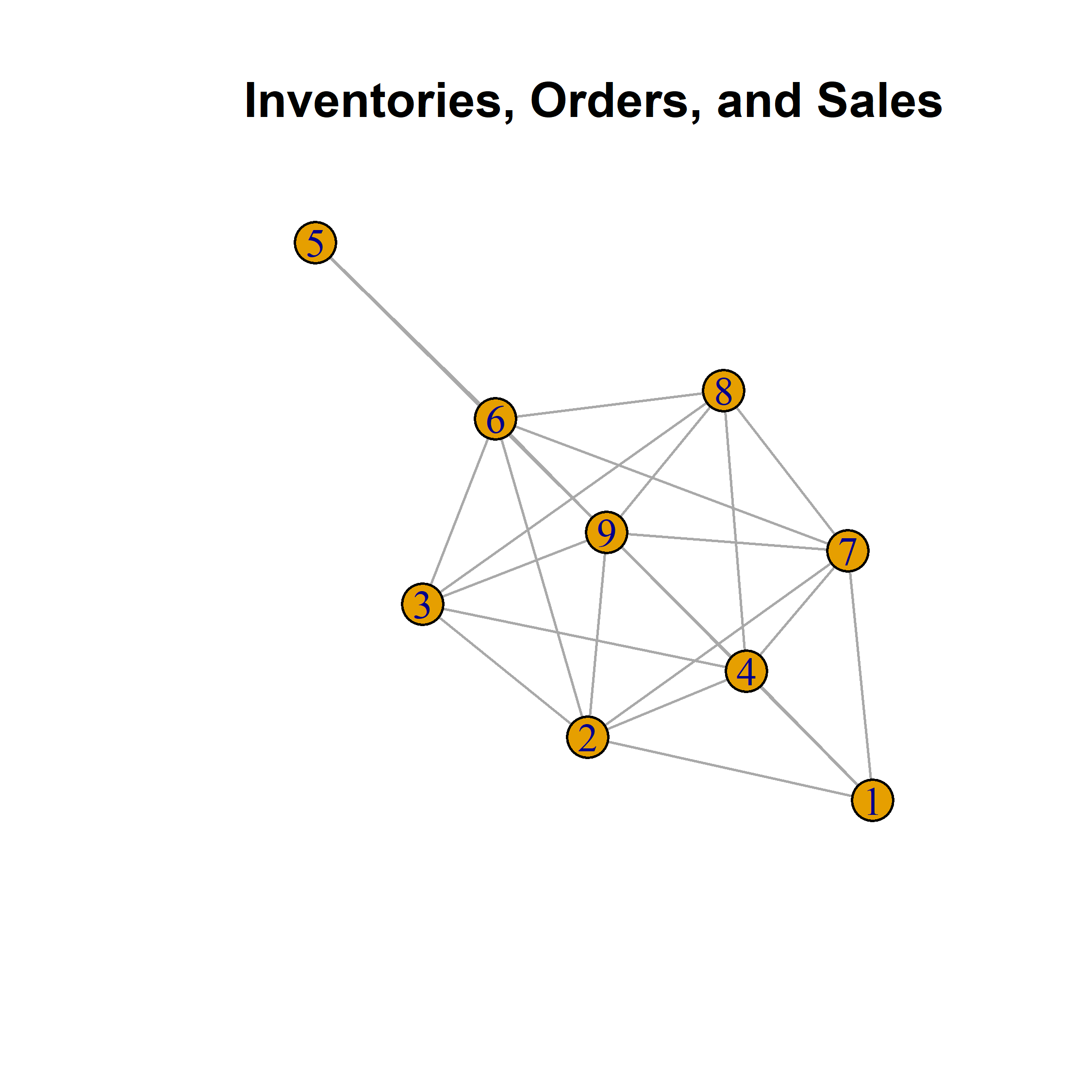}}
\subfigure{\includegraphics[width = 0.45\textwidth, trim=1cm 2cm 0cm 0cm, clip=true]{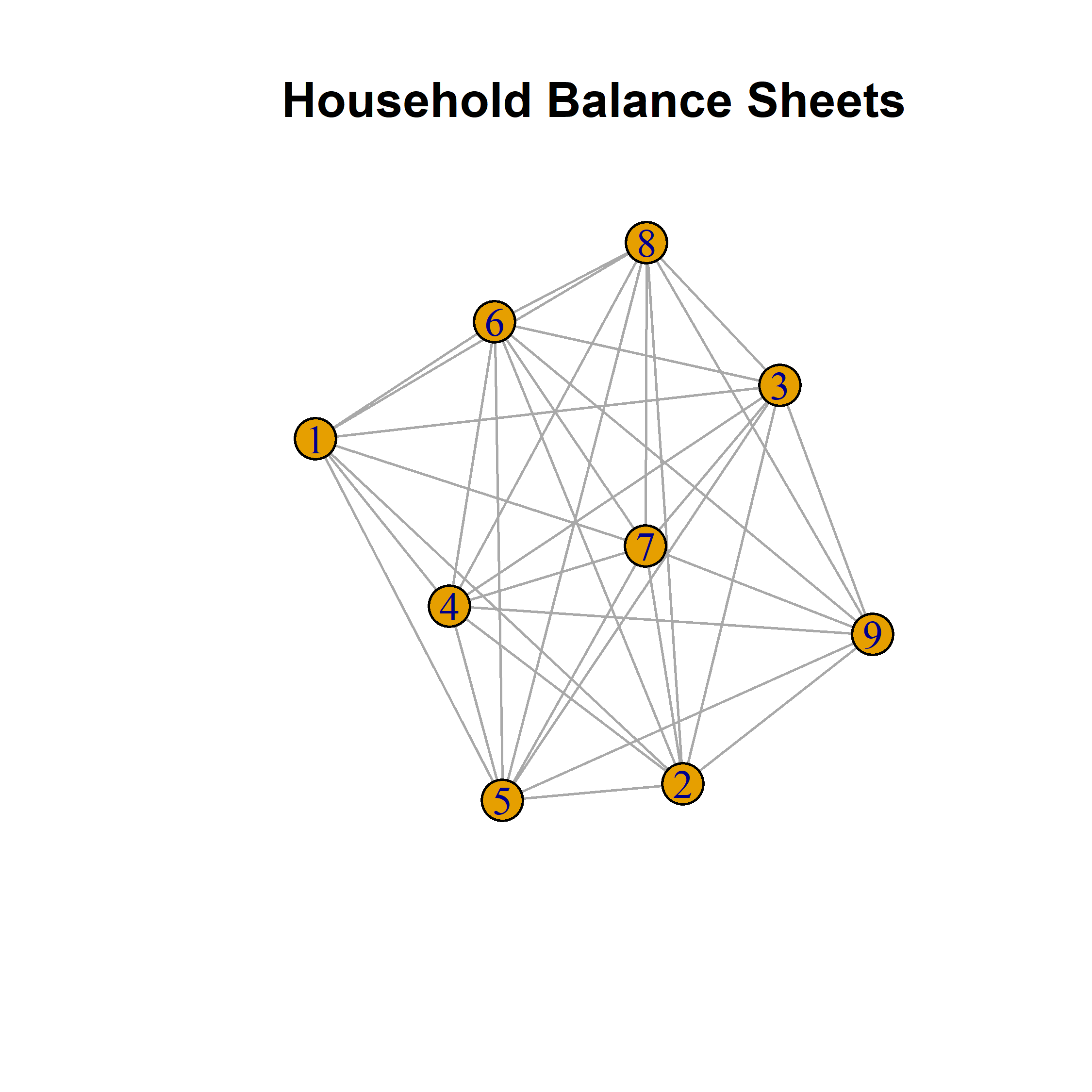}}
\subfigure{\includegraphics[width = 0.45\textwidth, trim=1cm 2cm 0cm 0cm, clip=true]{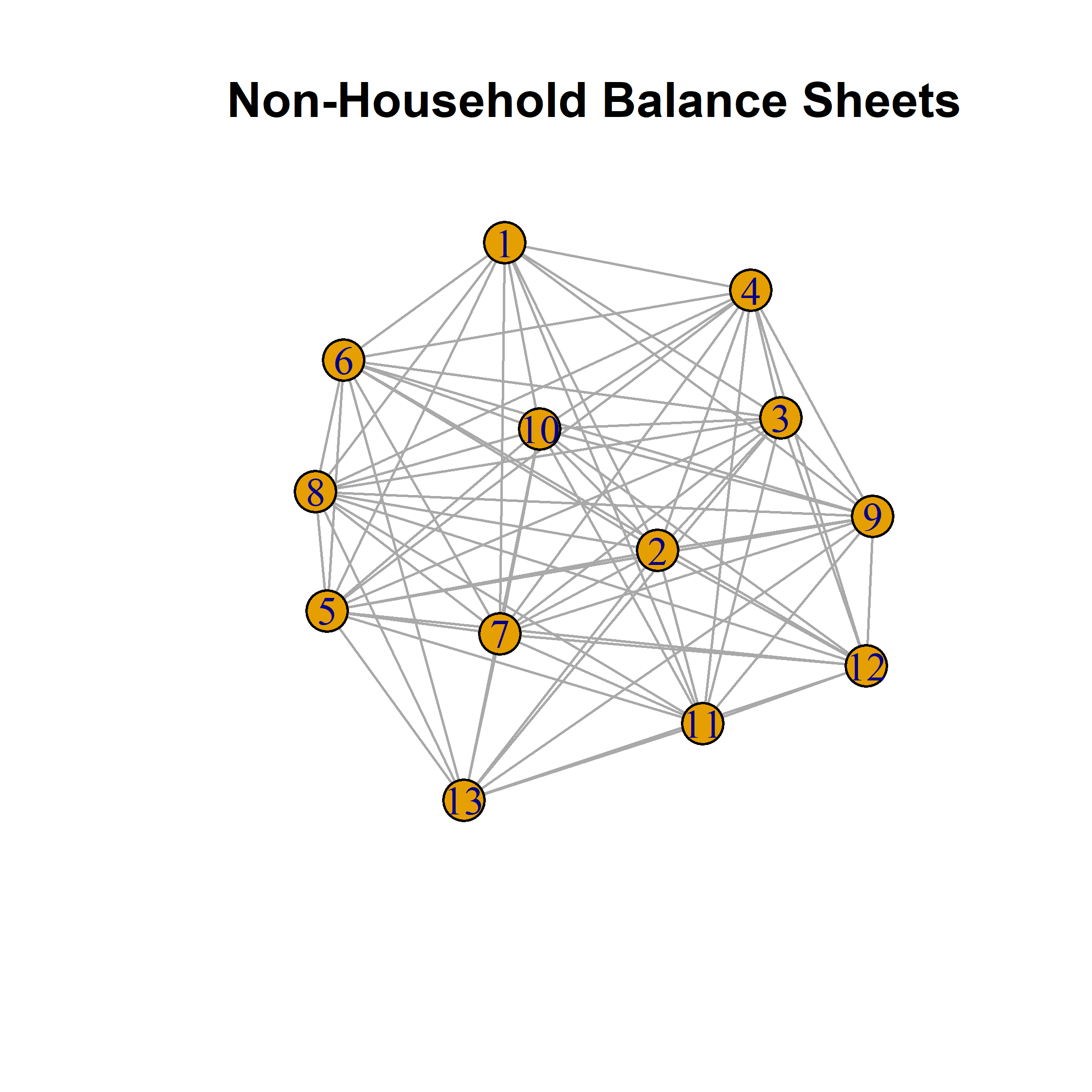}}
\caption{Estimated change networks. The node numbers correspond to the variable names as described in \cite{mccracken2020fred}.}
\label{fig: changegraph}
\end{figure}

\bibliographystyle{plainnat}
	\bibliography{main}

\end{document}